\documentclass[letterpaper,onecolumn]{IEEEtran}
\usepackage{amsmath,amsfonts,amsthm,amssymb}
\usepackage{array}
\usepackage{textcomp}
\usepackage{stfloats}
\usepackage{url}
\usepackage{verbatim}
\usepackage{cite}
\usepackage{booktabs,multirow, multicol}       % professional-quality tables
\usepackage{amsfonts}       % blackboard math symbols
\usepackage{nicefrac}       % compact symbols for 1/2, etc.
\usepackage[nopatch=footnote]{microtype}     % microtypography
\usepackage{fancyhdr}       % header
\usepackage{graphicx}       % graphics
\usepackage{enumerate,enumitem}
\usepackage{mathrsfs}
\usepackage{bbm}
\usepackage{float}
\usepackage{xspace}
\usepackage{comment}
\usepackage{adjustbox}
\usepackage{diagbox}
\usepackage{boxedminipage}
\usepackage[dvipsnames]{xcolor}
\usepackage{authblk}
\usepackage{wrapfig}
\usepackage{needspace}
\usepackage{hyperref}
\usepackage{algorithm,algpseudocode}

\usepackage{adjustbox}
\usepackage{setspace}
\usepackage{amsmath,amsfonts,bm,mathtools}

\providecommand{\matrixnorm}[1]{\left|\!\left|\!\left|{#1}
  \right|\!\right|\!\right|} % Full-sized matrix norm

\def\mG{{\bm{G}}}

\def\mI{{\bm{I}}}

\DeclareMathAlphabet{\mathsfit}{\encodingdefault}{\sfdefault}{m}{sl}
\SetMathAlphabet{\mathsfit}{bold}{\encodingdefault}{\sfdefault}{bx}{n}

\def\gB{{\mathcal{B}}}
\def\gC{{\mathcal{C}}}

\def\gE{{\mathcal{E}}}
\def\gF{{\mathcal{F}}}

\def\gI{{\mathcal{I}}}

\def\gL{{\mathcal{L}}}

\def\gN{{\mathcal{N}}}

\def\gR{{\mathcal{R}}}

\def\gX{{\mathcal{X}}}

\def\sC{{\mathbb{C}}}

\def\sE{{\mathbb{E}}}

\def\sG{{\mathbb{G}}}

\def\sO{{\mathbb{O}}}
\def\sP{{\mathbb{P}}}

\def\sR{{\mathbb{R}}}
\def\sS{{\mathbb{S}}}

\newcommand{\KL}{D_{\mathrm{KL}}}

\DeclareMathOperator*{\argmax}{arg\,max}
\DeclareMathOperator*{\argmin}{arg\,min}

\DeclareMathOperator*{\argsup}{arg\,sup}

\newtheorem{theorem}{Theorem}%  meant for continuous numbers
\newtheorem{lemma}[theorem]{Lemma}

\newtheorem{corollary}[theorem]{Corollary}
\newtheorem{example}{Example}%
\newtheorem{remark}{Remark}%

\newtheorem{definition}{Definition}
\newtheorem{assumption}{Assumption}

\newcommand{\bea}{\begin{eqnarray*}}
\newcommand{\eea}{\end{eqnarray*}}
\newcommand{\ba}{\begin{eqnarray*}}
\newcommand{\ea}{\end{eqnarray*}}
\newcommand{\be}{\begin{equation}}
\newcommand{\ee}{\end{equation}}
\newcommand{\bi}{\begin{itemize}}
\newcommand{\ei}{\end{itemize}}

\makeatletter
\newcommand*\rel@kern[1]{\kern#1\dimexpr\macc@kerna}
\newcommand*\widebar[1]{%
  \begingroup
  \def\mathaccent##1##2{%
    \rel@kern{0.8}%
    \overline{\rel@kern{-0.8}\macc@nucleus\rel@kern{0.2}}%
    \rel@kern{-0.2}%
  }%
  \macc@depth\@ne
  \let\math@bgroup\@empty \let\math@egroup\macc@set@skewchar
  \mathsurround\z@ \frozen@everymath{\mathgroup\macc@group\relax}%
  \macc@set@skewchar\relax
  \let\mathaccentV\macc@nested@a
  \macc@nested@a\relax111{#1}%
  \endgroup
}
\makeatother

\DeclarePairedDelimiter{\braces}{\lbrace}{\rbrace}

\newcommand{\selected}{\sC}

\newcommand{\dfmr}{\operatorname{DFMR}}

\newcommand{\mr}{\operatorname{MR}}
\newcommand{\coat}{\operatorname{COAT}}

\newcommand{\ours}{{CFMR}}
\begin{document}

\title{Byzantine-tolerant distributed learning of finite mixture models under partial corruptions}

\author{Yimei Zhang$^1$, Jiahua Chen$^2$, Xiaozhou Wang$^{3,4}$, Yan Shuo Tan$^5$, and Qiong Zhang$^{1,4}$\\
\thanks{$^1$ Institute of Statistics and Big Data, Renmin University of China, Beijing, 100872 China. 
$^2$ Department of Statistics, University of British Columbia, Vancouver, BC V6T 1Z4 Canada.
$^3$ School of Statistics, East China Normal University, Shanghai, 200062 China.
$^4$ Key Laboratory of Advanced Theory and Application in Statistics and Data Science-MOE, East China Normal University, Shanghai, 200062 China.
$^5$ Department of Statistics and Data Science, National University of Singapore, Singapore 117546.
}% <-this % stops a space
\thanks{Correspondence to: Qiong Zhang (\href{mailto:qiong.zhang@ruc.edu.cn}{qiong.zhang@ruc.edu.cn}) and Xiaozhou Wang (\href{mailto:xzwang@ecnu.edu.cn}{xzwang@ecnu.edu.cn}).}
}

% \markboth{Journal of \LaTeX\ Class Files,~Vol.~1, No.~2, December~2026}%
\markboth{}%
{Zhang \MakeLowercase{\textit{et al.}}: Component-wise Byzantine-tolerant distributed learning of finite mixture models}

\IEEEpubid{0000--0000~\copyright~2026 IEEE}

\maketitle

\begin{abstract}
Finite mixture models characterize heterogeneous populations and are increasingly fitted to distributed data using split-and-conquer procedures that aggregate local mixture estimates at a central server. 
The aggregation step can be seriously compromised when transmitted local mixture estimates are partially or completely corrupted.
To guard against Byzantine failures, existing robust aggregation methods have been developed for settings in which a local mixture estimate is either entirely authentic or entirely corrupted. 
Such methods can discard useful information when only some component estimates are corrupted. 
We consider component-wise Byzantine failure, in which some component estimates may be corrupted, but for each mixture component, a majority of the corresponding local estimates remain authentic. 
We propose component-wise filtered mixture reduction (CFMR), which selects a data-driven anchor, aligns transmitted components, filters each aligned cluster by a majority radius, and aggregates the retained estimates through mixture reduction. 
By filtering out only unreliable components, CFMR preserves authentic information from partially corrupted machines without requiring knowledge of the failure rates. 
We establish an adaptive convergence bound for CFMR and show that, under suitable conditions, it attains the oracle rate that would be achieved if the authentic component estimates were known in advance. 
Simulations and a real-data application show that CFMR remains close to the component-level oracle, whereas whole-machine filtering and unprotected aggregation can deteriorate substantially.
\end{abstract}

\begin{IEEEkeywords}
Byzantine failure; clustering; federated learning; label switching; robust aggregation; unsupervised learning.
\end{IEEEkeywords}
%% The following is a directive for TeXShop to indicate the main file
%%!TEX root = ../main.tex

\section{Introduction}
\label{sec:introduction}

Finite mixture models represent a heterogeneous population as a weighted combination of component distributions, each corresponding to a latent subpopulation.
They are widely used for clustering and density estimation, with applications in image analysis, biology, medicine, and other fields~\cite{liesenfeld2001generalized,baldry2004quantifying,sanchez2009retinal,kerbl20233d}.
In modern applications, a large dataset may be partitioned across $m$ machines for computational scalability, whereas decentralized or privacy-sensitive data may already reside on separate machines and cannot be pooled.
Both settings motivate split-and-conquer (SC) learning: each machine fits a local mixture model using its $n$ observations and transmits the estimate to a central server for aggregation.
For ordinary Euclidean parameters, this aggregation can often be performed by simply averaging local estimates or using a robust variant.
For finite mixture models, however, component ordering is not unique because any permutation represents the same mixture.
Independently fitted local mixtures may therefore use different orderings, so the $k$th component on one machine need not correspond to the $k$th component on another.
This mismatch prevents direct component-wise averaging but can be resolved by aligning and aggregating local mixtures through an optimal-transport formulation known as mixture reduction (MR)~\cite{zhang2022distributed}.

Because SC learning relies on local estimates transmitted to a central server, its aggregation step is vulnerable to Byzantine failure, under which some machines may transmit arbitrary estimates because of malfunction, data corruption, or adversarial intervention~\cite{lamport1982byzantine,yin2018byzantine,yin2019defending,tu2021variance}.
Distance-filtered mixture reduction (DFMR)~\cite{zhang2026byzantine} extends MR to this setting by filtering unreliable local mixtures before applying MR to the retained estimates.
DFMR adopts a machine-level failure model: each transmitted local mixture is treated as either entirely authentic or entirely corrupted, and \emph{a strict majority of machines must be fully failure-free}.
Although appropriate when failure affects a complete machine, this treatment can be unnecessarily conservative when a local mixture contains both authentic and corrupted component estimates.
Figure~\ref{fig:failure_models} illustrates this limitation.
In the right panel, every transmitted local mixture contains at least one corrupted component, so no machine is fully failure-free and the machine-level majority condition required by DFMR is violated.
Nevertheless, for each mixture component, fewer than half of its transmitted local estimates are corrupted, leaving most of the component-level information authentic.
Treating each local mixture as a single unit therefore fails to exploit the authentic component estimates carried by these partially corrupted machines.

This limitation motivates a more granular model, which we call \emph{component-wise Byzantine failure}.
Here, corruption acts on an individual estimated mixture component, including its mixing weight and distributional parameter, rather than on the entire local mixture.
For each component, a possibly different subset of machines may transmit arbitrary corrupted estimates; consequently, a machine may contain any combination of authentic and corrupted component estimates.
Such partial corruptions arise naturally in networked information systems: communication errors may affect selected symbols of a transmitted block~\cite{lin2001error}, sensor networks may exhibit channel- or node-specific faults~\cite{nowak2003distributed,gu2008distributed}, and federated or on-device systems may transmit unreliable class-wise summaries for only part of a local model~\cite{jeong2018communication}.
From a defender's perspective, the corruption pattern can generally be neither controlled nor known in advance, so a robust procedure must accommodate such heterogeneous patterns.
\begin{figure}[!ht]
\centering
\includegraphics[width=0.35\textwidth]{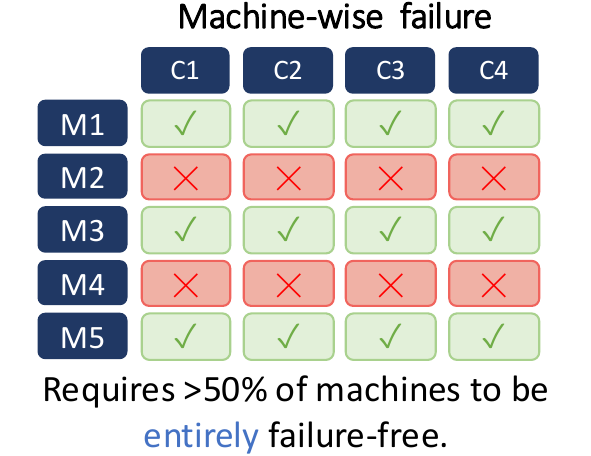}
\includegraphics[width=0.35\textwidth]{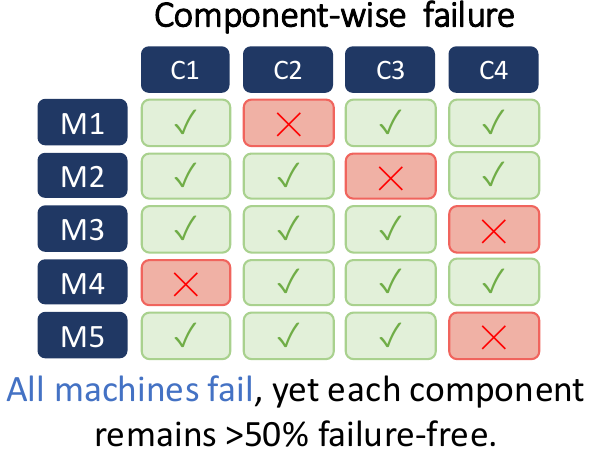}
\caption{\textbf{Machine-level and component-wise Byzantine failure.} Each row represents a machine, and each column represents a mixture component. Under machine-level failure (left), all components transmitted by a failed machine are corrupted. Under component-wise failure (right), corruption may affect different components on different machines.}
\label{fig:failure_models}
\end{figure}
Constructing such a procedure, however, requires resolving two coupled difficulties.
First, before filtering can be performed separately for each component, the central server must determine which transmitted components correspond to the same latent subpopulation; otherwise, even authentic estimates associated with different components may be filtered and aggregated together.
Second, the component-wise failure rates may differ across components and are unknown to the central server.
Because these rates are defined relative to the component correspondence, they cannot even be assessed before the transmitted components have been aligned.
Existing Byzantine-robust aggregation procedures are developed for Euclidean parameter vectors whose coordinates are already matched across machines~\cite{diakonikolas2017being,yin2018byzantine,lugosi2021robust,novikov2023robust}, and some additionally require knowledge of the failure rate to calibrate their filtering or trimming rules.
They therefore cannot be applied directly in the present setting.
Component alignment, robust filtering, and aggregation must therefore be addressed jointly.

To address these coupled difficulties, we propose \emph{component-wise filtered mixture reduction (\ours{})}.
Its construction exploits two forms of authentic information.
First, for each component, authentic local estimates constitute a majority across the $m$ machines, providing the basis for component-wise filtering.
Second, at least one machine transmits a fully failure-free local mixture, although its identity is unknown to the central server.
If this machine were known, its components would provide a common ordering template: the transmitted components could be assigned to their nearest anchor components, and the majority radius of each resulting cluster would provide a data-driven filtering scale.
Because the failure-free machine is unknown, \ours{} treats every transmitted local mixture as a candidate anchor.
For each candidate, it forms the anchor-induced clusters using a cost that combines discrepancies in mixing weights and subpopulation parameters, computes their majority radii, and selects the candidate with the smallest total radius.
Using the selected anchor, \ours{} aligns the transmitted components, filters each cluster with an inflated majority radius, and aggregates the retained estimates through MR.
The filtering thresholds are determined from the data, so the component-wise failure rates need not be known or specified as tuning parameters.
The resulting estimator requires no additional communication beyond the transmitted local mixtures.

Our theoretical analysis shows how the filtering mechanism determines the estimation error of \ours{}.
Under the likelihood, separation, and corruption-capacity conditions stated in Section~\ref{subsec:theoretical_properties}, we establish an adaptive convergence bound that separates the oracle sampling error, the effect of removing authentic tail estimates, and the contribution of corrupted estimates that survive filtering.
The bound therefore depends on the corrupted mass that remains near the selected anchor, rather than merely on the overall component-wise failure rate.
Under the explicit small-ball or separation conditions given in that section, the additional terms are asymptotically negligible; consequently, when $m\le n$ and $N=mn$, \ours{} attains the oracle rate $O_P(N^{-1/2})$, even when the machine-level majority condition required by DFMR is violated.

Complementing the theoretical analysis, our numerical studies evaluate anchor selection and component-wise aggregation under a range of corruption settings.
In simulations with Gaussian mixtures, the data-driven anchor selection succeeds with high probability, and \ours{} closely tracks the component-level oracle across different failure mechanisms, failure rates, sample sizes, numbers of machines, and degrees of component overlap.
The comparisons further demonstrate the loss of accuracy incurred by methods that aggregate or filter at the machine level under component-wise corruption.
An application to the CWRU bearing dataset~\cite{cwru_bearing_data} shows that \ours{} maintains high clustering accuracy as the component-wise failure rate increases.

The remainder of the paper is organized as follows.
Section~\ref{sec:preliminary_sc_gmm} introduces finite mixture models and SC learning, reviews MR, and summarizes DFMR under machine-level Byzantine failure.
Section~\ref{sec:proposed_estimator_comp}, including Subsection~\ref{subsec:theoretical_properties}, formulates component-wise Byzantine failure, develops \ours{}, and presents its theoretical properties.
Sections~\ref{sec:exp} and~\ref{sec:real_data} present the simulation study and a real-data application in industrial engineering, respectively.
The Appendix provides computational details, technical proofs, and additional simulation results.

%% The following is a directive for TeXShop to indicate the main file
%%!TEX root = ../main.tex

\section{A brief review of distributed learning of finite mixture models}
\label{sec:preliminary_sc_gmm}
In this section, we briefly review finite mixture models, the split-and-conquer (SC) framework, and the Distance-Filtered Mixture Reduction (DFMR) method of Zhang et al.~\cite{zhang2026byzantine}.

%%%%%%%%%%%%%%%%%%%%%%%%%%%%%%%%%%%%%%%%%%%%%%%%%%%%%%%%%%%%%%%%%%
\subsection{Finite mixture models}
A finite mixture model represents a heterogeneous population as a weighted combination of finitely many component distributions. 
We recall its definition below.
\begin{definition}[Finite mixture]
Let $\gF=\{f(x;\theta):x\in \gX\subseteq \sR^d,\theta\in\Theta\subseteq \sR^p\}$ be a parametric family of probability density functions, where $d$ and $p$ denote the dimensions of the observation and parameter spaces, respectively.
Throughout the paper, we assume that the parameter space $\Theta$ is compact and convex.
A finite mixture model with $K$ subpopulations has a density \[f_G(x)=\int_{\Theta} f(x;\theta)\, dG(\theta) =\sum_{k=1}^K w_k\, f(x;\theta_k),\]
which is parameterized by the mixing distribution $G=\sum_{k=1}^K w_k\,\delta_{\theta_k}$.
Here, $G$ assigns mass $w_k$ to the subpopulation parameter $\theta_k\in\Theta$ for each $k\in[K]=\{1,\ldots,K\}$. 
\end{definition}
Here, $f(x;\theta_k)$ denotes the $k$th component density, and the mixing weights $w=(w_1,\ldots,w_K)^\top$ lie in the probability simplex $\Delta_{K-1}=\{(w_1,\ldots,w_K): w_k\ge 0,\ \sum_{k=1}^K w_k=1\}$.
We use $\sG_K$ to denote the class of mixing distributions with at most $K$ distinct support points, with $\sG_{0}=\emptyset$.
We also assume that the finite mixture model is identifiable, in the sense that $f_{G_1}(x)=f_{G_2}(x)$ for all $x$ implies $G_1=G_2$.
This condition is standard for many commonly used mixture families, including Gaussian, Poisson, and Gamma mixtures.
For simplicity, we write all component distributions in terms of densities with respect to a common dominating measure, although finite mixture models and the proposed method can be defined more generally.

\begin{remark}[The label switching problem and finite mixture model parameterization]
\label{remark:parameterization}
One may instead use a vector such as $(w_1, w_2, \ldots, w_K, \theta_1, \theta_2, \ldots, \theta_K)^{\top}$ to parameterize the finite mixture.
However, this vector parameterization encounters the well-known label-switching problem: the component labels can be permuted without changing the mixture distribution.
Thus, different vectors may represent the same mixture and differ only in the ordering of their subpopulations.
For example, $(w_K, w_{K-1}, \ldots, w_1, \theta_K, \theta_{K-1}, \ldots, \theta_1)^{\top}$ yields the same mixture distribution, as do the other $K!-1$ permutations.
In contrast, parameterization through the mixing distribution $G$ is invariant to such permutations and avoids this artificial non-identifiability.
\end{remark}

%%%%%%%%%%%%%%%%%%%%%%%%%%%%%%%%%%%%%%%%%%%%%%%%%%%%%%%%%%%%%%%%%%
\subsection{Split-and-conquer (SC) learning of finite mixture models}
Let $\gX = \{x_{ij}: i\in[m], j\in [n]\}$
\footnote{
For simplicity, we present the analysis assuming equal local sample sizes across machines. The algorithm itself accommodates heterogeneous sample sizes through the weights $\lambda_i = n_i / N$, where $n_i$ is the sample size on the $i$th machine. The same results extend to this more general setting if we assume a relatively balanced partition, i.e. that $\min_i n_i / \max_i n_i \geq c$ for some constant $c > 0$.} 
be a collection of $N=nm$ independently and identically distributed (IID) samples from a mixture distribution $f_{G^*}(x)$ for some $G^* \in \sG_{K} \backslash \sG_{K-1}$ with $K$ distinct subpopulations, $K \geq 2$.
Suppose $\gX$ is partitioned completely at random into $m$ subsets $\gX_1, \ldots, \gX_m$, which are stored on $m$ local machines, where $\gX_{i} = \{x_{ij}: j \in [n]\}$ for $i \in [m]$.
A general SC procedure for estimating $G^*$ consists of the following two steps:
\begin{itemize}
\item
\emph{Local inference}: Obtain an ordinary constrained global MLE $\widehat{G}_i = \sum_{k} \widehat{w}_{ik}\delta_{\widehat{\theta}_{ik}}$ by maximizing the local log-likelihood function $\ell_i(G) = \sum_{j=1}^{n} \log f_G(x_{ij})$ over the compact model class, based on the data $\gX_i$ on the $i$th machine.

\item
\emph{Global aggregation}: Send $\{\widehat{G}_i: i \in [m]\}$ to a central server and aggregate them to obtain the final estimator $\widehat G$:
\be
\label{eq:sc_regular_aggregation}
\widehat G^{\mathfrak{A}} = \mathfrak{A} (\{\widehat{G}_i: i \in [m]\} ),
\ee
where $\mathfrak{A}(\cdot)$ is an aggregation operator, such as those studied in Zhang and Chen~\cite{zhang2022distributed}.
\end{itemize}

\begin{remark}[Ordinary and penalized local fits]
\label{rem:ordinary_and_penalized_local_fits}
The theoretical results in Section~\ref{subsec:theoretical_properties} concern the ordinary constrained global MLE. 
For Gaussian-mixture computations, a penalized likelihood and an EM implementation may instead be used to avoid degeneracy and obtain a stable local fit~\cite{chen2009inference}; see Appendix~\ref{app:em_algorithm}. 
Those penalized fits are used in the numerical work, but they are not covered by the local-MLE theorem below. 
Extending that theorem to a particular penalized estimator or an approximate EM solution would require separate conditions on the penalty and optimization error.
\end{remark}

Zhang and Chen~\cite{zhang2022distributed} proposed a mixture reduction (MR) estimator for aggregation. 
MR first pools the $m$ local mixtures into an overspecified \emph{average mixing distribution} with $mK$ components:
$
\widebar{G} = \sum_{i=1}^m \lambda_i \widehat{G}_i,
$
where $\lambda_i = n/N$ is the proportion of observations stored on the $i$th machine.  
Since the local estimators are consistent, each $\widehat{G}_i$ is close to $G^*$, and hence $\widebar{G}$ is also close to $G^*$.
However, $\widebar{G}$ generally has $mK$ support points and therefore lies outside the target parameter space $\sG_K$.
MR then seeks a $K$-component mixing distribution $G \in \sG_K$ that best approximates $\widebar{G}$ by minimizing the composite transportation divergence.
Specifically, let $c: \gF \times \gF \to \sR_{+}$ be a cost function, and write $c(\theta, \theta') = c(f(\cdot; \theta), f(\cdot; \theta'))$ for notational simplicity.
Let $G = \sum_{k=1}^{K} w_k \delta_{\theta_k}$ and  $G' = \sum_{k'=1}^{K'} w'_{k'} \delta_{\theta'_{k'}}$ be two mixing distributions with $K$ and $K'$ support points.
The composite transportation divergence between these mixing distributions is defined as
\begin{equation*}
    T_{c}(G, G') = \min \left \{
    \sum_{k=1}^{K}\sum_{k'=1}^{K'} \pi_{k, k'} c(\theta_k, \theta'_{k'}):~\sum_{k=1}^{K} \pi_{k,k'} = w'_{k'},  \sum_{k'=1}^{K'} \pi_{k,k'}= w_k, \pi_{k,k'} \geq 0
    \right \},
\end{equation*}
where the minimum is taken over all transportation plans $\pi$ satisfying the two marginal constraints.
This divergence represents the minimum cost of transporting the subpopulations of $G$ to those of $G'$~\cite{chen2017optimal,delon2020wasserstein,bing2022estimation,nguyen2013convergence}.
The MR estimator is then defined as
\begin{equation*}
\widehat{G}^{\mr} =\argmin\left\{T_{c}(\widebar{G}, G):~G\in\sG_{K}\right\}.
\end{equation*}
Although $\widehat G^{\mr}$ is defined through an apparently complicated bilevel optimization problem, the program simplifies and admits a clear interpretation, as discussed in Zhang and Chen~\cite{zhang2022distributed}.
More precisely, when minimizing $T_c(\widebar G, G)$ over $G \in \sG_K$, the marginal constraints associated with the unknown weights of $G$ are redundant, and each support point $\widehat{\theta}_{ik}$ in $\widebar G$ sends its full mass to the closest support point of $\widehat{G}^{\mr}$.
The optimal transportation plan can therefore be interpreted as clustering the local component parameters $\{\widehat{\theta}_{ik} : (i,k) \in [m] \times [K]\}$ into $K$ groups, with the support points of $\widehat{G}^{\mr}$ serving as cluster barycenters.
Thus, MR implicitly aligns local components through clustering while simultaneously aggregating them through barycentric averaging.
We give a more detailed description of the numerical algorithm for computing the MR estimator in Appendix~\ref{app:mm_algorithm}.

%%%%%%%%%%%%%%%%%%%%%%%%%%%%%%%%%%%%%%%%%%%%%%%%%%%%%%%%%%%%%%%%%%
\subsection{Robust aggregation under machine-level Byzantine failure}
Distributed learning systems are vulnerable to Byzantine failure, under which some local machines may deviate arbitrarily from the prescribed communication protocol because of hardware or software faults, communication errors, or malicious attacks. 
Consequently, the central server may receive arbitrary local estimates without knowing which transmissions are authentic.

Existing Byzantine-tolerant distributed learning methods commonly formulate failure at the machine level: each machine either transmits its complete authentic local estimate or sends an arbitrary replacement. 
Specifically, let $\gB\subset[m]$ denote the unknown set of Byzantine machines and let $\alpha=|\gB|/m$ be the corresponding failure rate. 
Under the machine-level Byzantine failure model, the central server observes
\[
    \widetilde{G}_i =
    \begin{cases}
        \widehat{G}_i, & i\notin \gB,\\
        \xi_i, & i\in \gB,
    \end{cases}
\]
where $\xi_i$ is arbitrary. 
This formulation treats the complete local mixture as a single transmitted object: once machine $i$ is Byzantine, no distinction is made between its reliable and corrupted component estimates.

The MR estimator is sensitive to such failures because it treats all transmitted local components as legitimate inputs. 
Consequently, corrupted components may be assigned to genuine clusters and pull the corresponding barycenters away from the true mixing distribution. 
Robust aggregation is therefore needed to limit the influence of Byzantine transmissions while preserving the accuracy gained from distributed learning.
Zhang et al.~\cite{zhang2026byzantine} address this problem through Distance-Filtered Mixture Reduction (DFMR), which combines machine-level filtering with MR aggregation.
At a high level, DFMR first identifies a reliable neighborhood among the transmitted local mixtures and then applies MR only to the estimates in this neighborhood.

This filtering strategy relies on a strict majority of machines being failure-free, namely $|\gB|<m/2$.
Because the failure-free local estimators concentrate around $G^*$, they form a tight cluster containing more than half of the transmitted mixtures.
Byzantine transmissions may be arbitrary and even coordinated, but they cannot constitute half of the transmitted estimates on their own.
DFMR exploits this majority structure through a robust machine-level center.
Specifically, DFMR first computes a robust center among the transmitted local estimates, called the Center Of ATtention (COAT).
For a distance $D$ between mixture densities, let $r(G)$ be the smallest radius of a ball centered at $G$ that contains at least half of the transmitted estimates.
The COAT estimator $\widehat{G}^{\coat}$ is chosen from the transmitted estimates as the local estimate that minimizes this radius.
Intuitively, COAT selects the transmitted mixture located near the center of the majority cluster.
DFMR then retains only the estimates lying within an inflated neighborhood of this robust center.
For a tuning parameter $\rho\geq 1$, define
\[
    \selected_{\rho}=\left\{i\in[m]:D(\widetilde{G}_i,\widehat{G}^{\coat})\leq \rho r(\widehat{G}^{\coat})\right\}.
\]
The final DFMR estimator applies MR only to the selected estimates:
\[
    \widehat{G}^{\dfmr}_{\rho}
    =
    \mathfrak{A}\left(\{\widetilde{G}_i:i\in\selected_{\rho}\}\right),
\]
where $\mathfrak{A}(\cdot)$ denotes the MR aggregation operator.
Thus, the distance-filtering step protects MR from arbitrary outlying mixtures, while the subsequent MR step aggregates the retained failure-free information.

DFMR has strong theoretical guarantees under the machine-level Byzantine failure model.  
To state its adaptive result, Zhang et al.~\cite{zhang2026byzantine} define $\widetilde\alpha_m(r) =(r/m)\left|\{i\in\gB:D(\xi_i,G^*)\le r\}\right|$.
Under their regularity, identifiability, and first-order strong-identifiability conditions, DFMR satisfies
\[
\|\widehat G^{\dfmr}_\rho-G^*\|
=O_P\{N^{-1/2}+\widetilde\alpha_m(2\rho n^{-1/2})\}.
\]
The first term is the oracle distributed-learning error, whereas the second adaptively measures the aggregate contribution of Byzantine transmissions lying near $G^*$.  
The analysis permits a nonzero number of authentic local mixtures to be filtered out and bounds their aggregate contribution.  
Under a finite $q$th-moment condition, the inflation parameter may be chosen as $\rho=m^\delta$ for $\delta>1/\{2(q-1)\}$; under sub-Gaussian local errors, a logarithmic choice is sufficient.  
When corrupted local estimates do not concentrate near $G^*$, the adaptive term is asymptotically negligible.
Consequently, DFMR is asymptotically equivalent to the oracle MR estimator that aggregates only the failure-free machines and inherits its first-order limiting behavior.

These guarantees rely on a strict majority of fully failure-free machines.
Under component-wise corruption, this condition may fail even when each true component retains an authentic majority across machines, leaving substantial authentic information within partially corrupted local mixtures.
This gap motivates the component-wise Byzantine failure model and aggregation method developed in the next section.

%% The following is a directive for TeXShop to indicate the main file
%%!TEX root = ../main.tex

\section{Byzantine-tolerant aggregation under component-wise failure}
\label{sec:proposed_estimator_comp}

The DFMR estimator reviewed in the previous section is designed for machine-level Byzantine failure, under which a complete local mixture is treated as either authentic or corrupted.
This treatment can be unnecessarily conservative when a transmitted local mixture contains both authentic and corrupted component estimates.
In this section, we formulate component-wise Byzantine failure and propose \emph{component-wise filtered mixture reduction} (\ours{}), which aligns the transmitted components, filters them within aligned clusters, and aggregates the retained estimates.
We then establish the convergence rate of \ours{} relative to the component-wise oracle estimator that aggregates only authentic component estimates.

\subsection{Problem setting: component-wise Byzantine failure}
Component-wise failure must be defined carefully because the components of a finite mixture are identifiable only up to permutation.
We therefore fix a labeling of the population mixture $G^*=\sum_{k=1}^K w_k^*\delta_{\theta_k^*}$ and interpret $k$ as the index of a population component.
Although each locally computed mixture estimate is unordered, the corruption model must specify which population component is corrupted on each machine.
For theoretical indexing, we therefore relabel the components of each clean local MLE by the permutation minimizing its Euclidean distance from the ordered components of $G^*$ and write $\widehat\psi_{ik}=(\widehat w_{ik},\widehat\theta_{ik}^{\top})^{\top}$ for the resulting estimate of $\psi_k^*=(w_k^*,\theta_k^{*\top})^{\top}$.
If several reordered component vectors attain the minimum, we choose the lexicographically smallest one; any remaining tie yields the same ordered vector.
This oracle relabeling is applied before contamination and serves only to define the Byzantine failure model; it is unavailable to the central server.

\begin{definition}[Component-wise Byzantine failure]
\label{def:component-wise_failure}
For each population component $k\in[K]$, let $\gB_k\subseteq[m]$ be a deterministic set, unknown to the central server, containing the machines on which the oracle-relabelled estimate of component $k$ is replaced.
The transmitted component occurrence with latent population index $k$ is
\[
    \widetilde{\psi}_{ik}
    =
    \begin{cases}
        \widehat{\psi}_{ik}, & i\notin\gB_k,\\
        \xi_{ik}, & i\in\gB_k,
    \end{cases}
\]
where the replacement values may be chosen jointly and may depend on all local data.
The index $k$ identifies the replaced population component and does not constrain the value of $\xi_{ik}$.
Writing $\widetilde\psi_{ik}=(\widetilde w_{ik},\widetilde\theta_{ik}^{\top})^{\top}$, define the transmitted local mixture estimate by $\widetilde G_i = \sum_{k=1}^K\widetilde w_{ik}\delta_{\widetilde\theta_{ik}}$.
Define $\gI_k=[m]\setminus\gB_k$, $\sO_k=\{(i,k):i\in\gI_k\}$.
Define the sets of machines containing at least one corrupted component and of fully authentic machines, respectively, by $\gB=\bigcup_{k=1}^K\gB_k$, $\gB^c=[m]\setminus\gB=\bigcap_{k=1}^K\gI_k$.
\end{definition}

The server observes $\widetilde G_i$ as an unordered multiset of $K$ transmitted component pairs; the latent indices $k$ and sets $\sO_k$ are used only for analysis and are not required by our method.
Each message must be a valid mixing distribution, with $\widetilde\theta_{ik}\in\Theta$, $\widetilde w_{ik}>0$, and $\sum_{k=1}^K\widetilde w_{ik}=1$.
This protocol requirement ensures that the discrepancy and barycentric update below are well defined.
Here corruption concerns the full component parameter $\psi_{ik}=(w_{ik},\theta_{ik}^{\top})^\top$: machine $i$ belongs to $\gB_k$ if either the mixing weight or the subpopulation parameter of component $k$ is replaced. 
Because the mixing weights must sum to one, changing one weight may require changing another; in that case, the machine belongs to the failure set of each affected component. Thus, the sets $\gB_1,\ldots,\gB_K$ may overlap.

\begin{assumption}[Component-wise majority]
\label{assump:component_wise_majority}
For every $k\in[K]$, let $\alpha_k=|\gB_k|/m$ denote the failure rate for component $k$. 
There exists a constant $\kappa\in(0,1/2)$ such that, for all sufficiently large $m$, $\max_{k\in[K]}\alpha_k\leq 1/2-\kappa$.
\end{assumption}
Thus, authentic estimates retain a uniformly positive majority margin separately for every true component. 
This condition places no restriction on $|\gB|/m$.
In particular, every transmitted local mixture may contain a corrupted component even though each true component retains an authentic majority across the $m$ machines.

\begin{remark}[Comparison with machine-level Byzantine failure]
Machine-level Byzantine failure is the special case $\gB_1=\cdots=\gB_K$, so every component on a failed machine is corrupted.
In that setting, methods such as DFMR filter whole local mixtures and therefore require a majority of fully failure-free machines.
By contrast, under component-wise failure, a machine may be unreliable for one component while still carrying useful information for the others.
Thus, Assumption~\ref{assump:component_wise_majority} can hold even when the machine-level majority condition fails, as illustrated in Figure~\ref{fig:failure_models}.
\end{remark}

\begin{remark}[Relation to cellwise contamination]
The proposed model resembles the cellwise or coordinate-wise contamination models studied in robust statistics~\cite{liu2021robust,raymaekers2024challenges}, but label switching creates a fundamental distinction.
Euclidean coordinates have fixed identities across observations, whereas mixture components are identifiable only up to permutation.
Consequently, the component correspondence must be recovered before coordinate-wise robust aggregation can be applied.
\end{remark}

\subsection{The component-wise filtered mixture reduction (\ours{}) estimator}
If the transmitted components were already aligned, Assumption~\ref{assump:component_wise_majority} would allow the central server to filter and aggregate the estimates of each true component separately.
The difficulty is that this majority condition is defined only after oracle alignment, whereas the server observes arbitrarily permuted component labels.
Without a common labeling template, authentic estimates from different true components may be grouped and aggregated together, as illustrated in the left panel of Figure~\ref{fig:alignment_with_without_anchor}.

\begin{figure}[!ht]
\centering
\includegraphics[width=0.35\textwidth]{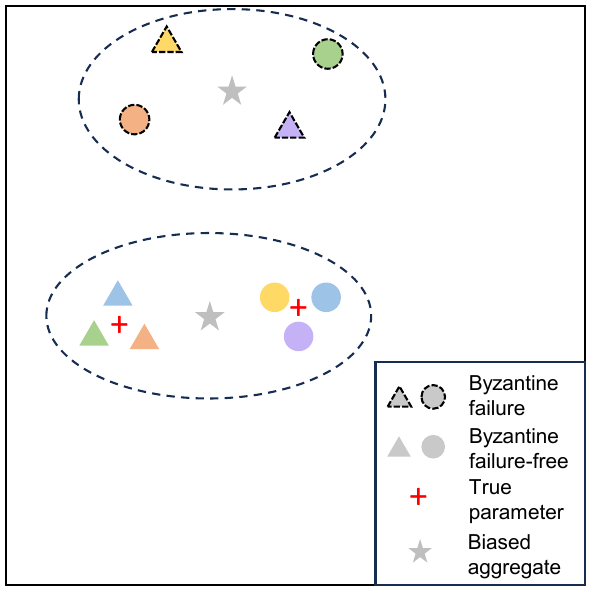}
\includegraphics[width=0.35\textwidth]{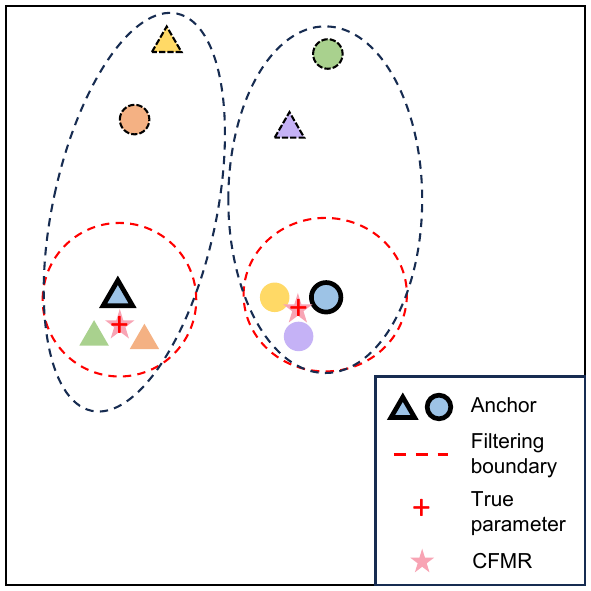}
\caption{\textbf{Alignment with and without an anchor.}
The example contains $m=5$ machines and $K=2$ mixture components.
Colors distinguish machines, while triangles and circles represent the two true components.
Dashed outlines indicate Byzantine-corrupted estimates, bold outlines identify the components transmitted by the unique failure-free anchor machine, and red crosses ($+$) mark the true parameters.
Here, the component-wise failure rate is $\alpha_k=0.4$, whereas the machine-level failure rate is $0.8$.
\textbf{Left:} Without a common label template, distance-based clustering can group failure-free estimates from different true components and form an artificial cluster, producing a biased aggregate (gray star $\star$).
\textbf{Right:} Nearest-anchor assignment correctly aligns the local components; the red dashed circles indicate the subsequent CFMR filtering regions, and the retained estimates yield a robust aggregate (red star $\star$).}
\label{fig:alignment_with_without_anchor}
\end{figure}

Component-wise majority alone does not resolve this ambiguity, because corruptions associated with different true components may be coordinated to create competing component groups as illustrated in the example below.
\begin{example}[Insufficiency of component-wise majority]
\label{ex:component_majority_not_sufficient}
Consider an idealized setting without local estimation error, with $m=K=3$, equal mixing weights, and four distinct subpopulation parameters $a,b,c,x$.
Suppose the central server receives
\[
    \widetilde G_1=\frac{1}{3}(\delta_x+\delta_b+\delta_c),\qquad
    \widetilde G_2=\frac{1}{3}(\delta_a+\delta_x+\delta_c),\qquad
    \widetilde G_3=\frac{1}{3}(\delta_a+\delta_b+\delta_x).
\]
These transmissions are compatible with the true mixing distribution $G_A^*=\frac{1}{3}(\delta_a+\delta_b+\delta_c)$: on machines $1$, $2$, and $3$, respectively, the authentic estimates of $a$, $b$, and $c$ are replaced by $x$.
Each true component then has failure rate $1/3<1/2$, although no machine is fully failure-free.
The same transmissions are also compatible with $G_B^*=\frac{1}{3}(\delta_a+\delta_b+\delta_x)$: on machines $1$ and $2$, respectively, the authentic estimates of $a$ and $b$ are replaced by $c$, while machine $3$ is failure-free.
Again, every component has failure rate below $1/2$.
Hence, under Assumption~\ref{assump:component_wise_majority} alone, the central server cannot distinguish these two true mixing distributions from the transmitted mixtures.
Some additional structural restriction is therefore needed to distinguish authentic component groups from coordinated corrupted alternatives.
\end{example}

To supply the additional structure identified in Example~\ref{ex:component_majority_not_sufficient}, we require at least one transmitted local mixture to provide a coherent labeling template.
A fully failure-free machine is a transparent sufficient condition, although it is not necessary; alternative structural conditions may also permit global alignment.
We therefore impose the following anchoring condition.
\begin{assumption}[Existence of a failure-free anchor machine]
\label{assump:one_failure-free_machine-method}
There exists at least one local machine that transmits all $K$ component estimates without corruption, i.e., $\gB^c\neq\emptyset$.
\end{assumption}
If the identity of a failure-free machine were known, its $K$ authentic components would establish a common component correspondence across the local mixture estimates.
Under the conditions stated later, this correspondence correctly groups authentic estimates associated with the same true component with probability tending to one.
The ambiguity in Example~\ref{ex:component_majority_not_sufficient} is thereby removed, as illustrated in the right panel of Figure~\ref{fig:alignment_with_without_anchor}.
In practice, the identity of the failure-free machine is unknown and need not be supplied to \ours{}.
We first describe how a given anchor is used to construct the estimator and then present a data-driven procedure for selecting the anchor.

\subsubsection{Aggregation given an anchor}
\label{sec:construction_of_CFMR}

Suppose an anchor $\bm{\psi}^{\dagger}=(\psi_1^{\dagger},\ldots,\psi_K^{\dagger})^{\top}$ is available.
When the anchor is supplied by machine $i^*$, we have $\bm{\psi}^{\dagger}=\widetilde{\bm\psi}_{i^*}$.
Given this anchor, \ours{} proceeds in three steps.
It first assigns each transmitted component estimate to its closest anchor component.
Within each aligned cluster, it then retains estimates whose discrepancy from the anchor does not exceed an inflated majority radius.
Finally, it combines the retained estimates to construct the corresponding component of the global mixture estimate.
To define the alignment and filtering criteria, let
$\psi=(w,\theta^\top)^\top$ denote a generic component parameter vector and
consider the discrepancy
\[
c_{\lambda}(\psi',\psi)=D(\theta',\theta)+\lambda(w'-w)^2,
\]
where $\lambda>0$ balances discrepancies in the subpopulation parameter and the mixing weight.
Including both terms is important because Byzantine corruption may affect either quantity or both.
We take $D(\theta',\theta)$ to be the reversed Bregman divergence generated by a strictly convex function\footnote{Appropriate choices of $A(\cdot)$ yield the squared Euclidean distance, the squared Mahalanobis distance, and the Kullback--Leibler divergence for exponential families.} $A(\cdot)$: $D(\theta',\theta) =A(\theta)-A(\theta')-(\theta-\theta')^\top\nabla A(\theta')$.

\begin{wrapfigure}{r}{0.48\textwidth}
\vspace{-0.7\baselineskip}
\centering
\footnotesize
\begin{minipage}{\linewidth}
\hrule
\vspace{0.4ex}
\refstepcounter{algorithm}
\textbf{Algorithm~\thealgorithm: \ours{} aggregation given an anchor.}
\label{alg:submachine_level}
\par\vspace{0.4ex}
\hrule
% \vspace{0.3ex}
\begin{algorithmic}[1]
\State {\bfseries Input:} Anchor machine $i^*$, $\{\widetilde{G}_i:i\in[m]\}$, $c_{\lambda}(\cdot,\cdot)$, and $\rho$.
\For{$k=1,\ldots,K$}
    \State Assign
    \[
        \sC_k\leftarrow
        \{(i,j):\widetilde{\psi}_{ij}
        \text{ is closest to }\widetilde{\psi}_{i^*k}\}.
    \]
    \State Compute the majority radius
    \[
        r_k\leftarrow
        \{c_{\lambda}(\widetilde{\psi}_{ij},\widetilde{\psi}_{i^*k}):
        (i,j)\in\sC_k\}_{\lceil m/2\rceil}.
    \]
    \State Retain
    \[
        \widehat{\sO}_k\leftarrow
        \{(i,j)\in\sC_k:
        c_{\lambda}(\widetilde{\psi}_{ij},\widetilde{\psi}_{i^*k})\leq\rho r_k\}.
    \]
    \State Update $\widehat{\theta}_k$ by the barycenter over $\widehat{\sO}_k$.
    \State Set
    $\widehat{w}'_k\leftarrow|\widehat{\sO}_k|^{-1}
    \sum_{(i,j)\in\widehat{\sO}_k}\widetilde{w}_{ij}$.
\EndFor
\State Normalize $\widehat{w}_k\leftarrow
\widehat{w}'_k/\sum_{\ell=1}^{K}\widehat{w}'_\ell$.
\State {\bfseries Output:}
$\widehat{G}^{\mathrm{CFMR}}\leftarrow
\sum_{k=1}^{K}\widehat{w}_k\delta_{\widehat{\theta}_k}$.
\end{algorithmic}
% \vspace{0.3ex}
\hrule
\end{minipage}
\end{wrapfigure}

The construction begins by assigning each transmitted component estimate to its nearest anchor component.
For each $(i,j)\in[m]\times[K]$, define the anchor-assignment label
\[
\mathfrak a_{ij}^\dagger\in
\operatorname*{arg\,min}_{k\in[K]}
c_{\lambda}(\widetilde\psi_{ij},\psi_k^\dagger),
\]
where, among distinct tied anchor components, we choose the one whose component vector is smallest in a fixed lexicographic order.
This value-based convention is invariant to the order in which the anchor components are presented.
A candidate anchor containing duplicate component vectors is assigned an infinite total radius and excluded from anchor selection.
The anchor $\bm\psi^\dagger$ then induces the transmitted-index clusters
\be
\label{eq:cost_based_alignment}
\sC_k(\bm{\psi}^{\dagger})
    =
    \left\{
        (i,j)\in[m]\times[K]:
        \mathfrak a_{ij}^\dagger=k
    \right\}.
\ee
Thus, for an admissible anchor with distinct component vectors, $\sC_1(\bm\psi^\dagger),\ldots,\sC_K(\bm\psi^\dagger)$ form a partition of $[m]\times[K]$.
Under accurate alignment, authentic components from the same machine enter distinct clusters, whereas multiple corrupted components may still enter one cluster.
When no ambiguity arises, we suppress the dependence on $\bm\psi^\dagger$.

Within each aligned cluster, we next compute a majority radius.
For the $k$th anchor component, define
\begin{equation}
\label{eq:majority_radius}
    r(\psi_k^{\dagger})
    =
    \begin{cases}
    \left\{
        c_{\lambda}(\widetilde{\psi}_{ij},\psi_k^{\dagger}):
        (i,j)\in\sC_k(\bm{\psi}^{\dagger})
    \right\}_{\lceil m/2\rceil},
    & |\sC_k(\bm\psi^\dagger)|\ge \lceil m/2\rceil,\\
    \infty,
    & |\sC_k(\bm\psi^\dagger)|< \lceil m/2\rceil
    \end{cases},
\end{equation}
where $\{\cdot\}_{q}$ denotes the $q$th order statistic of the displayed multiset of real numbers.
When finite, $r(\psi_k^\dagger)$ is the smallest radius centered at $\psi_k^\dagger$ that contains at least $\lceil m/2\rceil$ assigned estimates.
If fewer than $\lceil m/2\rceil$ component estimates are assigned to the $k$th anchor component, the resulting cluster lacks the required majority support; we therefore set $r(\psi_k^\dagger)=\infty$.
Under correct alignment, the authentic estimates of every true component provide the required majority support.
Thus, the infinity convention flags a candidate anchor whenever one of its components fails to attract a sufficiently large cluster, while a finite majority radius provides a data-driven filtering scale without requiring knowledge of $\alpha_k$.

Within the $k$th aligned cluster, \ours{} retains the component estimates lying within an inflated radius of the anchor component:
\begin{equation*}
    \widehat{\sO}_k
    =
    \left\{
        (i,j):(i,j)\in\sC_k,\ 
        c_{\lambda}(\widetilde{\psi}_{ij},\psi_k^{\dagger})
        \leq \rho r(\psi_k^{\dagger})
    \right\},
\end{equation*}
where $\rho\geq1$ controls the tolerance of the filter.
Using the majority radius directly can exclude authentic estimates when the corruption rate is well below one half, so the radius is inflated before filtering.
The retained subpopulation parameters are then aggregated through the barycenter
\[
    \widehat{\theta}_k
    =
    \argmin_{\theta}
    \sum_{(i,j)\in\widehat{\sO}_k}
        \widetilde{w}_{ij}D(\widetilde{\theta}_{ij},\theta),
\]
while the preliminary mixing weight is
\[    \widehat{w}'_k
    =
    \frac{1}{|\widehat{\sO}_k|}
    \sum_{(i,j)\in\widehat{\sO}_k}\widetilde{w}_{ij}.
\]
The preliminary weights are then normalized by $\widehat{w}_k=\widehat{w}'_k/\sum_{\ell=1}^{K}\widehat{w}'_\ell$ so that they sum to one.
The barycentric update is not robust by itself: robustness comes from the component alignment and majority-radius filtering performed before this ordinary aggregation step.
Algorithm~\ref{alg:submachine_level} summarizes the aggregation stage.

\Needspace{0.42\textheight}
\begin{wrapfigure}{r}{0.48\textwidth}
\vspace{-0.7\baselineskip}
\centering
\footnotesize
\begin{minipage}{\linewidth}
\hrule
\vspace{0.4ex}
\refstepcounter{algorithm}
\textbf{Algorithm~\thealgorithm: Anchor machine selection.}
\label{alg:initial_estimate}
\par\vspace{0.4ex}
\hrule
% \vspace{0.3ex}
\begin{algorithmic}[1]
\State {\bfseries Input:} $\{\widetilde{G}_i:i\in[m]\}$ and $c_{\lambda}(\cdot,\cdot)$.
\State Initialize $\gR\leftarrow\infty$ and $i^*\leftarrow1$.
\For{$i'=1,\ldots,m$}
    \If{$\widetilde{\bm\psi}_{i'}$ contains duplicate component vectors}
        \State Set $R_{i'}\leftarrow\infty$ and continue to the next candidate.
    \EndIf
    \State Set $R_{i'}\leftarrow0$.
    \State Assign all transmitted components to the nearest component of $\widetilde{G}_{i'}$ and form $\sC_1(\widetilde{\bm\psi}_{i'}),\ldots,\sC_K(\widetilde{\bm\psi}_{i'})$.
    \For{$k=1,\ldots,K$}
        \If{$|\sC_k(\widetilde{\bm\psi}_{i'})|<\lceil m/2\rceil$}
            \State Set $r_{i'k}\leftarrow\infty$.
        \Else
            \State Compute
            \[
                r_{i'k}\leftarrow
                \{c_{\lambda}(\widetilde{\psi}_{ij},\widetilde{\psi}_{i'k}):
                (i,j)\in\sC_k(\widetilde{\bm\psi}_{i'})\}_{\lceil m/2\rceil}.\]
        \EndIf
        \State Update $R_{i'}\leftarrow R_{i'}+r_{i'k}$.
    \EndFor
    \If{$R_{i'}<\gR$}
        \State Set $\gR\leftarrow R_{i'}$ and $i^*\leftarrow i'$.
    \EndIf
\EndFor
\State {\bfseries Output:} index of the selected anchor machine $i^*$.
\end{algorithmic}
% \vspace{0.3ex}
\hrule
\end{minipage}
% \vspace{-2.5\baselineskip}
\end{wrapfigure}

\subsubsection{Choosing the anchor}
Because the identity of the failure-free machine is unknown, \ours{} searches for an anchor among the $m$ transmitted local mixtures.
Assumption~\ref{assump:one_failure-free_machine-method} guarantees that this candidate set contains at least one valid labeling template.
Therefore, to identify a valid template within this candidate set, \ours{} forms the anchor-induced clusters and majority radii for each candidate $\widetilde G_{i'}$ and summarizes their overall concentration by the total radius score $R(\widetilde{\bm\psi}_{i'}) = \sum_{k=1}^{K}r(\widetilde\psi_{i'k})$.
The total radius score favors candidates for which every anchor component is supported by a compact, majority-sized cluster.
By convention, a candidate containing duplicate component vectors or inducing an undersized cluster receives an infinite score and is excluded whenever a candidate with a finite score exists.
Among the remaining candidates, minimizing the score favors those that induce more concentrated component clusters.
Accordingly, \ours{} selects any minimizer 
\[i^*\in\operatorname*{arg\,min}_{i'\in[m]}R(\widetilde{\bm\psi}_{i'}).\]
Algorithm~\ref{alg:initial_estimate} summarizes this selection step.

Together, Algorithms~\ref{alg:initial_estimate} and~\ref{alg:submachine_level} constitute \ours{}.
The procedure requires neither the oracle ordering nor the component-wise failure rates, uses only the transmitted local mixtures, and introduces no additional communication among machines.
Its computational cost is dominated by anchor selection; once the anchor is selected, the remaining operations have the same basic assignment-and-barycenter structure as MR.

\section{Theoretical properties}
\label{subsec:theoretical_properties}
This section investigates the statistical properties of the proposed \ours{}. 
The analysis proceeds in three steps. 
We first study the ideal case in which a failure-free anchor machine is available and establish the resulting alignment and filtering properties.
We then show that the proposed anchor selection procedure identifies such an anchor with probability tending to one.
Finally, combining these results, we establish the asymptotic properties of the resulting \ours{} estimator.

%%%%%%%%%%%%%%%%%%%%%%%%%%%%%%%%%%%%%%%%%%%%%%%%%%%%%%%%%%%%%%%%%%%%%%%%%%%%%%%%%%%%%%%%%%%
\subsection{Notation and assumptions}
For a positive integer $m$, $[m]$ denotes the set $ \{1, 2, \ldots, m\}$.
For a vector $x \in \sR^d$, let $\|x\|$ denote its Euclidean norm.
Let $d_E(x,x')=\|x-x'\|$ denote the Euclidean distance.
For a square matrix $H$, let $\lambda_{\min}(H)$ and $\lambda_{\max}(H)$ denote its smallest and largest eigenvalues, respectively, and let $\matrixnorm{H} = \sqrt{\lambda_{\max}(H^\top H)}$ denote its operator (spectral) norm.
We write $H \succeq 0$ and $H \succ 0$ for positive semidefinite and positive definite matrices, respectively, with analogous notation for
negative semidefinite and negative definite matrices.
For a parametric density function $f(x;\theta)$, the gradient with respect to the parameter is denoted by $\nabla f(x;\theta')=
\partial f(x;\theta)/\partial \theta|_{\theta=\theta'}.$
For a sequence of random variables $X_n$ and a positive deterministic sequence $a_n$, we write $X_n=O_P(a_n)$ if, for every $\epsilon>0$, there exist constants $C_{\epsilon}>0$ and $n_{\epsilon}$ such that $\sP\{|X_n|\le C_{\epsilon}a_n\}\ge 1-\epsilon$ for all $n\ge n_{\epsilon}$.
Similarly, $X_n=\Omega_P(a_n)$ if, for every $\epsilon>0$, there exist constants $c_{\epsilon}>0$ and $n_{\epsilon}$ such that $\sP\{|X_n|\ge c_{\epsilon}a_n\}\ge 1-\epsilon$ for all $n\ge n_{\epsilon}$.
We write $X_n=\Theta_P(a_n)$ when both $X_n=O_P(a_n)$ and $X_n=\Omega_P(a_n)$ hold.
We write $X_n=o_P(a_n)$ if $|X_n|/a_n\rightarrow0$ in probability, and $X_n=\omega_P(a_n)$ if $|X_n|/a_n\rightarrow\infty$ in probability; equivalently, $X_n=\omega_P(a_n)$ if $\sP\{|X_n|>Ca_n\}\rightarrow 1$ for every fixed constant $C>0$.
More generally, for a nonnegative random sequence $Y_n$, the notation $X_n=O_P(Y_n)$ means that, for every $\epsilon>0$, there exist constants $C_\epsilon>0$ and $n_\epsilon$ such that $\sP\{|X_n|\le C_\epsilon Y_n\}\ge1-\epsilon$ for all $n\ge n_\epsilon$. When $Y_n>0$ almost surely, this is equivalent to $X_n/Y_n=O_P(1)$; in that case, $X_n=o_P(Y_n)$ means $X_n/Y_n=o_P(1)$.
Given a nonnegative discrepancy metric $d(\cdot,\cdot)$, we define a ball at $z$ with radius $r$ by $B_r(z;d)=\{z':d(z',z)\le r\}$.

We next state the assumptions used to establish the statistical properties of \ours{}. They fall into three groups.
The first group concerns the finite mixture likelihood and the behavior of the local estimators. The second describes the local geometry of the component discrepancy. 
The third concerns the amount and concentration of the corrupted component estimates.

%%%%%%%%%%%%%%%%%%%%%%%%%%%%%%%%%%%%%%%%%%%%%%%%%%%%%%%%%%%%%%%%%%%%%%%%%%%%%%%%%%%%%%%%%%%

\vspace{1ex} \noindent
\paragraph{\bf Assumptions on the finite mixture model}
The first group consists of regularity conditions on the finite mixture model and its local likelihood.
Recall that the true mixing distribution $G^*$ has components $\psi_k^* = (w_k^*, \theta_k^{*\top})^\top$ and that a generic $K$-component mixing distribution $G$ has components $\psi_k = (w_k, \theta_k^{\top})^\top$.

\begin{assumption}[Finite mixture model]
\label{assump:compact_parameter_space}
The parameter space $\Theta$ is compact and convex.
The true support points are distinct and satisfy
$\theta_k^*\in\operatorname{int}(\Theta)$ for every $k\in[K]$, while the true weight vector lies in the relative interior of the simplex: $\Delta_{K-1}=\{(w_1,\ldots,w_K)^{\top}, w_k>0,\sum_{k=1}^K w_k=1\}$.
Thus, $G^*=\sum_{k=1}^K w_k^*\delta_{\theta_k^*}\in \mathbb G_K\setminus\mathbb G_{K-1}$.
The ordinary likelihood is maximized over the compact constrained model class induced by $\Theta$ and the closed weight simplex. 
\end{assumption}

\begin{assumption}[Identifiability]
\label{assump:identifiability}
The mixture model on $\gF$ is identifiable: for any $G,G'\in\sG_K$, equality of $f_G$ and $f_{G'}$ almost everywhere implies $G=G'$ as mixing distributions.
\end{assumption}
This is the standard identifiability condition for finite mixtures and is satisfied by many commonly used component families, including Gaussian mixtures~\cite{teicher1961identifiability,yakowitz1968identifiability}.

To state the likelihood conditions, fix an ordering of the components of $G^*$.
Let $S_K$ denote the set of permutations of $[K]$.
The oracle ordering of $G$ is determined by
\begin{equation}
\label{eq:parameter_vector}
\sigma_G\in \operatorname*{arg\,min}_{\sigma\in S_K}
\sum_{k=1}^K
\left\{
|w_{\sigma(k)}-w_k^*|^2
+\|\theta_{\sigma(k)}-\theta_k^*\|^2
\right\}.
\end{equation}
If several permutations attain the minimum, we choose the lexicographically smallest one.
Because the true support points are distinct, the minimizing permutation is unique with probability tending to one for any consistent estimator of $G^*$.
For $G=\widehat G_i$, this ordering gives the oracle relabeling used in Definition~\ref{def:component-wise_failure}.
After oracle ordering, define the free parameter vector
\[
\mG
=\bigl(w_{\sigma_G(1)},\ldots,w_{\sigma_G(K-1)},
\theta_{\sigma_G(1)}^\top,\ldots,\theta_{\sigma_G(K)}^\top\bigr)^\top
\in\sR^{d_G},
\qquad d_G=Kp+K-1,
\]
where $w_{\sigma_G(K)}=1-\sum_{k=1}^{K-1}w_{\sigma_G(k)}$.
In particular, $\mG^*=(w_1^*,\ldots,w_{K-1}^*,\theta_1^{*\top},\ldots,\theta_K^{*\top})^\top$.
All likelihood derivatives and MLE expansions below use the free coordinate $\mG$.

Let $f_{\mG}$ denote the mixture density represented by $\mG$, and define the log-likelihood contribution of one observation and its population counterpart by $\ell(\mG;X)=\log f_{\mG}(X)$ and $L(\mG)=\sE_{G^*}\{\ell(\mG;X)\}$.
The score and Fisher information at $\mG^*$ are $s_{\mG^*}(X)=\left.\nabla_{\mG}\ell(\mG;X)\right|_{\mG=\mG^*}$ and $\mI(\mG^*)=\sE_{G^*}\{s_{\mG^*}(X)s_{\mG^*}(X)^\top\}$.
When no ambiguity arises, the subscript on the expectation is omitted.

\begin{assumption}[Likelihood geometry]
\label{assump:local_strong_concavity}
$L(\mG)$ is twice continuously differentiable in a neighborhood of $\mG^*$, differentiation may be interchanged with expectation at $\mG^*$, and the information identity holds: $\mI(\mG^*)
=-\sE_{G^*}\{\nabla^2_{\mG}\ell(\mG^*;X)\} =-\nabla^2L(\mG^*)\succ0$.
\end{assumption}

\begin{assumption}[Likelihood smoothness, concentration, and moment order]
\label{assump:smoothness}
For a fixed $q\geq8$, the following conditions hold.
\begin{enumerate}[label=(\alph*), leftmargin=*]
\item \textbf{Local $q$th-order smoothness.}
There exist constants $\delta_0,H_q<\infty$ and a nonnegative function $W(x)$ satisfying $\sE\{W(X)^q\}\leq H_q$ s.t., for all $\mG,\mG'$ in a $\delta_0$-neighborhood of $\mG^*$, $\sE\{\|\nabla\ell(\mG;X)\|^q\}\le H_q$, $\sE\{\matrixnorm{\nabla^2\ell(\mG;X)-\nabla^2L(\mG)}^q\}
\le H_q$, and
\[
\matrixnorm{\nabla^2\ell(\mG';x)-\nabla^2\ell(\mG;x)}
\le W(x)\|\mG'-\mG\|.
\]
\item \textbf{Global likelihood concentration.}
The population log-likelihood $L(\mG)=\sE\{\log f_{\mG}(X)\}$ is continuous on the compact constrained model class.  
If $L_n(\mG)=n^{-1}\sum_{j=1}^n\log f_{\mG}(X_j)$ is the empirical log-likelihood based on an IID sample of size $n$, then
\[
\sE\left[
\sup_{\mG\in\sG_K}
\left|
\{L_n(\mG)-L_n(\mG^*)\}-\{L(\mG)-L(\mG^*)\}
\right|^q
\right]
\le H_q n^{-q/2}.
\]
\end{enumerate}
\end{assumption}
Part~(a) is the $q$th-order version of the local smoothness condition in Zhang, Tan and Chen~\cite{zhang2026byzantine}; their stated condition has $q=8$.  
Part~(b) is a quantitative uniform law of large numbers for the ordinary global MLE.  
It is not implied by local smoothness and compactness alone.  
Instead, it replaces the global concavity used to localize a global $M$-estimator in Huang and Huo~\cite{huang2019distributed}, which is generally unavailable for a finite-mixture likelihood.  
Part~(b) supplies polynomial control of the rare event on which the global maximizer leaves the local likelihood chart.  
As proved in Appendix~\ref{app:local_mle}, these likelihood-level conditions give $\sE\|\widehat{\mG}_i-\mG^*\|^s=O(n^{-s/2})$ for every $1\leq s\leq q$.
Thus the moment order used below is inherited from likelihood smoothness, rather than imposed directly on the local estimators.  
Here $q$ is fixed; if $q$ were allowed to grow with $n$, the dependence of $H_q$ and the resulting moment constants on $q$ would also have to be controlled.

\begin{remark}[A sufficient condition for global likelihood concentration]
\label{rem:global_likelihood_concentration}
Part~(b) of Assumption~\ref{assump:smoothness} follows from standard bracketing maximal inequalities for empirical processes~\cite[Chapter~19]{van2000asymptotic} if the centered log-likelihood class $\mathcal L
=\{x\mapsto\ell(\mG;x)-\ell(\mG^*;x):\mG\in\sG_K\}$
has a $q$-integrable envelope and polynomial $L_2(P_{\mG^*})$ bracketing numbers.

For example, consider a fixed-$K$ Gaussian location--covariance mixture with means in a compact set and covariance matrices satisfying $\underline\lambda I\preceq\Sigma\preceq\overline\lambda I$.  
Uniformly over the compact component-parameter space, there are constants $c_1,c_2,c_3>0$ such that
\[
c_1\exp\{-c_2(1+\|x\|^2)\}
\le \inf_{\theta\in\Theta}f(x;\theta)
\le f_{\mG}(x)
\le \sup_{\theta\in\Theta}f(x;\theta)
\le c_3
\]
for every $G\in\sG_K$.  
Consequently, $\sup_{G\in\sG_K}|\ell(G;x)-\ell(G^*;x)|
\le C(1+\|x\|^2)$.
This envelope has moments of every fixed order under the true Gaussian mixture.  
On a bounded observation set the log-density class is a uniformly smooth finite-dimensional parametric class.  
Truncating to balls of radius $O\{\sqrt{\log(1/\varepsilon)}\}$, using the Gaussian tail outside those balls, and covering the compact parameter set inside them gives polynomial
$L_2(P_{G^*})$ bracketing numbers.  
The bracketing maximal inequality therefore yields the $O(n^{-q/2})$ bound in Assumption~\ref{assump:smoothness}(b).  
\end{remark}

\vspace{1ex} \noindent
\paragraph{\bf Assumption on the component discrepancy}
The proposed procedure performs component assignment and filtering using $c_\lambda(\cdot,\cdot)$.
The following assumption requires this discrepancy to induce the same local geometry as the Euclidean distance in a neighborhood of the true component parameters.
\begin{assumption}[Divergence]
\label{assump:reverse_bregman_divergence}
The subpopulation-parameter divergence $D(\theta',\theta)$ is a reversed Bregman divergence induced by a continuously differentiable, strictly convex function $A(\cdot)$ defined on an open set containing $\Theta$:
\[
D(\theta',\theta)=D_A(\theta,\theta')
=A(\theta)-A(\theta')-(\theta-\theta')^\top\nabla A(\theta').
\]
For each $k\in[K]$, the function $A$ is twice continuously differentiable on a neighborhood $B_{\epsilon}(\theta_k^*;\|\cdot\|)$ and, for all $\theta,\theta'\in B_{\epsilon}(\theta_k^*;\|\cdot\|)$,
\[
\eta_-\mI\preceq\nabla^2A(\theta)\preceq\eta_+\mI,
\qquad
\matrixnorm{\nabla^2A(\theta)-\nabla^2A(\theta')}\le \widetilde A\|\theta-\theta'\|,
\]
for constants $\eta_+\ge\eta_->0$ and $\widetilde A>0$.
Since $K$ is fixed, the same constants may be chosen for all components.
\end{assumption}
Together with the quadratic weight term in $c_\lambda$, Assumption~\ref{assump:reverse_bregman_divergence} implies that, in a neighborhood of each $\psi_k^*$, there exist constants $0<c_-\leq c_+<\infty$ such that $c_-\|\psi'-\psi\|^2
\leq c_\lambda(\psi',\psi)
\leq c_+\|\psi'-\psi\|^2$.
Thus, locally, a $c_\lambda$-ball of radius $r$ corresponds to a Euclidean neighborhood of radius proportional to $r^{1/2}$.

\vspace{1ex} \noindent
\paragraph{\bf Assumptions on the corruption mechanism}
%%%%%%%%%%%%%%%%%%%%%%%%%%%%%%%%
The theoretical properties of \ours{} are established under the following assumptions on the corruption mechanism.
In addition to the component-wise majority condition and the existence of a fully authentic local mixture estimate assumed in Assumptions~\ref{assump:component_wise_majority} and~\ref{assump:one_failure-free_machine-method}, respectively, we further require the exclusion of pathological contamination patterns that may mimic authentic component estimates.

\begin{assumption}[Anchor separability]
\label{assump:anchor_separability}
For any fixed constant $C>0$,
    \[
    \sP\left(
    \forall i'\in\gB,\ \exists k\in[K]\ \text{such that}\
    \left|
    \left\{
    (i,j)\in\sC_k(\widetilde{\bm\psi}_{i'}):
    c_{\lambda}(\widetilde{\psi}_{ij},\widetilde{\psi}_{i'k})\le Cn^{-1}
    \right\}
    \right|
    <\left\lceil \frac{m}{2}\right\rceil
    \right)\to1.
    \]
Here $\sC_k(\widetilde{\bm\psi}_{i'})$ is the index set obtained by treating $\widetilde{\bm\psi}_{i'}$ as the candidate anchor and assigning transmitted components to their nearest anchor component.
\end{assumption}

\begin{assumption}[Corruption capacity]
\label{assump:component_capacity}
The following two capacity conditions distinguish fixed truth-centred balls from balls centred at random authentic candidates.
\begin{enumerate}[label=(\alph*), leftmargin=*]
\item \textbf{Truth-centred capacity.}
Given the inflation sequence $\rho$ employed in \ours{}, there exists a constant $\beta\in(0,1/2)$ such that, for any fixed constant $C>0$,
    \[
    \sP\left(
    \max_{k\in[K]}
    \left|
    \left\{(i,\ell):\ \ell\in[K],\ i\in\gB_\ell,\
    c_{\lambda}(\xi_{i\ell},\psi_k^*)\le C\rho n^{-1}
    \right\}
    \right|
    \le \beta m
    \right)\to1.
    \]
\item \textbf{Authentic-anchor capacity.}
There exists a constant $\beta_\dagger\in(0,1/2)$ such that, for every fixed $C>0$,
\[
\sP\left(
\max_{i\in\gB^c}\max_{k\in[K]}
\left|
\left\{(r,\ell):\ \ell\in[K],\ r\in\gB_\ell,\
c_\lambda(\xi_{r\ell},\widehat\psi_{ik})\le Cn^{-1}
\right\}
\right|
\le\beta_\dagger m
\right)\to1.
\]
\end{enumerate}
\end{assumption}
Part~(a) limits the number of corrupted occurrences near each true component at the filtering scale and is required throughout the convergence analysis. 
Part~(b) provides additional uniform control around failure-free candidate anchors. 
It is required only when the analysis permits the filter to remove a vanishing number of authentic tail estimates, as formalized by the tail-trimming alternative in Theorem~\ref{theorem:rate_of_convergence}.
Appendix~\ref{app:assumptions} provides sufficient conditions for Assumptions~\ref{assump:anchor_separability} and~\ref{assump:component_capacity}, together with examples of coordinated contamination patterns under which these assumptions fail.

%%%%%%%%%%%%%%%%%%%%%%%%%%%%%%%%%%%%%%%%%%%%%%%%%%%%%%%%%%%%%%%%%%%%%%%%%%%%%%%%%%%%%%%%%%%
\subsection{Statistical properties}
Recall from Definition~\ref{def:component-wise_failure} that $\gI_k$ is the set of machines carrying an authentic estimate of component $k$ and that $\sO_k=\{(i,k):i\in\gI_k\}$ is the corresponding set of authentic component occurrences.
These latent occurrence labels are used only for analysis and are not available to the central server.

We first consider the idealized case in which the identity of a failure-free anchor machine is known.
To justify aggregation with this anchor, we must show that its components correctly align the authentic local estimates and that the resulting majority radii have the appropriate scale for filtering.
The following lemma establishes these two properties.

\begin{lemma}[Ideal failure-free anchor]
\label{lemma:properties_under_ideal_anchor-method}
Suppose $i^*$ is a failure-free anchor machine.  
Under Assumptions~\ref{assump:component_wise_majority},~\ref{assump:compact_parameter_space}--\ref{assump:smoothness},~\ref{assump:reverse_bregman_divergence}, and~\ref{assump:component_capacity}(a), with independent local data splits, consider the asymptotic regime where $m,n,\rho\to\infty$ and $m=o(n^{q/2})$ and $\rho=o(n)$.
Then
\[
\sP\left\{\bigcap_{k=1}^K
\bigl(\sO_k\subseteq\sC_k(\widetilde{\bm\psi}_{i^*})\bigr)
\right\}\to1,
\qquad
r(\widetilde\psi_{i^*k})=\Theta_P(n^{-1})
\]
for every fixed $k\in[K]$.
\end{lemma}
The probability event in Lemma~\ref{lemma:properties_under_ideal_anchor-method} means that, simultaneously for all $k\in[K]$, every authentic estimate of the $k$th true component is assigned to the $k$th anchor-induced cluster with probability tending to one.
The relation $r(\widetilde\psi_{i^*k})=\Theta_P(n^{-1})$ further shows that the corresponding majority radius remains on the $n^{-1}$ component-cost scale and therefore provides a properly calibrated basis for filtering.
In practice, however, the identity of a failure-free anchor is unknown.
Algorithm~\ref{alg:initial_estimate} addresses this problem by minimizing the total radius score over all transmitted local mixtures.
A failure-free candidate has a total radius of order $n^{-1}$, whereas Assumption~\ref{assump:anchor_separability} ensures that every corrupted candidate has an asymptotically larger radius score.
The following theorem formalizes this separation and establishes that the minimization procedure selects a failure-free anchor.

\begin{theorem}[Selection of a failure-free anchor]
\label{theorem:accurate_selection}
Under Assumptions~\ref{assump:component_wise_majority},~\ref{assump:one_failure-free_machine-method},~\ref{assump:compact_parameter_space}--\ref{assump:smoothness},~\ref{assump:reverse_bregman_divergence}, and~\ref{assump:anchor_separability}, with independent local data splits, consider the asymptotic regime where $m,n\to\infty$ and $m=o(n^{q/2})$.
Then
\[
\sP\left\{
\operatorname*{arg\,min}_{i\in[m]}R(\widetilde{\bm\psi}_i)
\subseteq\gB^c
\right\}\to1.
\]
\end{theorem}
The event in Theorem~\ref{theorem:accurate_selection} states that every minimizer of the total radius score belongs to $\gB^c$.
Thus, with probability tending to one, the selected anchor is failure-free regardless of how ties among minimizers are resolved.

Having established that the data-driven procedure recovers a failure-free anchor, we next study the convergence rate of the resulting \ours{} estimator.
Even with a failure-free anchor, some corrupted component occurrences may remain within the filtering neighborhoods, so the convergence rate depends on how many of them are retained.
For any selected minimizer $i^*$, let $\widehat\sO_k$ be the set retained by Algorithm~\ref{alg:submachine_level} for the $k$th anchor-induced cluster, and define the retained-corruption level
\[
\bar\alpha_k
=\frac1m|\widehat\sO_k\setminus\sO_k|,
\qquad
\bar\alpha_{\max}=\max_{k\le K}\bar\alpha_k.
\]
Also set $m_k=|\gI_k|=(1-\alpha_k)m$, $R_{ok}=(m_kn)^{-1/2}+n^{-1}$, and $R_{o,\max}=\max_{k\le K}R_{ok}$.
With this notation, the following theorem quantifies the component-wise estimation error of \ours{} and distinguishes whether all authentic estimates are retained or a tail subset of authentic estimates may be removed by the filter.

\begin{theorem}[Adaptive convergence rate of \ours{}]
\label{theorem:rate_of_convergence}
Under Assumptions~\ref{assump:component_wise_majority},~\ref{assump:one_failure-free_machine-method},~\ref{assump:compact_parameter_space}--\ref{assump:smoothness},~\ref{assump:reverse_bregman_divergence},~\ref{assump:anchor_separability}, and~\ref{assump:component_capacity}(a), with independent local data splits, consider the asymptotic regime
where $m,n,\rho\to\infty$ and $m=o(n^{q/2})$ and $\rho=o(n)$.
Within this asymptotic regime, distinguish the following two cases:
\begin{enumerate}[label=(\alph*), leftmargin=*]
\item \textit{Exact retention.} If $m\rho^{-q/2}\to0$, set $T_{nq}=0$.
\item \textit{Tail trimming.} If Assumption~\ref{assump:component_capacity}(b) also holds, set $T_{nq}=n^{-1/2}\rho^{-(q-1)/2}$.
\end{enumerate}
Then, in the corresponding regime, for every $k\in[K]$,
\[
|\widehat w_k-w_k^*|=O_P\left\{R_{o,\max}
+T_{nq}+\bar\alpha_{\max}\sqrt{\rho/n}\right\},\quad\|\widehat\theta_k-\theta_k^*\|
=O_P\left\{R_{ok}
+T_{nq}+\bar\alpha_k\sqrt{\rho/n}\right\}.
\]
Moreover, in regime~(a), every authentic component estimate is retained with probability tending to one.
\end{theorem}
The bounds in Theorem~\ref{theorem:rate_of_convergence} separate three sources of error. The terms $R_{ok}$ and $R_{o,\max}$ are the convergence rates of the corresponding oracle aggregators based only on authentic component estimates; $T_{nq}$ accounts for authentic tail estimates removed by filtering; and $\bar\alpha_k\sqrt{\rho/n}$, or $\bar\alpha_{\max}\sqrt{\rho/n}$ for the weights, measures the contribution of corrupted occurrences that remain after filtering.
In the exact-retention regime, $T_{nq}=0$ and every authentic component estimate is retained with probability tending to one.
The tail-trimming regime permits some authentic estimates to be removed, but requires Assumption~\ref{assump:component_capacity}(b) to prevent corrupted values from concentrating around an unusually noisy failure-free anchor candidate.

When $m\le n$, component-wise majority implies $R_{ok}=O(N^{-1/2})$ for $N=mn$.
The exact-retention regime therefore achieves the oracle convergence rate whenever $\bar\alpha_{\max}=o_P\{(m\rho)^{-1/2}\}$, and the tail-trimming regime achieves the same rate if, in addition, $\rho^{q-1}/m\to\infty$.
These conditions depend on the corrupted mass that survives the actual filter: $\bar\alpha_k$ may be much smaller than the overall component corruption fraction, or it may be identically zero.
Because the sets $\sO_k$ depend on the latent population correspondence and failure sets, however, $\bar\alpha_k$ is an analytical quantity rather than an observable statistic.
Theorem~\ref{theorem:rate_of_convergence} is therefore an adaptive oracle bound.

To obtain an explicit rate, we next impose a stochastic small-ball condition under which the probability that a corrupted occurrence falls near a true component decreases polynomially with the neighborhood radius.
The following corollary converts this condition into a bound on $\bar\alpha_k$ and then substitutes that bound into Theorem~\ref{theorem:rate_of_convergence}.

\begin{corollary}[Small-ball stochastic contamination]
\label{cor:oracle_rate_small_ball}
Assume the common structural and likelihood conditions of Theorem~\ref{theorem:rate_of_convergence}, except that Assumption~\ref{assump:component_capacity}(a) is to be deduced below. 
Conditional on the clean local estimates, suppose there are constants $\gamma>0$, $C_Q<\infty$, and $r_0>0$ such that, for every $\ell,k\in[K]$, $i\in\gB_\ell$, and $0<r<r_0$, $\sP\{c_\lambda(\xi_{i\ell},\psi_k^*)\le r
\mid\mathcal F_{\mathrm{clean}}\}
\le C_Qr^\gamma$.
Let $m\le n$ and $\rho=m^a$.  
Consider either
\begin{enumerate}[label=(\alph*)]
\item the exact-retention regime $2/q<a<1$, or
\item the tail-trimming regime $1/(q-1)<a<1$, together with
Assumption~\ref{assump:component_capacity}(b).
\end{enumerate}
Then Assumption~\ref{assump:component_capacity}(a) holds, $\bar\alpha_{\max}=O_P\{(\rho/n)^\gamma\}$, and, for every $k\in[K]$,
\[
|\widehat w_k-w_k^*|+\|\widehat\theta_k-\theta_k^*\|
=O_P\left\{N^{-1/2}+(\rho/n)^{\gamma+1/2}\right\}.
\]
Consequently, the oracle convergence rate holds if $m^{1/2}\rho^{\gamma+1/2}=o(n^\gamma)$, or, equivalently for $\rho=m^a$, $m^{1/2+a(\gamma+1/2)}=o(n^\gamma)$.
\end{corollary}
The first conclusion of Corollary~\ref{cor:oracle_rate_small_ball} shows that only an $O_P\{(\rho/n)^\gamma\}$ fraction of corrupted occurrences fall sufficiently close to the true components to survive filtering.
Substitution into the adaptive bound gives a contamination contribution of order $O_P\{(\rho/n)^{\gamma+1/2}\}$.
Under the displayed joint growth condition, this contribution is $o_P(N^{-1/2})$, and \ours{} therefore attains the oracle convergence rate.

The choice of $\rho$ balances two requirements: it must grow sufficiently quickly to retain authentic component estimates, but sufficiently slowly to exclude corrupted occurrences and satisfy $\rho=o(n)$.
Under the fixed $q$th-moment condition, a convenient exact-retention choice is $\rho=m^{2/q}L_m$ $m^{2/q}L_m=o(n)$ with $L_m\to\infty$ or, when $m\le n$, $\rho=m^{2/q+\delta}$ with $0<\delta<1-2/q$.
For the tail-trimming regime, the corresponding oracle-rate choice is $\rho=m^{1/(q-1)+\delta}$ with $0<\delta<1-1/(q-1)$.
For example,
\[
\begin{array}{c|c|c|c}
q&\text{exact-retention }\rho
&\text{exact-retention Euclidean scale}&\text{tail-trimming }\rho\\ \hline
8&m^{1/4+\delta}&m^{1/8+\delta/2}n^{-1/2}&m^{1/7+\delta}\\
12&m^{1/6+\delta}&m^{1/12+\delta/2}n^{-1/2}&m^{1/11+\delta}\\
16&m^{1/8+\delta}&m^{1/16+\delta/2}n^{-1/2}&m^{1/15+\delta}.
\end{array}
\]
Higher moments reduce the required inflation, but do not improve the clean $n^{-1}$ cost scale, the oracle error $R_{ok}$, the contamination small-ball exponent, or the structural anchor conditions.

Under the small-ball model with $\rho=m^a$, the exact joint growth condition in Corollary~\ref{cor:oracle_rate_small_ball} is primary.
If only $m\le n$ is imposed, a convenient sufficient condition is $a<(2\gamma-1)/(2\gamma+1)$, provided that the right-hand side is positive.
If only the $p$-dimensional subpopulation parameter has a bounded density, then $\gamma=p/2$ and the conditions become $m^{1+(p+1)a}=o(n^p)$ and $a<(p-1)/(p+1)$.
These upper bounds on $a$ must be combined with the lower requirement $a>2/q$ for exact retention or $a>1/(q-1)$ for tail trimming.
If the corrupted values are separated from every true component by a fixed positive cost, then $\bar\alpha_{\max}=0$ with probability tending to one and no small-ball remainder appears.

\begin{remark}[Sub-Weibull authentic tails]
If, instead of a fixed-moment tail, the authentic local estimators satisfy
\[
\sup_{i,k}\sP\{\sqrt n\|\widehat\psi_{ik}-\psi_k^*\|>t\}
\le C\exp(-ct^\nu),
\]
then the exact-retention uniform-clean condition becomes $m\exp(-c'\rho^{\nu/2})\to0$.
Thus $\rho$ can be chosen on the order of $(\log m)^{2/\nu}$ with a sufficiently large multiplier; for sub-Gaussian tails this is the familiar logarithmic scale.
This is a separate tail-class assumption, not a consequence of possessing all fixed moments.
\end{remark}

\begin{remark}[Comparison with DFMR]
The inflation scales in CFMR and DFMR are directly comparable once their different radius conventions are taken into account.
DFMR inflates a distance radius and, under a finite $q$th-moment condition, permits an inflation factor $m^\delta$ for any $\delta>1/\{2(q-1)\}$~\cite{zhang2026byzantine}.
CFMR instead inflates the locally quadratic component cost $c_\lambda$.
In the tail-trimming regime, its leading inflation threshold is $m^{1/(q-1)}$ on the cost scale and hence $m^{1/\{2(q-1)\}}$ on the induced Euclidean-distance scale.
For example, when $q=8$, the CFMR threshold $m^{1/7}$ on the cost scale corresponds to $m^{1/14}$ on the Euclidean-distance scale, matching the leading DFMR threshold.
Thus, component-wise filtering does not by itself require an additional tail-inflation penalty.

The two methods differ in how distance is related to mixture parameters.
DFMR filters using an $L^2$ distance between mixture densities and invokes a family-wide first-order strong-identifiability condition to convert density distance into parameter distance.
CFMR filters directly using a component discrepancy that is locally equivalent to squared Euclidean parameter distance.
Its analysis therefore relies instead on nonsingular Fisher information at $G^*$ and the local geometry of $c_\lambda$.
\end{remark}

The preceding results describe the error of the estimated weights and subpopulation parameters separately.
The following corollary expresses the same conclusion for the estimated mixing distribution as a whole under the Wasserstein distance.

\begin{corollary}[Mixing-distribution error]
\label{cor:wasserstein_rate}
Equip $\Theta$ with Euclidean distance and let $W_1$ be the corresponding Wasserstein distance between mixing distributions.  
Under Theorem~\ref{theorem:rate_of_convergence},
\[
W_1(\widehat G,G^*)
=O_P\left\{R_{o,\max}
+T_{nq}+\bar\alpha_{\max}\sqrt{\rho/n}\right\},
\]
where $T_{nq}$ is defined by the applicable regime in that theorem.
\end{corollary}
Corollary~\ref{cor:wasserstein_rate} shows that passing from component-wise errors to mixing-distribution error introduces no additional order of approximation.
The result follows by matching the estimated and true components in their canonical order, coupling their common mass directly, and transporting the residual weight over the compact parameter space.

%% The following is a directive for TeXShop to indicate the main file
%%!TEX root = ../main.tex

\section{Simulation study}
\label{sec:exp}
We generate data from finite Gaussian mixtures to assess the performance of \ours{} under component-wise Byzantine failure.
Gaussian mixtures are widely used finite mixture models and retain the main difficulties considered in this paper, including label switching, overlapping components, and corruption in component means, covariances, and weights.
The experiments investigate whether \ours{} can effectively recover reliable component-level information across a range of distributed learning settings.

%%%%%%%%%%%%%%%%%%%%%%%%%%%%%%%%%%%%%%%%%%%%%%%%%%
\subsection{Simulation setting}
Let $\gF = \{\phi(x;\mu,\Sigma): \mu \in \sR^d, \Sigma\in \sS_{+}^{d\times d}\}$, where $\phi(x;\mu,\Sigma)$ is the density of a multivariate Gaussian distribution
with mean $\mu$ and covariance matrix $\Sigma$.
We denote by $\sS_{+}^{d\times d}$ the set of $d\times d$ positive definite matrices.
Finite Gaussian mixtures are finite mixture models whose component distributions belong to $\gF$.

To mitigate potential human bias in the choice of simulation settings, the parameter values of the Gaussian mixtures are generated at random using the \texttt{R} package \texttt{MixSim}~\cite{melnykov2012mixsim}.
An important characteristic of a Gaussian mixture is the degree of overlap between its component distributions. 
Following \cite{melnykov2012mixsim}, we quantify the difficulty of learning a mixture by its maximum pairwise overlap. 
Specifically, define
\[
\texttt{MaxOmega} = \max_{i,j\in[K],\,i \neq j} \{ o_{j|i} + o_{i|j} \},
\]
where $o_{j|i} = \sP\!\left(w_{i} f(X;\theta_i)< w_{j} f(X;\theta_j)\mid X \sim f(\cdot;\theta_i)\right)$.
Larger values of \texttt{MaxOmega} indicate greater overlap among the mixture components and consequently a more challenging learning problem.

We consider a collection of simulation configurations specified by the dimension $d$, mixture order $K$, maximum overlap \texttt{MaxOmega}, local sample size $n$, and number of local machines $m$.
Specifically, we set $d=10$, $K=5$, \texttt{MaxOmega}$\in\{10\%,20\%,30\%\}$, $n\in\{5000,10000,20000\}$, and $m\in\{20,50,100\}$.
The choices of $d$, $K$, $n$, and $m$ are intended to represent distributed learning problems involving high-dimensional mixtures and large sample sizes. 
For each simulation configuration, the complete data generation, local estimation, corruption, 
and aggregation procedure is repeated independently $R=300$ times. 
The resulting Monte Carlo replications are used to estimate the performance measures reported below.

For the $r$th replication under a given configuration, we first generate the parameter values of a Gaussian mixture distribution $G^{(r)}$ using \texttt{MixSim}.
We then draw an IID sample $\gX^{(r)}$ from $G^{(r)}$ and partition it uniformly at random among the $m$ local machines. 
As noted in Remark~\ref{rem:ordinary_and_penalized_local_fits}, on each machine, we estimate the local mixing distribution using the penalized maximum likelihood estimator proposed by Chen and Tan~\cite{chen2009inference}, computed via the EM algorithm with penalty size $n^{-1/2}$ and a convergence tolerance of $10^{-6}$ for the per-observation log-likelihood.

To isolate the performance of the aggregation procedure from that of the local optimization algorithm, the EM algorithm is initialized at the true mixing distribution on every machine. 
The resulting local estimates are then used as the clean baseline estimates before Byzantine corruption is injected into the transmitted component estimates.
This initialization also aligns the component labels across machines by construction, so that the $k$th corrupted component always corresponds to the same true mixture component. 
This alignment is used solely to generate the corrupted transmissions and is not available to any aggregation procedure, which receives only the transmitted local mixture estimates without cross-machine component labels.

Component-wise Byzantine corruption is then generated as follows.
For each mixture component $k$, the failure rate is set to $\alpha_k$.
We vary $\alpha_k$ from $0.1$ to $0.4$ in increments of $0.1$ across different configurations.
Machine $m$ is always kept fully failure-free to satisfy Assumption~\ref{assump:one_failure-free_machine-method}.
To investigate the robustness of \ours{} under different corruption mechanisms, we consider three representative component-wise Byzantine attacks:

\begin{itemize}
    \item \emph{Mean failure.}
    The component mean is the primary location parameter of a Gaussian mixture and is often the most sensitive to adversarial perturbation.
    To generate substantial location errors, each corrupted component mean is replaced by a random vector whose entries are generated independently from $N(0,100^2)$, a contamination mechanism commonly adopted in Byzantine aggregation studies~\cite{tu2021variance}.
    For each component $k$, the corrupted machines are selected uniformly at random from the first $m - 1$ machines.

    \item \emph{Covariance failure.}
    To perturb the dispersion structure of a Gaussian component while maintaining a valid covariance matrix, we add $\sum_{i=1}^{d}\xi_i\xi_i^{\top}$ to each corrupted covariance matrix, where $\xi_1,\ldots,\xi_d$ are IID standard Gaussian random vectors. 
    As with mean failure, the corrupted machines for each component are sampled uniformly at random from the first $m - 1$ machines.

    \item  \emph{Weight failure.}
    Because the mixing weights must sum to one, changing one weight must change at least one other weight on the same machine.
    We therefore corrupt the weights jointly.
    For each component $k$, we select exactly $c_k=\lfloor\alpha_km\rfloor$ corrupted occurrences from the first $m-1$ machines.
    Let $S_i\subseteq[K]$ denote the set of corrupted components on machine $i$, where $|S_i|\geq2$ whenever $S_i$ is nonempty; the weights outside $S_i$ remain unchanged.
    To perturb the weights in $S_i$ without changing their total mass, define $M_i=\sum_{k\in S_i}\widehat{w}_{ik}$, draw $Z_i\sim\operatorname{Dirichlet}(\beta\mathbf{1}_{|S_i|})$ with $\beta=0.5$, and set $\widetilde{\mathbf{w}}_{i,S_i}=M_iZ_i$.
    This construction randomly reallocates $M_i$ among the corrupted components while ensuring that the transmitted weights remain positive and sum to one.
\end{itemize}
Together, these three corruption mechanisms perturb the canonical sets of parameters in a Gaussian mixture model: location, dispersion, and mixing proportions. 
Other forms of corruption, such as parameter values outside the admissible parameter space, are often detectable from prior knowledge of the mixture model and are therefore excluded from the present study.

%%%%%%%%%%%%%%%%%%%%%%%%%%%%%%%%%%%%%%%%%%%%%%%%%%
\subsection{Baselines and performance metrics}
\label{sec:estimators}
To evaluate the performance of \ours{}, we compare it with a collection of existing aggregation procedures and oracle benchmarks.
Since no existing Byzantine-tolerant split-and-conquer method is designed specifically for component-wise Byzantine failure in finite mixture models, the competing methods represent different strategies for handling either standard aggregation or machine-level Byzantine failure.

Whenever a competing procedure requires a tuning parameter related to the contamination level, we set it equal to the true failure rate. 
This oracle choice favors the competing methods and provides a conservative assessment of the proposed procedure.
The competing aggregation procedures are as follows.
Whenever Mixture Reduction (MR) is employed, we use the Kullback--Leibler (KL) divergence between Gaussian components as the discrepancy measure.

\begin{enumerate}
\item \textbf{Vanilla.}
This is the original Mixture Reduction (MR) estimator of~\cite{zhang2022distributed}, which aggregates all transmitted local mixture estimates without robust filtering.

\item \textbf{TRIM.}
The trimmed $k$-barycenter estimator of~\cite{del2019robust}, 
which discards an $\eta$-fraction of component estimates before aggregation.
Since $\eta$ is a user-specified hyperparameter, we set it equal to the true component-wise failure rate $\alpha_k$.

\item \textbf{DFMR.}
This is the Distance-Filtered Mixture Reduction estimator of~\cite{zhang2026byzantine}, designed for machine-level Byzantine failure.
The implementation follows the original paper using the same filtering rule and tuning parameter $\rho$.

\item \textbf{Filtering.}
The filtering procedure of~\cite{diakonikolas2017being}, originally proposed for robust mean estimation in Euclidean space.
Since the method requires aligned observations, we first apply the proposed anchor selection and alignment procedure (Algorithm~\ref{alg:initial_estimate}) and then perform filtering independently within each component group before aggregation.
The procedure is applied separately to the transmitted means, covariance matrices, and mixing weights according to the simulated failure type.
The filtering threshold is set equal to the true component-wise failure rate.
\end{enumerate}

The following two estimators serve as oracle benchmarks.
\begin{enumerate}
\setcounter{enumi}{4}
\item 
\textbf{M-Oracle.}
This is a machine-level oracle that applies MR only to the fully failure-free local mixture estimates indexed by $\gB^c$.

\item \textbf{C-Oracle.}
This is a component-level oracle that, for each component $k$, aggregates only the failure-free component estimates $\{\widetilde{\psi}_{ik}: i\in\gI_k\}$ using the same barycentric update as MR.
This estimator provides the benchmark performance under component-wise Byzantine failure.
\end{enumerate}

The estimated mixing distribution is denoted by $\widehat{G}^{(r)} = \sum_k \widehat{w}_k^{(r)} \delta_{(\widehat{\mu}_k^{(r)}, \widehat{\Sigma}_k^{(r)})}$. We use the following performance metrics(see Appendix~\ref{app:metrics} for their precise ground-cost definitions):
\begin{itemize}
    \item \textbf{Wasserstein distance} ($W_1 = W_1(\widehat{G}^{(r)}, G^{(r)})$): the Wasserstein distance between the estimated and true mixing distributions.
    \item \textbf{Adjusted Rand Index (ARI)}: mixture models are widely used for model-based clustering. We measure ARI between the maximum a posteriori (MAP) clustering assignments induced by $\widehat{G}^{(r)}$ and those induced by the true mixing distribution $G^{(r)}$. Values closer to 1 indicate better performance.
\end{itemize}

%%%%%%%%%%%%%%%%%%%%%%%%%%%%%%%%%%%%%%%%%%%%%%%%%%
\subsection{Sensitivity analysis of tuning parameters}
\subsubsection{Robustness to the inflation factor \texorpdfstring{$\rho$}{rho}}
The inflation factor $\rho$ is a tuning parameter in \ours{} that determines the fraction of component estimates retained in each component cluster (Algorithm~\ref{alg:submachine_level}).
Our theoretical analysis provides a sufficient order-of-magnitude requirement on $\rho$ to establish the asymptotic properties of \ours{}, but this requirement may be conservative in finite samples.
It is therefore of practical interest to investigate the empirical sensitivity of \ours{} to different choices of $\rho$.
The results show that \ours{} performs consistently well over a wide range of $\rho$ values.

Figure~\ref{fig:varying_rho_m_100} examines the effect of $\rho$ under the setting where $m=100$, $n=10^4$, and \texttt{MaxOmega}$=0.3$.
The component-wise failure rate varies over $\{0.0, 0.1, 0.2, 0.3, 0.4\}$, while $\rho$ ranges from $1.0$ to $3.0$ in increments of $0.1$.
Each configuration is replicated $R=100$ times.
For mean failures, \ours{} performs comparably to C-Oracle once $\rho \ge 1.6$, and the estimation error remains nearly unchanged as $\rho$ increases.
For covariance and weight failure, a mild U-shaped pattern is observed, but the estimator remains stable over a broad range of $\rho$ values, particularly when $\alpha_k \le 0.2$.
\begin{figure}[!htbp]
    \centering
    \begin{minipage}[t]{0.32\linewidth}
        \centering
        \includegraphics[width=\linewidth]{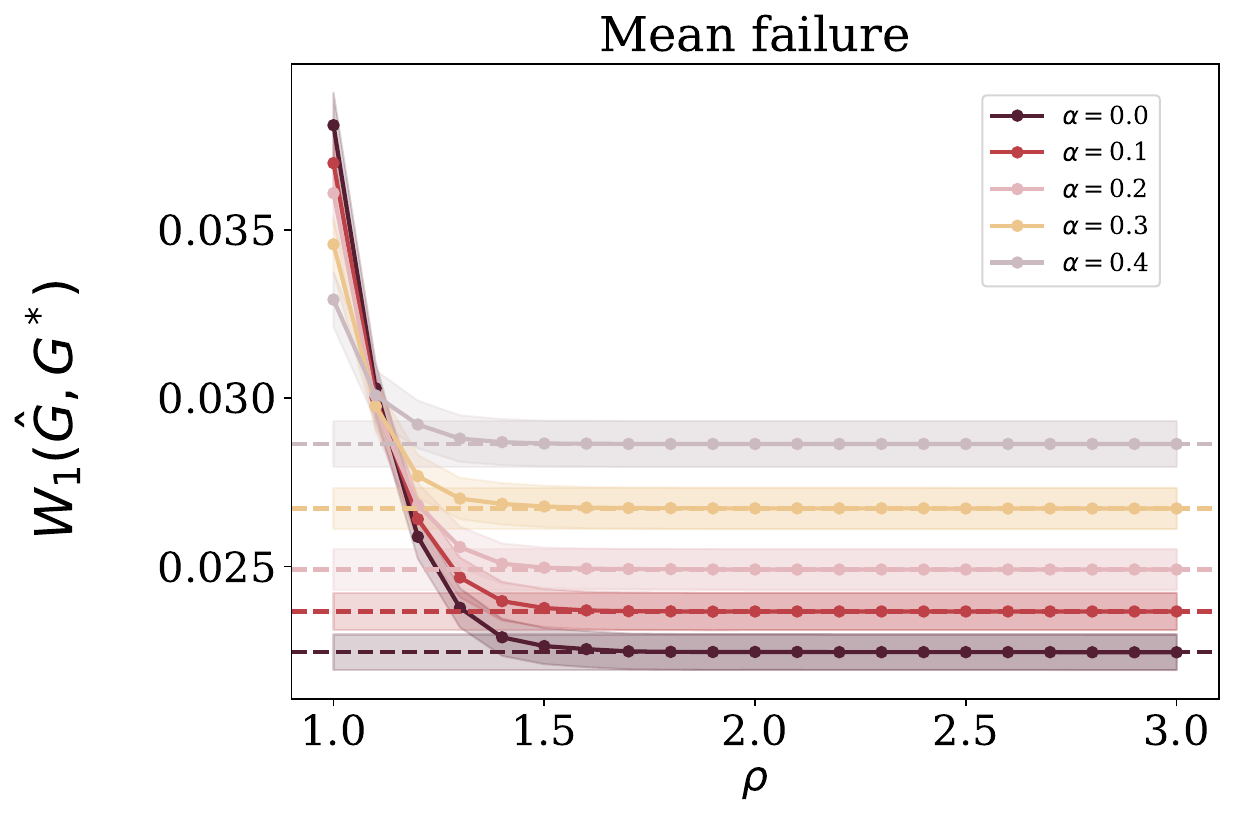}
        \par\vspace{2pt}{\footnotesize (a) Mean failure}
    \end{minipage}
    \hfill
    \begin{minipage}[t]{0.32\linewidth}
        \centering
        \includegraphics[width=\linewidth]{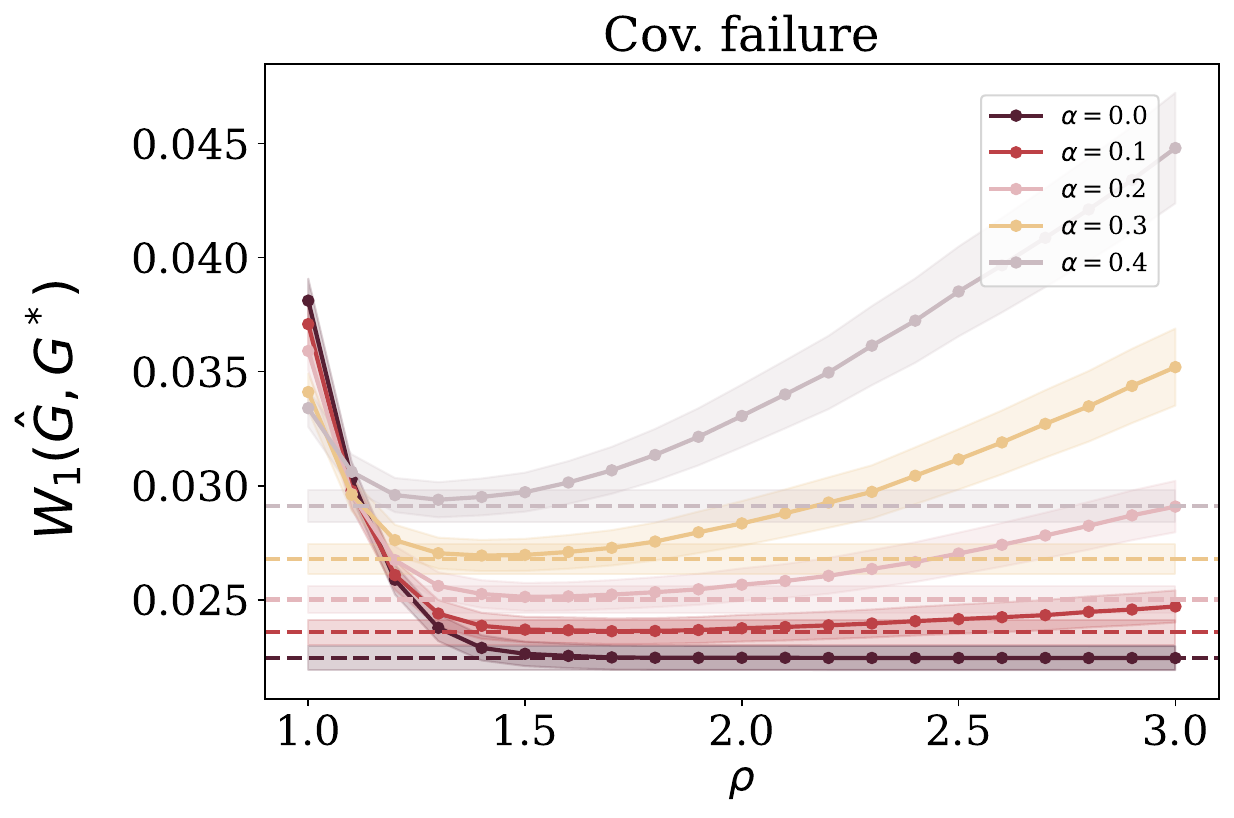}
        \par\vspace{2pt}{\footnotesize (b) Covariance failure}
    \end{minipage}
    \hfill
    \begin{minipage}[t]{0.32\linewidth}
        \centering
        \includegraphics[width=\linewidth]{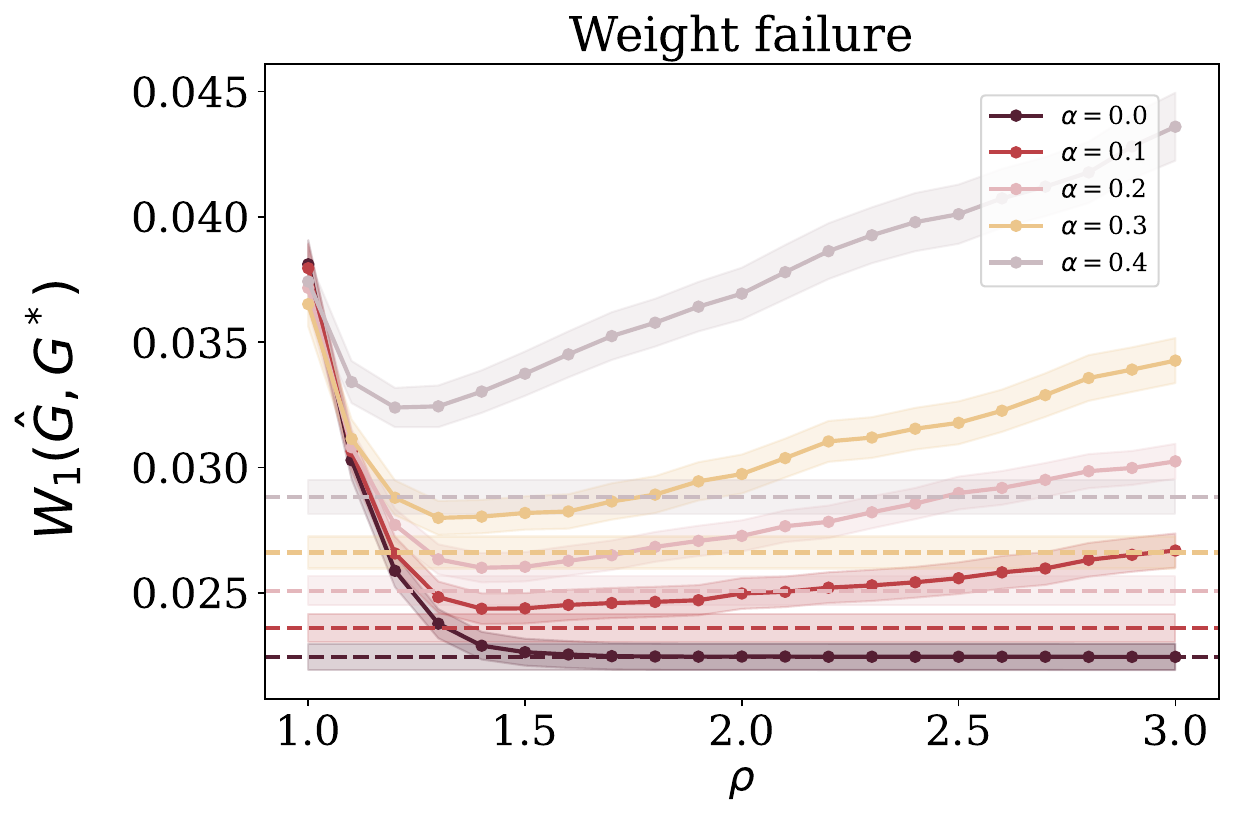}
        \par\vspace{2pt}{\footnotesize (c) Weight failure}
    \end{minipage}
    \caption{$W_1$ of \ours{}$(\rho)$ as $\rho$ varies under each failure type and failure rate, with $m=100$, $n=10{,}000$, and \texttt{MaxOmega}$=0.3$. Dotted lines: \ours{}; dashed lines: C-Oracle.}
    \label{fig:varying_rho_m_100}
\end{figure}

In Figure~\ref{fig:varying_rho_varying_m}, we further examine the effect of $\rho$ as the number of local machines varies.
We fix the total sample size $N=500{,}000$, with \texttt{MaxOmega}$=0.3$, and $\alpha_k=0.2$, and vary $m \in \{20,50,100\}$.
Although our theoretical analysis suggests that the recommended choice of $\rho$ should increase with $m$, the empirical results exhibit a much weaker dependence.
The results suggest that the desirable range of $\rho$ remains similar across different values of $m$, indicating that \ours{} empirically requires little tuning as the system scales.
\begin{figure}[!htbp]
    \centering
    \begin{minipage}[t]{0.32\linewidth}
        \centering
        \includegraphics[width=\linewidth]{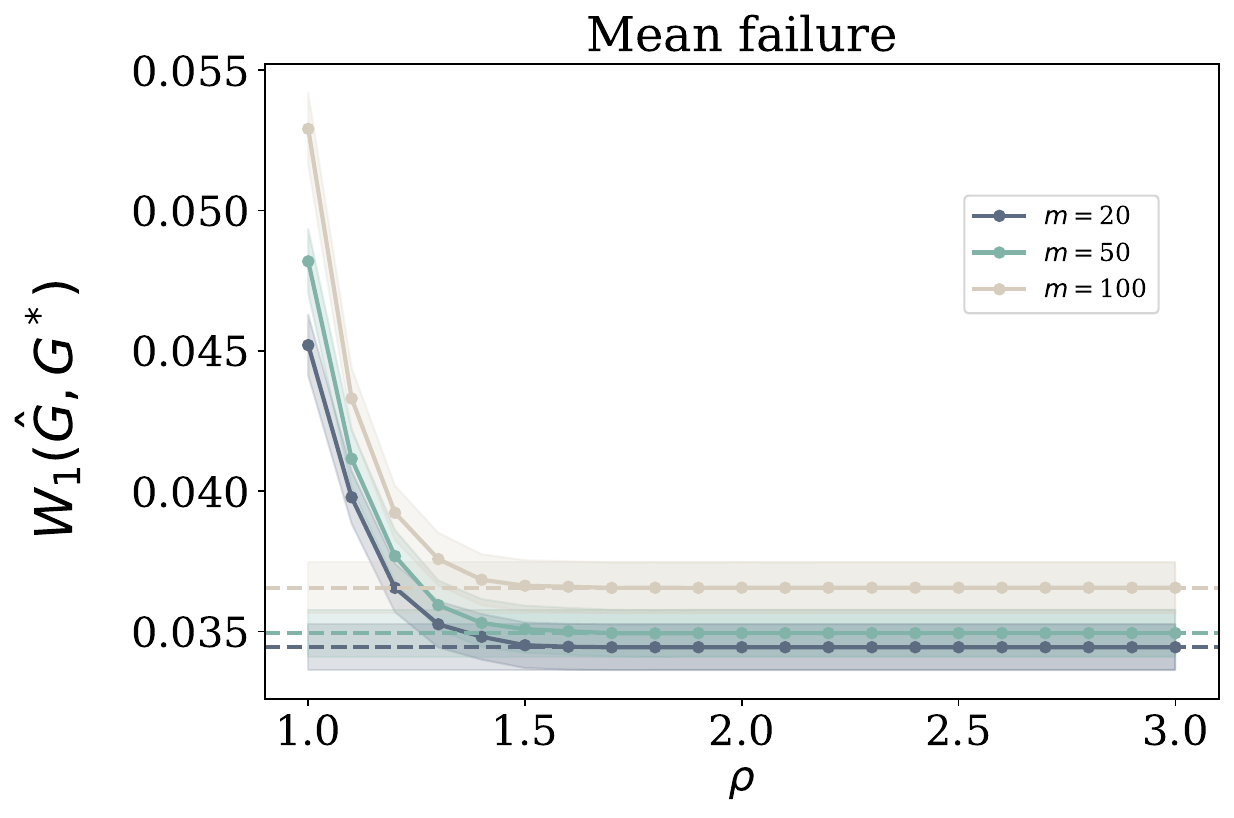}
        \par\vspace{2pt}{\footnotesize (a) Mean failure}
    \end{minipage}
    \hfill
    \begin{minipage}[t]{0.32\linewidth}
        \centering
        \includegraphics[width=\linewidth]{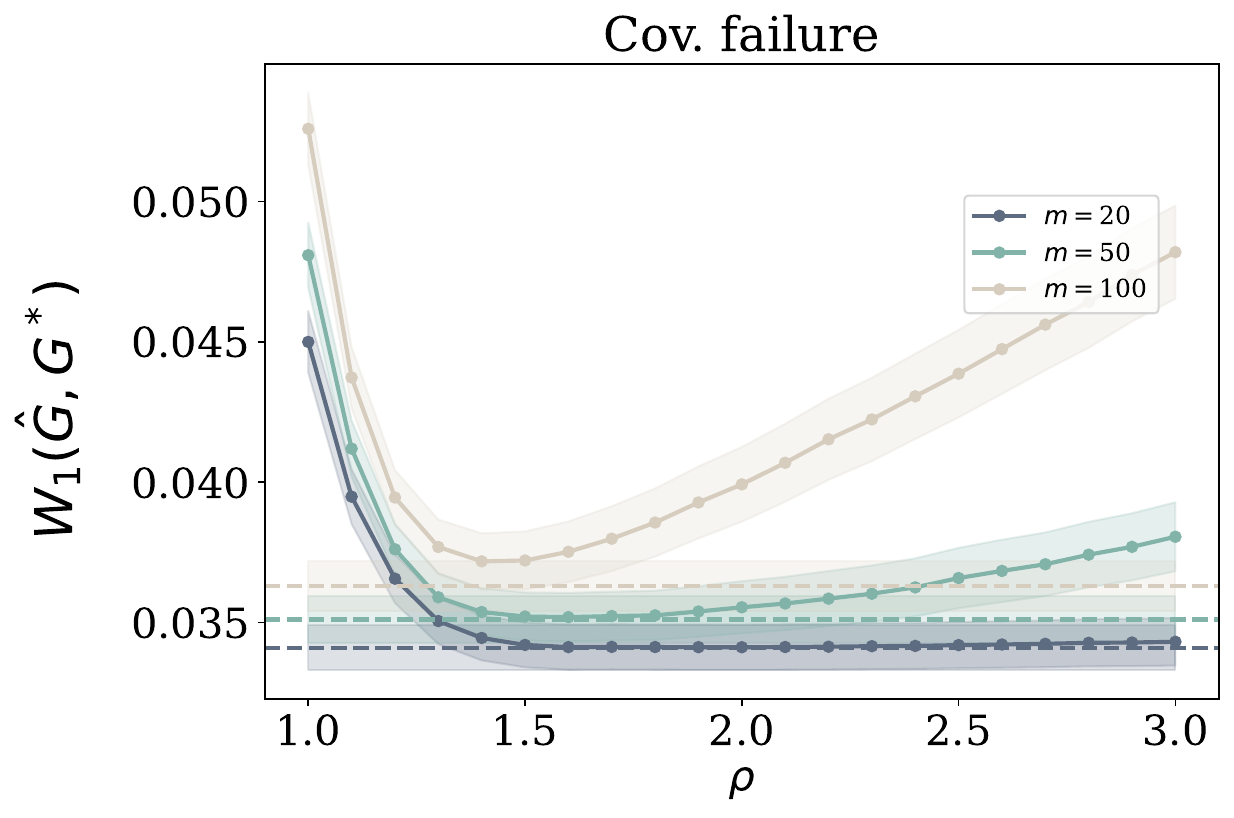}
        \par\vspace{2pt}{\footnotesize (b) Covariance failure}
    \end{minipage}
    \hfill
    \begin{minipage}[t]{0.32\linewidth}
        \centering
        \includegraphics[width=\linewidth]{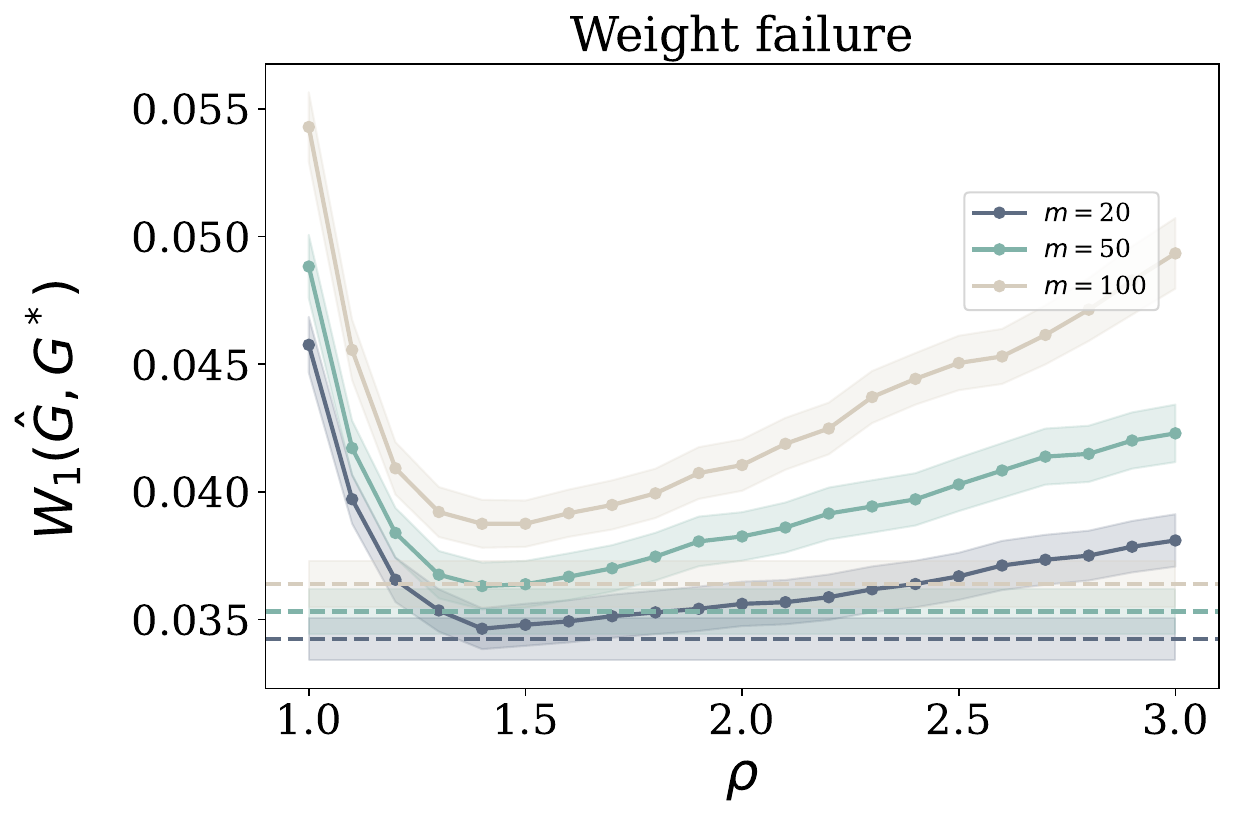}
        \par\vspace{2pt}{\footnotesize (c) Weight failure}
    \end{minipage}
    \caption{$W_1$ of \ours{}$(\rho)$ as $\rho$ varies for $m \in \{20,50,100\}$ under each failure type, with $N=500{,}000$, $\alpha_k=0.2$, and \texttt{MaxOmega}$=0.3$. Dotted lines: \ours{}; dashed lines: C-Oracle.}
    \label{fig:varying_rho_varying_m}
\end{figure}

\subsubsection{Robustness to the trade-off parameter \texorpdfstring{$\lambda$}{lambda}}
Recall from Section~\ref{sec:construction_of_CFMR} that the component discrepancy $c_\lambda(\psi',\psi)$ contains a trade-off parameter $\lambda$, which balances the discrepancy between component parameters and the difference between mixing weights.
Since either the component parameters or the mixing weights may be corrupted, it is important to check that the method is not overly sensitive to this balance.

Figure~\ref{fig:varying_lambda_m_100} examines the effect of $\lambda$ under the setting with $m=100$, $n=10^4$, and \texttt{MaxOmega}$=0.3$.
The component-wise failure rate varies from $0.0$ to $0.4$, while $\lambda$ ranges from $1$ to $100$ in increments of $5$.
Each configuration is replicated $R=100$ times, and the mean Wasserstein distance together with its standard error is recorded.

For mean failure, \ours{} performs comparably to C-Oracle when $\lambda$ lies between $1$ and $40$, with the estimation error remaining nearly constant throughout this range.
As $\lambda$ increases beyond $40$, the estimation error increases slightly, although the deterioration is small relative to the overall error level.
For covariance and weight failures, \ours{} matches the performance of C-Oracle at low failure rates when $10<\lambda<40$.
At higher failure rates, the estimator exhibits a modest degradation relative to C-Oracle but remains stable over a wide range of $\lambda$.
Overall, these results indicate that \ours{} is not sensitive to the choice of the trade-off parameter $\lambda$.
\begin{figure}[!htbp]
    \centering
    \begin{minipage}[t]{0.32\linewidth}
        \centering
        \includegraphics[width=\linewidth]{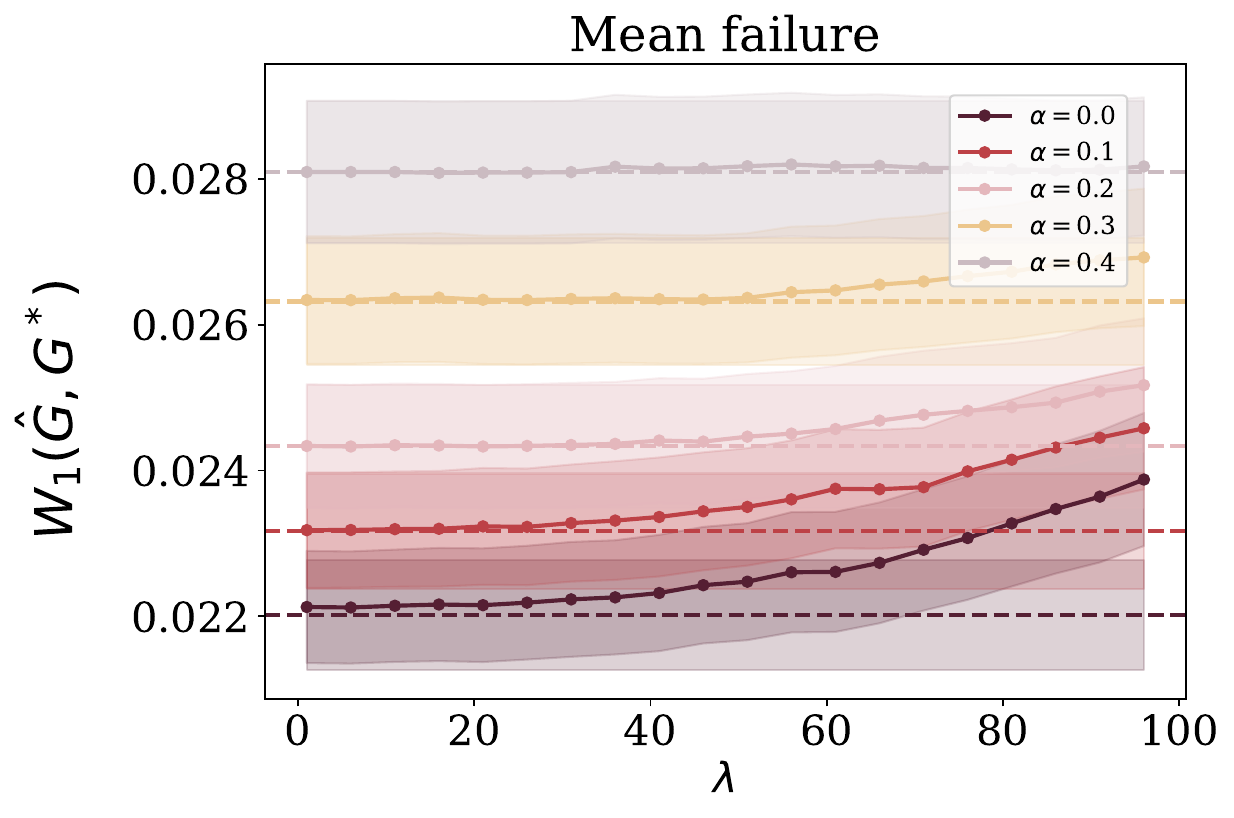}
        \par\vspace{2pt}{\footnotesize (a) Mean failure}
    \end{minipage}
    \hfill
    \begin{minipage}[t]{0.32\linewidth}
        \centering
        \includegraphics[width=\linewidth]{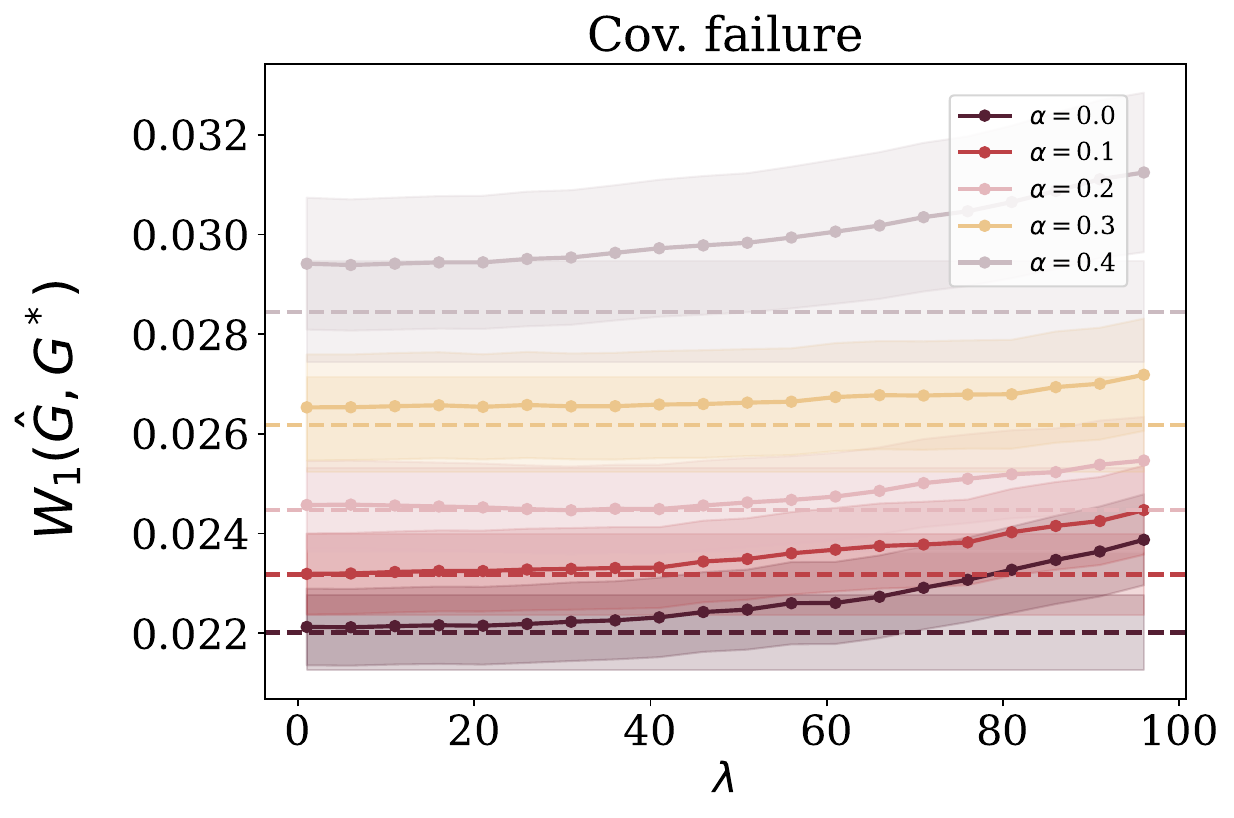}
        \par\vspace{2pt}{\footnotesize (b) Covariance failure}
    \end{minipage}
    \hfill
    \begin{minipage}[t]{0.32\linewidth}
        \centering
        \includegraphics[width=\linewidth]{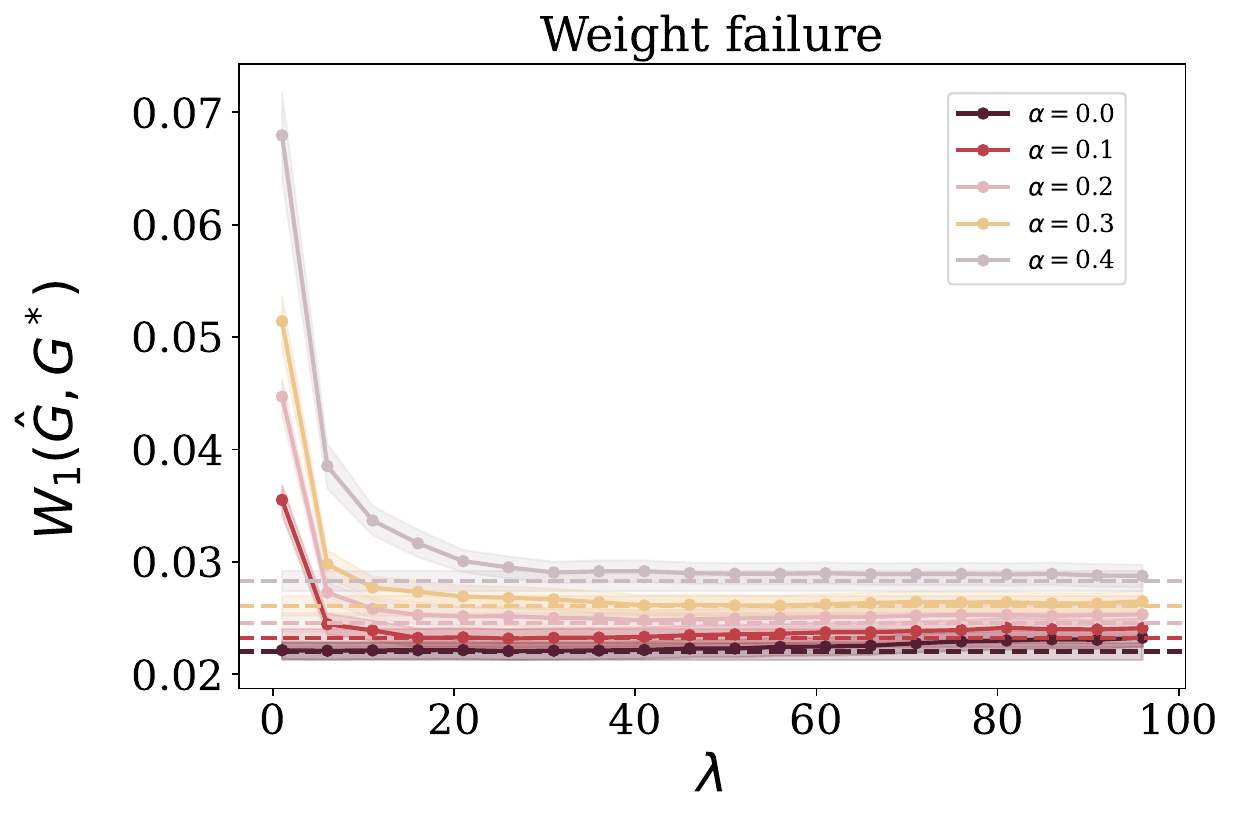}
        \par\vspace{2pt}{\footnotesize (c) Weight failure}
    \end{minipage}
    \caption{$W_1$ of \ours{} as $\lambda$ varies under each failure type and failure rate, with $m=100$, $n=10{,}000$, and \texttt{MaxOmega}$=0.3$. Dotted lines: \ours{}; dashed lines: C-Oracle.}
    \label{fig:varying_lambda_m_100}
\end{figure}
Together, these results demonstrate that \ours{} is stable across the evaluated grids of both $\rho$ and $\lambda$, indicating that the procedure is easy to tune in practice.
Based on the sensitivity analysis, we fix $\rho=1.6$ and $\lambda=10$ in all subsequent experiments.
Note that the fixed choice $\rho=1.6$ serves purely as a practical, finite-sample empirical benchmark, rather than an explicit realization of the asymptotic rate conditions specified in Theorem~\ref{theorem:rate_of_convergence}. 
While the theoretical guarantees require $\rho$ to grow with $m$, the empirical performance exhibits minimal sensitivity to such scaling.
This demonstrates that the theoretical scaling of $\rho$ is inherently conservative in moderate, finite-sample applications.

%%%%%%%%%%%%%%%%%%%%%%%%%%%%%%%%%%%%%%%%%%%%%%%%%%
\subsection{Simulation results}

\subsubsection{Accuracy of data-driven anchor selection}
Before comparing \ours{} with other methods, we first evaluate whether the proposed data-driven anchor selection procedure (Algorithm~\ref{alg:initial_estimate}) successfully identifies a completely failure-free machine.
\begin{wraptable}{r}{0.45\textwidth}
\vspace{-\baselineskip}
\centering
\footnotesize
\begin{minipage}{\linewidth}
\caption{Proportion of the $R=300$ replications in which Algorithm~\ref{alg:initial_estimate} selects a fully failure-free machine, with $m=100$ and $\alpha_k=0.2$ for every component.}
\label{tab:true_selection}
% \vspace{0.5ex}
\centering
\resizebox{0.9\linewidth}{!}{%
\begin{tabular}{lcccc}
\toprule
\multirow{2}{*}{$N$} & \multirow{2}{*}{MaxOmega} & \multicolumn{3}{c}{Failure type} \\
\cmidrule(lr){3-5}
& & Mean & Cov. & Weight \\
\midrule
\multirow{3}{*}{$0.5$M}
& 0.1 & 100.00\% & 100.00\% & 96.67\% \\
& 0.2 & 100.00\% & 100.00\% & 92.67\% \\
& 0.3 & 100.00\% & 97.33\% & 91.33\% \\
\midrule
\multirow{3}{*}{$1$M}
& 0.1 & 100.00\% & 100.00\% & 96.00\% \\
& 0.2 & 100.00\% & 100.00\% & 94.67\% \\
& 0.3 & 100.00\% & 99.67\% & 95.33\% \\
\midrule
\multirow{3}{*}{$2$M}
& 0.1 & 100.00\% & 100.00\% & 98.00\% \\
& 0.2 & 100.00\% & 100.00\% & 97.33\% \\
& 0.3 & 100.00\% & 100.00\% & 97.67\% \\
\bottomrule
\end{tabular}%
}
\end{minipage}
\vspace{-\baselineskip}
\end{wraptable}
Table~\ref{tab:true_selection} summarizes the selection accuracy over $R=300$ Monte Carlo replications with $m=100$ and $\alpha_k=0.2$ for every component.
The proposed procedure selects a fully failure-free machine in every replication under mean failure.
Under covariance and weight failures, the lowest success rates are 97.33\% and 91.33\%, respectively; both occur in the most challenging configuration considered, with the smallest total sample size and the largest component overlap.
Overall, selection accuracy tends to improve as $N$ increases, consistent with the theoretical behavior of the anchor-selection procedure.
At $N=2$M, the success rate exceeds 97\% in every setting.
The remaining errors under covariance and weight failures may reflect the greater difficulty of distinguishing corrupted mixtures: covariance perturbations alter the component geometry, whereas the weight attack can leave the corrupted weights close to their authentic values by preserving their aggregate mass.
Consequently, a partially corrupted mixture may receive a majority-radius score comparable to that of a failure-free machine.
The subsequent experiments suggest that these occasional selection errors have limited influence on the performance of the final estimator in the settings considered.
%%%%%%%%%%%%%%%%%%%%%%%%%%%%%%%%%%%%%%%%%%%%%%%%%%
\subsubsection{Effect of total sample size and Byzantine failure rate}
%%%%%%%%%%%%%%%%%%%%%%%%%%%%%%%%%%%%%%%%%%%%%%%%%%
In all subsequent comparisons, we use the data-selected anchor to align the transmitted component estimates before each aggregation method is applied.
Performance is evaluated by the Wasserstein distance $W_1$ over $R=300$ Monte Carlo replications.
The relative performance of the estimators is consistent across overlap levels, although the advantage of \ours{} over Filtering is most pronounced when \texttt{MaxOmega}$=0.1$.
Accordingly, the main text reports the $W_1$ results for \texttt{MaxOmega}$=0.1$, while the corresponding results for \texttt{MaxOmega}$=0.3$ and all ARI results are provided in Appendices~\ref{app.:maximum_overlap_results} and~\ref{app.:ARI_results}, respectively.

We first examine how the estimators respond to changes in the total sample size $N$ and component-wise failure rate $\alpha_k$.
Figure~\ref{fig:boxplot_m_100_vary_N_vary_alpha} reports their Wasserstein errors with $m=100$ and \texttt{MaxOmega}$=0.1$.
The rows correspond to mean, covariance, and weight failures; the columns correspond to failure rates from $0.0$ to $0.4$; and the colors within each panel distinguish the three total sample sizes.
\begin{figure}[!htbp]
  \centering
  \includegraphics[width=\textwidth]{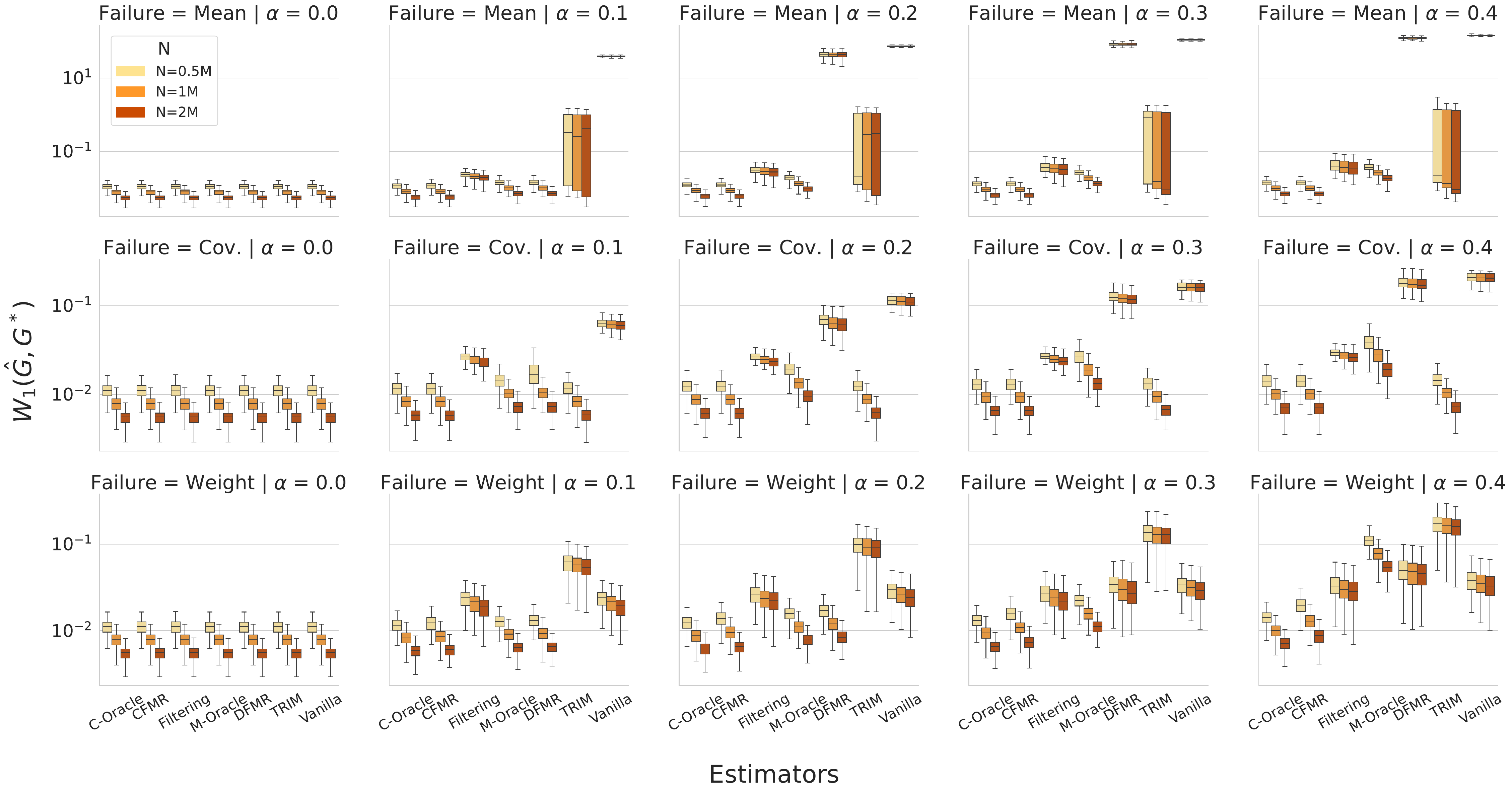}
  \caption{Boxplots of the $W_1$ error for each estimator as the total sample size $N$ and component-wise failure rate $\alpha_k$ vary, with $m=100$ and \texttt{MaxOmega}$=0.1$. Rows correspond to mean, covariance, and weight failures; columns correspond to $\alpha_k\in\{0.0,0.1,0.2,0.3,0.4\}$; and colors distinguish $N\in\{0.5\text{M},1\text{M},2\text{M}\}$.}
  \label{fig:boxplot_m_100_vary_N_vary_alpha}
\end{figure}

Several patterns are common to all three failure types.
First, C-Oracle uniformly outperforms M-Oracle because it retains every failure-free component estimate, whereas M-Oracle discards an entire machine once any of its components is corrupted.
Their performance gap highlights the statistical benefit of preserving component-level information.
Second, \ours{} remains close to C-Oracle across all tested failure rates $\alpha_k\leq0.4$, and its estimation error decreases as $N$ increases, consistent with the convergence behavior established in Theorem~\ref{theorem:rate_of_convergence}.

Among the competing methods, TRIM performs poorly under weight failure because it does not use mixing weights when detecting outlying components.
DFMR deteriorates as $\alpha_k$ increases because it assumes that a majority of the transmitted local mixture estimates remain completely failure-free.
Under component-wise Byzantine failure, however, even a moderate component-wise failure rate may cause a majority of local mixture estimates to contain at least one corrupted component.
For example, when $\alpha_k \ge 0.2$, more than half of the local mixture estimates contain at least one corrupted component, violating the machine-level majority condition required by DFMR and resulting in a substantial loss of accuracy.
Filtering also exhibits a certain degree of robustness, but it relies on the alignment produced by Algorithm~\ref{alg:initial_estimate}. 
Even with this assistance, its overall performance remains inferior to that of \ours{}.

%%%%%%%%%%%%%%%%%%%%%%%%%%%%%%%%%%%%%%%%%%%%%%%%%%
\subsubsection{Performance when the number of local machines scales up}
%%%%%%%%%%%%%%%%%%%%%%%%%%%%%%%%%%%%%%%%%%%%%%%%%%
We next investigate how the estimators scale with the number of local machines $m$ under \texttt{MaxOmega}$=0.1$ and $\alpha_k=0.2$.
To separate the benefit of an increasing total sample size from the effect of partitioning a fixed data set across more local machines, we consider two representative regimes.

First, we fix the local sample size at $n=5000$, so that the total sample size $N=nm$ increases with $m$.
Figure~\ref{fig:boxplot_vary_m_fix_n} shows that the error of \ours{} decreases as more machines are added and remains close to that of C-Oracle throughout.
This indicates that \ours{} effectively utilizes the additional failure-free component estimates contributed by the increasing number of machines.

Second, we fix the total sample size at $N=500{,}000$, so that increasing $m$ reduces the local sample size to $n=N/m$.
As shown in Figure~\ref{fig:boxplot_vary_m_fix_N}, the performance of \ours{} remains stable over $m\in\{20,50,100\}$ and continues to closely track C-Oracle.
Thus, within the range considered, distributing a fixed sample over more machines does not lead to an appreciable loss in accuracy for \ours{}.
By comparison, Filtering becomes less accurate as $m$ increases under mean and covariance failures in this fixed-$N$ regime.

\begin{figure}[!htbp]
  \centering
  \includegraphics[width=0.85\textwidth]{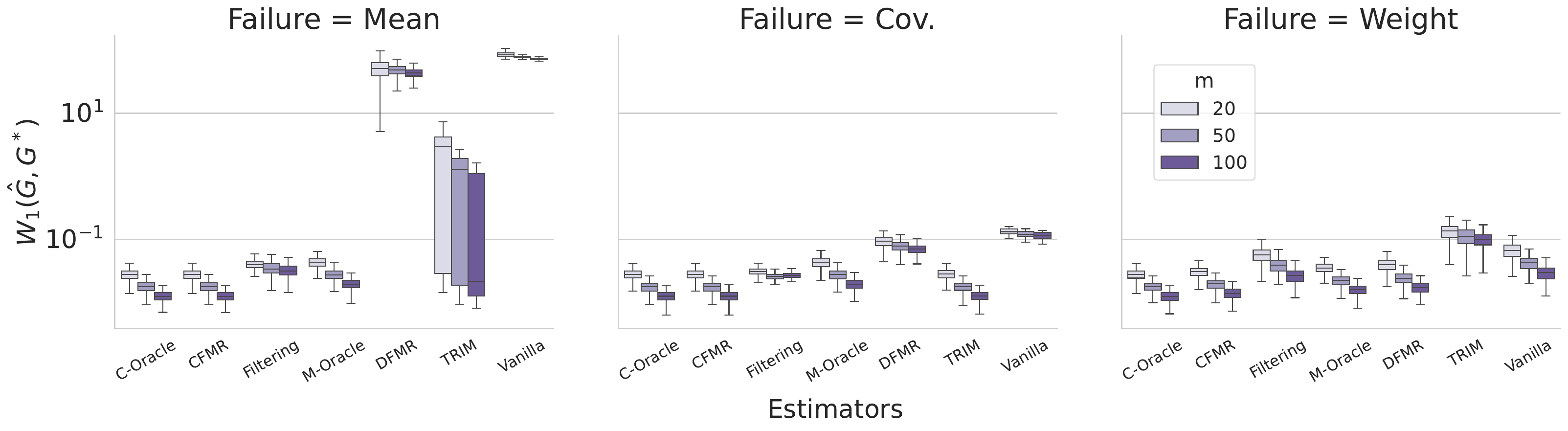}
  \caption{Boxplots of the $W_1$ error for each estimator as the number of machines varies over $m\in\{20,50,100\}$, with the local sample size fixed at $n=5000$, \texttt{MaxOmega}$=0.1$, and $\alpha_k=0.2$. Thus, the total sample size $N=nm$ increases with $m$. Panels correspond to mean, covariance, and weight failures, and colors distinguish the values of $m$.}
  \label{fig:boxplot_vary_m_fix_n}
\end{figure}

\begin{figure}[!htbp]
  \centering
  \includegraphics[width=0.85\textwidth]{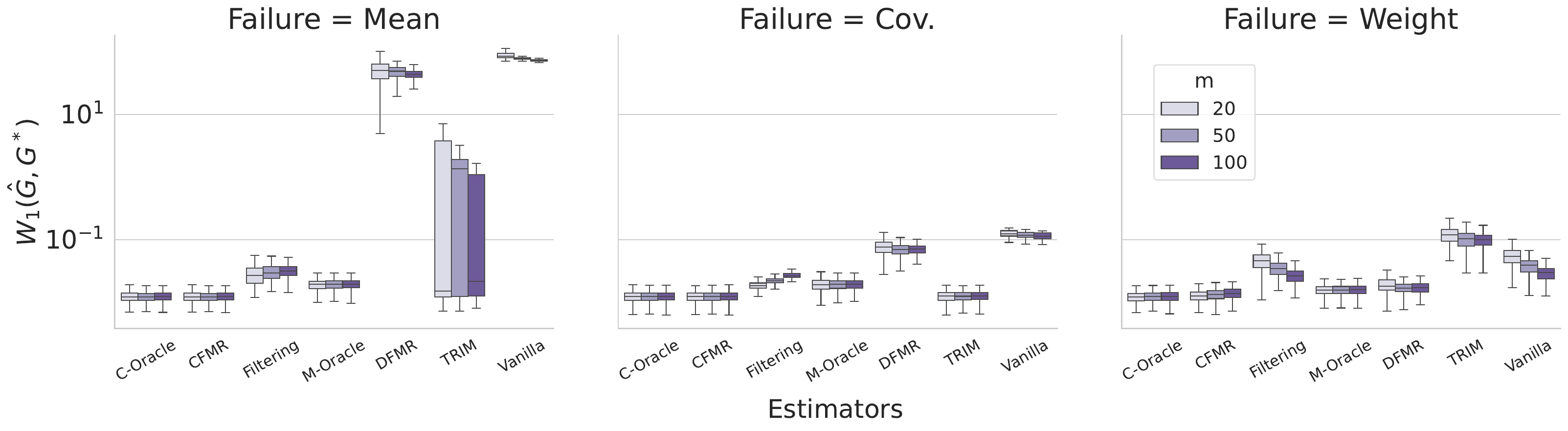}
  \caption{Boxplots of the $W_1$ error for each estimator as the number of machines varies over $m\in\{20,50,100\}$, with the total sample size fixed at $N=500{,}000$, \texttt{MaxOmega}$=0.1$, and $\alpha_k=0.2$. Thus, the local sample size $n=N/m$ decreases with $m$. Panels correspond to mean, covariance, and weight failures, and colors distinguish the values of $m$.}
  \label{fig:boxplot_vary_m_fix_N}
\end{figure}

%%%%%%%%%%%%%%%%%%%%%%%%%%%%%%%%%%%%%%%%%%%%%%%%%%
\subsubsection{Performance with different degrees of overlap}
%%%%%%%%%%%%%%%%%%%%%%%%%%%%%%%%%%%%%%%%%%%%%%%%%%
Finally, we investigate the effect of component overlap on the performance of the competing methods.
The degree of overlap directly affects the difficulty of estimating and distinguishing the mixture components.
Figure~\ref{fig:boxplot_m_100_vary_overlap_fix_alpha} reports the results for \texttt{MaxOmega}$\in\{0.1,0.2,0.3\}$ with $\alpha_k=0.2$, $m=100$, and $n=5000$.
As expected, the estimation errors generally increase with the degree of overlap.
Nevertheless, \ours{} remains close to C-Oracle across all three failure types and overlap levels, indicating that its component-wise alignment and aggregation remain effective even when the mixture components become less distinguishable.
Although the performance gap between \ours{} and Filtering narrows at higher overlap levels, \ours{} remains at least as accurate.
By comparison, DFMR, TRIM, and Vanilla continue to exhibit substantial degradation under one or more failure types.
\begin{figure}[!htbp]
  \centering
  \includegraphics[width=0.85\textwidth]{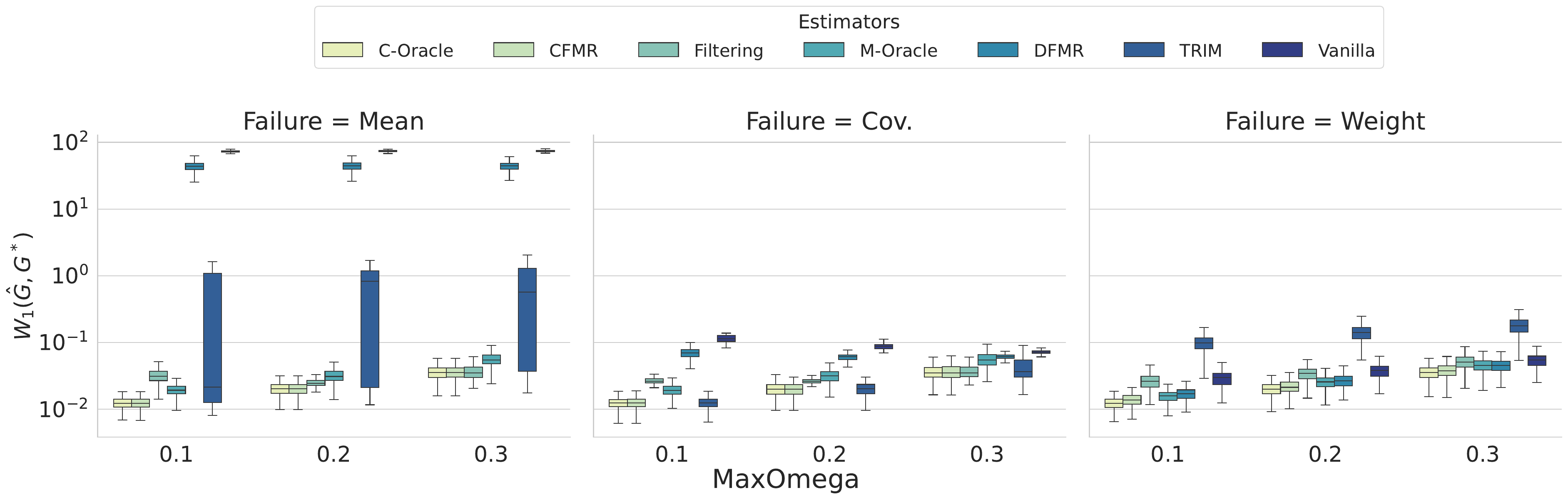}
  \caption{Boxplots of the $W_1$ error for each estimator as the degree of component overlap varies over \texttt{MaxOmega}$\in\{0.1,0.2,0.3\}$, with $m=100$, $n=5000$, and $\alpha_k=0.2$. Panels correspond to mean, covariance, and weight failures.}
  \label{fig:boxplot_m_100_vary_overlap_fix_alpha}
\end{figure}

%% The following is a directive for TeXShop to indicate the main file
%%!TEX root = ../main.tex

\section{Application to Real-World Data}
\label{sec:real_data}
\begin{wrapfigure}{r}{0.48\textwidth}
\centering
% \footnotesize
\includegraphics[width=\linewidth]{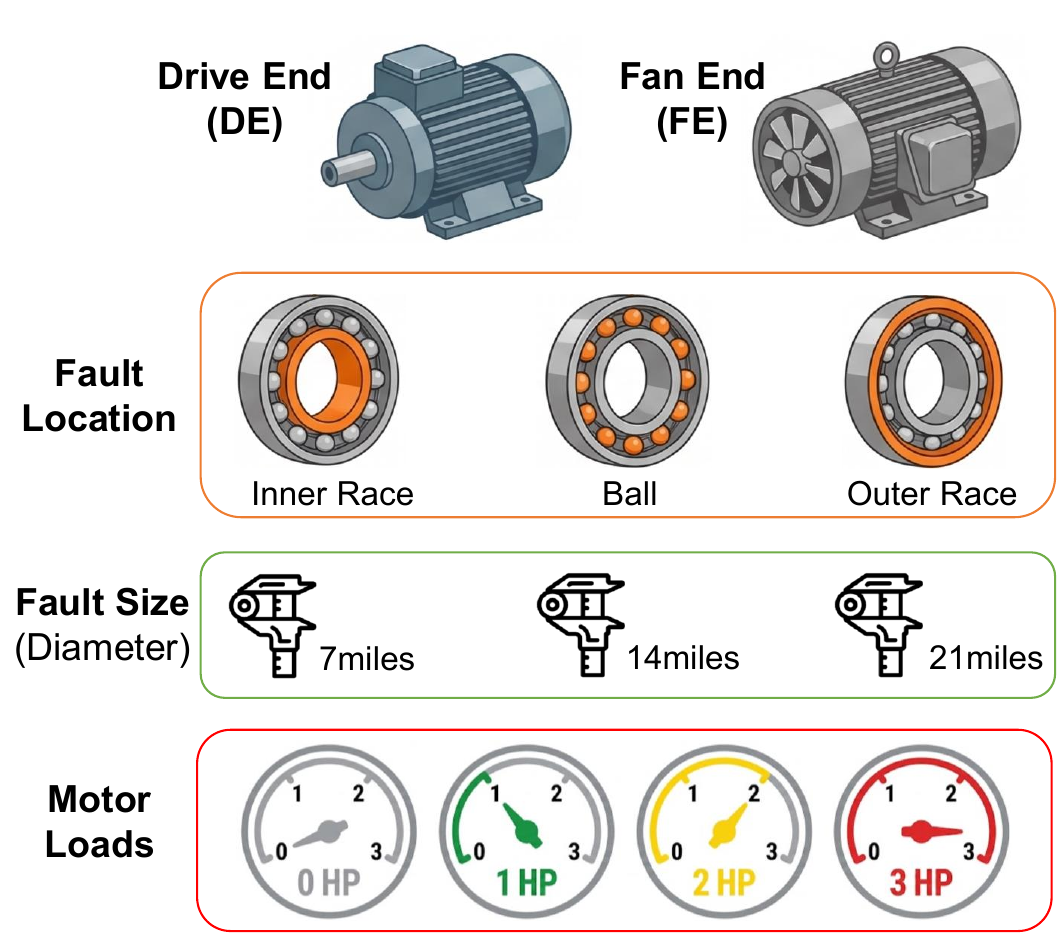}
\caption{Overview of the CWRU dataset experimental configurations.}
\label{fig:CWRU_fault_illustration}
\vspace{-\baselineskip}
\end{wrapfigure}

We evaluate the proposed method on a bearing fault identification task.
Our analysis uses the benchmark Case Western Reserve University (CWRU) bearing dataset~\cite{cwru_bearing_data}, which has been widely adopted for evaluating statistical and machine learning methods in machinery diagnostics~\cite{smith2015rolling}.

The CWRU dataset contains vibration signals collected from bearings with faults introduced by electro-discharge machining (EDM).
The experimental conditions vary along four main dimensions:
\begin{itemize}
    \item \textbf{Fault location:} Defects are introduced on the inner race, outer race, or rolling element (ball) of the bearing.
    \item \textbf{Fault severity:} The severity of each fault is characterized by a defect diameter of 7, 14, or 21 mils.
    \item \textbf{Spatial orientation:} Outer-race faults are placed at three positions relative to the load zone to capture spatial heterogeneity.
    \item \textbf{Measurement setting:} Faulty bearings are installed at either the Drive End (DE) or Fan End (FE), and vibration signals are recorded at 12 or 48 kHz.
\end{itemize}
For each fault configuration, the system is operated under four motor loads, namely 0, 1, 2, and 3 HP, and each raw signal spans approximately 10 seconds.

We formulate a nine-class clustering problem by treating each combination of fault location (inner race, ball, or outer race) and fault severity (7, 14, or 21 mils) as a distinct fault class.
We use the DE vibration signals recorded at 48 kHz under a motor load of 2 HP, for which each target class contains between 485,643 and 491,446 raw observations.
Representative time-domain signals are shown in Figure~\ref{fig:CWRU_data_processing}(a).
Their fluctuation patterns differ visibly across fault configurations, reflecting variation in the underlying physical mechanisms.

We next transform the raw signals into a feature representation suitable for Gaussian mixture modeling.
To prevent overlapping windows from creating leakage between training and testing data, the continuous signal from each fault class is first divided into a training segment ($70\%$) and a testing segment ($30\%$).
A sliding window of length $L=512$ and overlap ratio $75\%$ is then applied separately to the two segments.
For each window, we compute the Fast Fourier Transform (FFT), retain the frequency-domain magnitudes, and standardize the resulting features.
A one-dimensional convolutional neural network (1D-CNN) subsequently maps the normalized FFT features into a compact $10$-dimensional embedding space.
Further details of the architecture of the neural network are provided in Appendix~\ref{app:real_data}.
Figure~\ref{fig:CWRU_data_processing}(b) presents a t-SNE visualization of a random subset of the learned embeddings.
The embeddings exhibit a clear cluster structure across the nine fault classes, supporting the use of a Gaussian mixture model with $K=9$ components.

For each replication, we randomly sample $N=20{,}000$ training embeddings from the DE training pool and distribute them evenly across $m=20$ machines, where each machine fits a local Gaussian mixture model with $K = 9$ components. Prior to local model fitting, we obtain a baseline mixture estimate on the failure-free machine $m$ and use it as a shared initialization for the local EM algorithm on all local machines. This initialization aligns the component labels across machines by construction, ensuring that the $k$th corrupted component consistently corresponds to the same true mixture component; this alignment is used solely to generate the controlled component-wise corruptions and is withheld from all aggregation methods. 
We then introduce controlled component-wise Byzantine failures.
For each component $k\in[K]$, we randomly select $\alpha m$ corrupted machines from the first $m-1$ machines, leaving machine $m$ fully failure-free.
To construct the corrupted component estimates, 
we process the FE signals through the same feature pipeline and extract consecutive, non-overlapping blocks of 30 observations sequentially for each corrupted component occurrence; the authentic local mean and covariance matrix are then replaced by the empirical mean and covariance computed from the corresponding block.
Because the replacement components arise from the same bearing system under a different sensing condition, this design produces plausible but systematically mismatched local estimates rather than easily detectable artificial outliers.
\begin{figure}[htbp]
    \centering
    \begin{minipage}[t]{0.49\textwidth}
        \centering
        \includegraphics[width=0.9\linewidth]{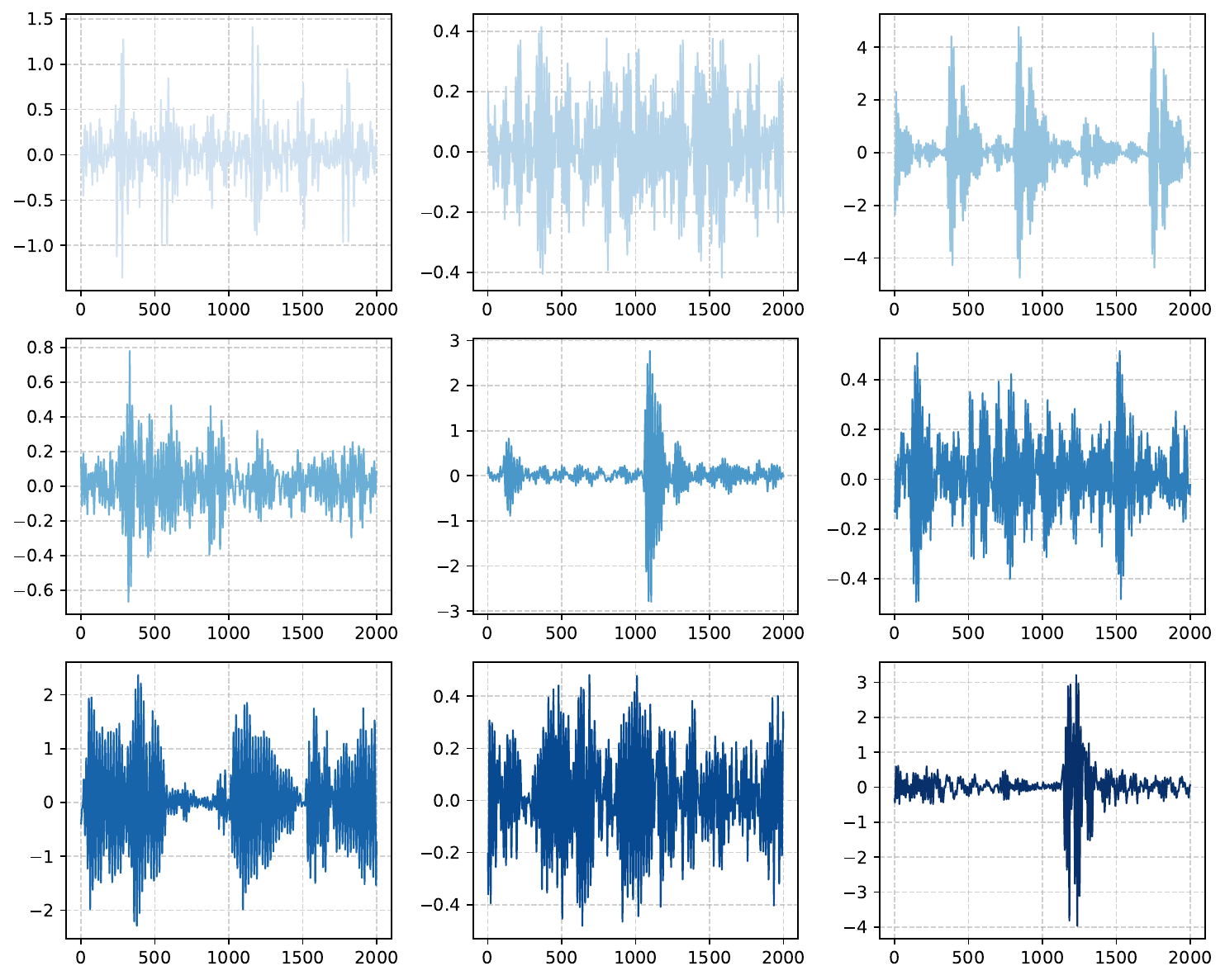}
        \par\smallskip\textbf{(a)} Representative raw time-domain vibration signals.
    \end{minipage}
    \begin{minipage}[t]{0.49\textwidth}
        \centering
        \includegraphics[width=0.9\linewidth]{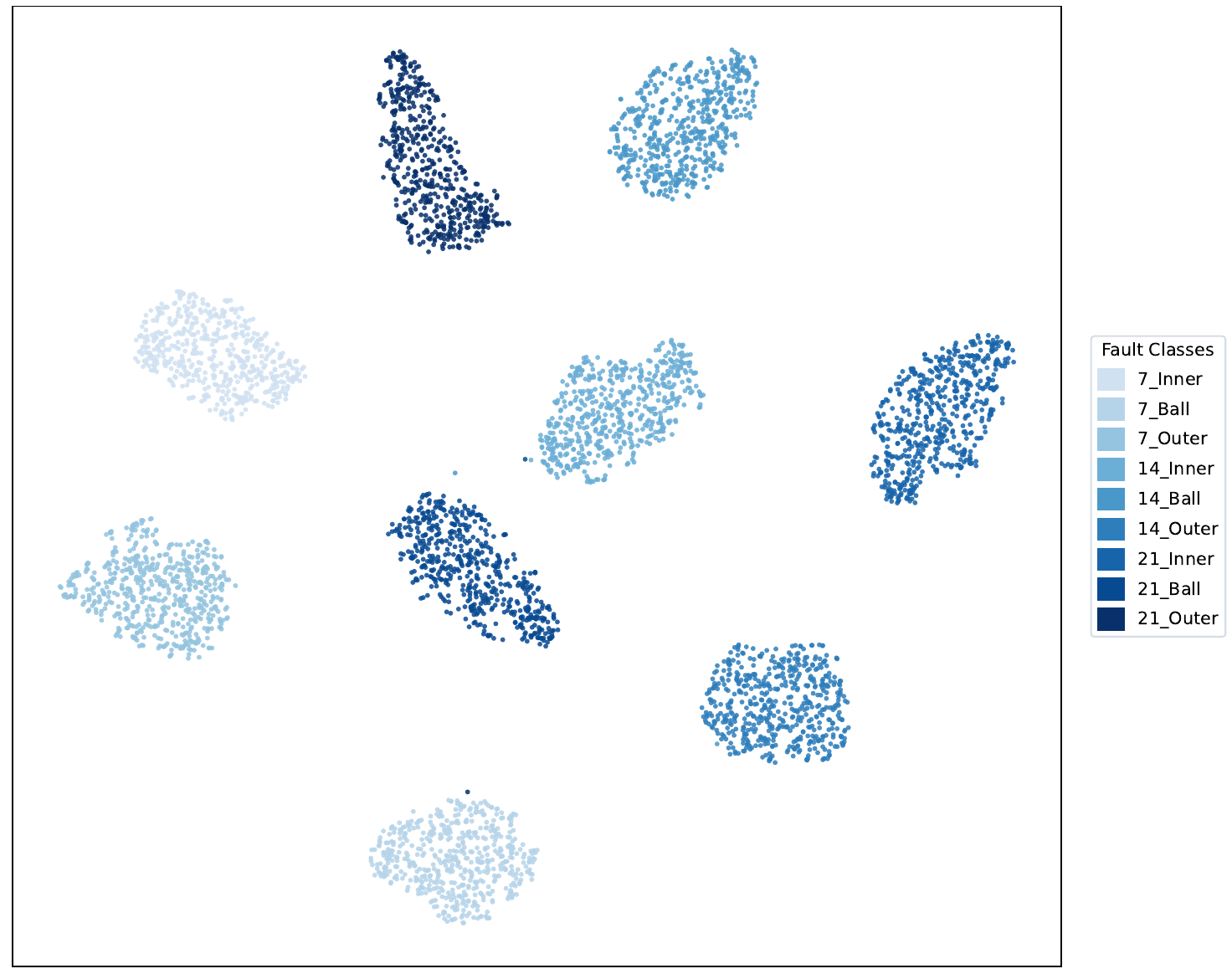}
        \par\smallskip\textbf{(b)} t-SNE visualization of the learned $10$-dimensional embeddings.
    \end{minipage}
    \caption{CWRU bearing data and their learned feature representation for the nine fault classes defined by fault location and severity: \textbf{(a)} raw signals and \textbf{(b)} learned embeddings.}
    \label{fig:CWRU_data_processing}
\end{figure}

We vary the component-wise failure rate $\alpha$ from $0$ to $0.4$ in increments of $0.1$.
Each aggregation method combines the transmitted local mixtures, and the resulting mixture is used to assign the held-out DE test embeddings to fault clusters via maximum a posteriori probability.
We repeat the complete procedure over $R=300$ independent replications and evaluate clustering accuracy using the Adjusted Rand Index (ARI) between the true fault classes and the estimated cluster assignments.
Table~\ref{tab:CWRU_results} reports the resulting median and interquartile range (IQR). Here, the reported IQRs specifically quantify the variability induced by the random training partitions and realized componnet-wise corruption patterns, conditional on the fixed feature representation and evaluation set.

\begin{table}[htbp]
\centering
\caption{Median (IQR) of the ARI on the CWRU test set over $R=300$ replications at each component-wise failure rate $\alpha$.}
\begin{tabular}{ccccccc}
\toprule
$\alpha$ &Vanilla & C-Oracle & \ours{} & Filtering & TRIM & DFMR \\
\midrule
0.0 & 0.9841(0.0020)& 0.9841 (0.0020) & 0.9839 (0.0022) & 0.9830 (0.0035) & 0.9832 (0.0059) & 0.9839 (0.0020) \\
0.1 & 0.9701(0.1034)& 0.9840 (0.0022) & 0.9839 (0.0020) & 0.9828 (0.0033) & 0.9821 (0.0082) & 0.9723 (0.1045) \\
0.2 &0.8713(0.1018)& 0.9841 (0.0022) & 0.9841 (0.0020) & 0.9828 (0.0035) & 0.9825 (0.0129) & 0.8723 (0.1024)  \\
0.3 & 0.8684(0.1051) & 0.9841 (0.0026) & 0.9842 (0.0026) & 0.9829 (0.0036) & 0.9835 (0.0083) & 0.8690 (0.1050) \\
0.4 & 0.8629(0.1964) & 0.9841 (0.0026) & 0.9843 (0.0026) & 0.9832 (0.0037) & 0.9843 (0.0045) & 0.8620 (0.1955)\\
\bottomrule
\end{tabular}
\label{tab:CWRU_results}
\end{table}

At $\alpha=0$, all methods attain high ARI values.
Once component-wise failures are introduced, however, the performance of Vanilla and DFMR deteriorates sharply, whereas \ours{}, C-Oracle, Filtering, and TRIM retain high median ARI values.
The poor performance of DFMR is consistent with its machine-level construction: discarding or retaining an entire local mixture is too coarse when corruption affects only selected components.

Across the nonzero failure rates, \ours{} remains highly competitive with the strongest baselines and achieves the highest  or tied-highest median ARI across all aggregation methods when $\alpha \ge 0.2$.
Notably, \ours{} slightly outperforms even C-Oracle under severe contamination regimes ($\alpha \ge 0.3$).
A possible explanation is that the component-wise filtering step removes not only corrupted component estimates but also unusually noisy failure-free estimates that lie far from the dominant component cluster.
Thus, beyond providing Byzantine robustness, the proposed filtering mechanism may also improve finite-sample stability by excluding poorly estimated local components.

%% The following is a directive for TeXShop to indicate the main file
%%!TEX root = ../main.tex

\section{Conclusion and Discussion}

This paper develops a Byzantine-tolerant aggregation procedure for split-and-conquer learning of finite mixture models under component-wise failures.
The proposed Component-wise Filtered Mixture Reduction (\ours{}) is 
designed for settings in which a local mixture estimate is only partially corrupted.
The method is computationally simple and communication-efficient: it uses the transmitted local mixtures only once, selects an anchor through a majority-radius criterion, aligns all transmitted components to this anchor, filters within each aligned component cluster, and then applies the MR barycentric update.
Our theoretical results show that, under component-wise authentic majorities,
at least one fully authentic local mixture, and suitable regularity,
separation, and capacity conditions, \ours{} matches the convergence rate of
component-level oracle aggregation while allowing the proportion of partially
corrupted machines to be large.  The adaptive remainder depends on the
realized corrupted mass surviving the filter.  The main theorem separates an
exact-retention regime, which keeps all authentic components asymptotically,
from a tail-trimming regime that carries the contribution of authentic tail
losses explicitly while permitting a narrower filter.

It is useful to place this contribution in the broader context of robust distributed mixture estimation.
For ordinary MR without Byzantine failure, the main role of the aggregation step is to resolve label switching and recover centralized first-order behavior.
For machine-level Byzantine failure, DFMR achieves robustness by filtering whole local mixtures, which is effective when a strict majority of machines are entirely failure-free.
The setting considered here is fundamentally different.
When corruption occurs component by component, filtering whole machines is too coarse: it may either retain corrupted components or discard authentic ones that could otherwise improve accuracy.
Thus, the main challenge is to combine label alignment and robust filtering at the component level, even though the server does not know the oracle ordering of the local components.
The anchor-based construction of \ours{} provides one way to overcome this difficulty.
Technically, this requires controlling several linked events: the selected anchor must be failure-free, failure-free components must be assigned to their correct anchor-induced clusters, the empirical majority radii must have the clean $n^{-1}$ cost scale, and the filter must retain almost all authentic component estimates while localizing any corrupted estimates that remain.
These ingredients allow the final convergence rate to separate the exact
oracle error $(m_kn)^{-1/2}+n^{-1}$ from the residual effect of corrupted
component estimates near the selected reference mixture.  The moment order in
the likelihood smoothness condition, and hence of the derived local-estimator
bound, controls how rapidly the filtering inflation must grow,
but does not alter the clean statistical scale or the structural anchoring
requirements.

The numerical experiments support the theoretical findings.
In Gaussian mixture simulations, the anchor-selection step succeeds with high probability across the configurations considered, and the resulting estimator remains close to the component-level oracle across different failure types, failure rates, sample sizes, numbers of machines, and overlap levels.
The real-data analysis on bearing fault diagnosis shows the same qualitative behavior: \ours{} remains stable as the component-wise failure rate increases, whereas methods based on whole-machine filtering or unprotected aggregation can deteriorate substantially.
These results demonstrate the practical advantage of exploiting partially reliable local mixture estimates rather than treating each machine as entirely reliable or entirely corrupted.

Several limitations suggest directions for future research.
First, the current theory assumes the existence of at least one fully failure-free local mixture estimate.
This assumption provides a transparent mechanism for resolving label switching, but it is stronger than the component-wise majority condition alone.
It would be interesting to develop anchor-free methods under alternative separation or clustering conditions that allow the correct component groups to be recovered directly from all transmitted components.
Second, the number of mixture components is treated as known and the true mixing distribution is fixed with separated components.
Allowing the number of components to be unknown, overfitted, or selected from data would require additional theory, and regimes in which components become less separated as the sample size grows may require techniques beyond the local normal approximation and cost-radius concentration arguments used here.
Third, the paper considers only a single-round communication protocol.
The aggregated estimator could naturally be used as an initialization for further distributed EM iterations, but understanding the robustness, communication cost, and statistical gain of such multi-round refinements remains open.
Fourth, the present adversarial theory treats the component-specific failure sets as fixed in advance.
The corrupted replacement values may otherwise be coordinated and adaptive,
subject to the stated capacity conditions.  The small-ball corollary is a
separate stochastic-contamination result, while a genuinely data-dependent
choice of which latent component occurrences to corrupt would require new
conditional arguments and is left for future work.
Finally, the tuning parameters in the component cost and filtering radius are guided by theory but chosen empirically in finite samples.
Developing adaptive tuning rules with theoretical guarantees would further improve the practical applicability of component-wise Byzantine-tolerant aggregation.

\section*{Acknowledgments}
Yimei Zhang and Qiong Zhang are supported by the National Natural Science Foundation of China under Grant 12301391 and by the Open Research Fund of the Key Laboratory of Advanced Theory and Application in Statistics and Data Science (East China Normal University), Ministry of Education.
Yan Shuo Tan is supported by NUS Start-up Grant A-8000448-00-00 and MOE AcRF Tier 1 Grants A-8002498-00-00 and A-8004458-00-00.
This research was enabled in part by support provided by Fir (fir.alliancecan.ca) and the Digital Research Alliance of Canada (alliancecan.ca).

AI Disclosure: Illustrative icons in Figure~\ref{fig:CWRU_fault_illustration} were generated using Google Gemini 3.5 Flash. 
The AI-generated icons were used solely for visualization purposes and do not affect the scientific content, experimental results, or conclusions of this work.
OpenAI's ChatGPT 5.6 was used to assist with editing and improving the clarity of the text and proofs.

%Bibliography
\bibliographystyle{IEEEtran}  
\bibliography{references}  

% \clearpage
% \newpage
\appendices
%% The following is a directive for TeXShop to indicate the main file
%%!TEX root = ../main.tex

\section{Numerical algorithms}
\subsection{EM algorithm}
\label{app:em_algorithm}
The Maximum Likelihood Estimate (MLE) is widely used for local inference under finite mixture models. 
The preferred numerical method for computing the MLE is the Expectation-Maximization (EM) algorithm, which is most conveniently described using the latent variable interpretation of the mixture model.

For the $j$th unit with an observed value $x_{ij}$ from the finite mixture $f_{G}(x)$ on the $i$th machine, there exists a latent variable $z_{ij}$ associated with the observed value. 
When the latent variables $\{z_{ij}, j \in [n]\}$ are known, the complete dataset $\{(z_{ij}, x_{ij})\}$ is obtained. 
This leads to the complete data log-likelihood function:
\[
\ell_{n}^c(G) 
 = \sum_{j=1}^{n} 
 \sum_{k=1}^K \mathbbm{1}(z_{ij}=k)\log\{w_kf(x_{ij}; \theta_k)\}.
 \]
However, the complete data log-likelihood cannot be directly used to estimate $G$ since the latent variables are unknown. 
In the EM algorithm, one replaces $z_{ij}$ in the complete data log-likelihood by its conditional expectation in each iteration.
Let $G^{(t)}$ denote the value of the mixing distribution after the $t$th iteration. 
The conditional probability of $z_{ij}$ given the observed data is:
\begin{equation}
\label{eq:conditional_expectation}
w_{ijk}^{(t)} 
= 
\sP(z_{ij}=k| G^{(t)}, \gX)
= 
\dfrac{w_k^{(t)} f(x_{ij};\theta_k^{(t)})}
{\sum_{k'=1}^{K} w_{k'}^{(t)} f(x_{ij}; \theta_{k'}^{(t)})}.
\end{equation}
The conditional expectation of $\ell_{n}^c(G) $ is then given by:
\begin{equation*}
Q(G|G^{(t)}) 
= 
\sum_{j=1}^{n} \sum_{k=1}^K w_{ijk}^{(t)}\log \{w_kf(x_{ij}; \theta_k)\}.
\end{equation*}
This constitutes the E-step of the algorithm.

In the M-step, instead of maximizing the complete data log-likelihood $\ell_{n}^c(G)$, the algorithm seeks the maximizer of $Q(G|G^{(t)})$ for $G \in \sG_K$. 
This task is simpler because $Q(G |G^{(t)})$ is additive in subpopulation parameters which allows for easier numerical computation.
As a result, the updated mixing distribution $G^{(t+1)}$ is composed of mixing weights:
\begin{equation*}
w_{k}^{(t+1)} =   n^{-1} \sum_{j=1}^{n} w_{ijk}^{(t)}
\end{equation*}
and subpopulation parameters: 
\begin{equation*}
\theta_{k}^{(t+1)} = \argmax_{\theta} \left\{\sum_{j=1}^{n} w_{ijk}^{(t)} \log f(x_{ij};\theta)\right\}.
\end{equation*}
For many parametric subpopulation models $\gF$, there exists an analytical solution for $\theta_{k}^{(t+1)}$, making the EM iteration straightforward to execute. 
Repeating the iteration generates a sequence of mixing distributions. 
Under certain conditions, Wu~\cite{wu1983convergence} show that $\ell_{n}(G^{(t)}) = \sum_{j=1}^{n} \log f_{G^{(t)}}(x_{ij})$ is an increasing sequence, and the sequence $\{G^{(t)}, t=1,2,\ldots\}$ converges to a local maximum of the log-likelihood function.

Despite its widespread use, the EM algorithm suffers from a slow convergence rate and can easily become trapped in local maxima. 
Various approaches have been proposed to accelerate convergence, see~\cite{meilijson1989fast, meng1993maximum, liu1994ecme, balakrishnan2017statistical}.

\paragraph{Issues with the MLE under finite Gaussian mixtures and finite location-scale mixtures}
The MLE is not well-defined for widely used finite Gaussian and location-scale mixtures in general. 
To address this issue, several approaches have been proposed. 
For example, Hathaway~\cite{hathaway1985constrained} introduces a constrained maximum likelihood formulation to avoid the unboundedness of the likelihood function, resulting in an estimator with desirable consistency properties. 
However, this approach alters the parameter space, which may be undesirable. 
Alternatively, Ridolfi and Idier~\cite{ridolfi2001penalized} propose a Bayesian approach by imposing an inverse gamma prior on the scale parameter in Gaussian mixture model. 
The posterior mode, known as the maximum a posterior (MAP) estimator, is later shown to be consistent by Chen et al.~\cite{chen2008inference}.

In Chen et al.~\cite{chen2008inference}, a penalized log-likelihood function is defined as:
\begin{equation*}
p\ell_n(G) = \ell_n(G) - a_n\sum_{k=1}^K p(\theta_k)
\end{equation*}
where $p(\cdot)$ is a penalty function and $a_n > 0$ is the penalty strength. 
The penalized MLE (pMLE) is then:
\begin{equation*}
\widehat{G}^{\text{pMLE}}= \argsup_{G\in\sG_{K}} p\ell_n(G)
\end{equation*}
The penalty function is applied to individual subpopulation parameters, ensuring ease of implementation in the modified EM algorithm.
Three penalty functions are recommended for different mixtures of $\gF$:
\begin{itemize}
\item 
Under \emph{finite Gaussian mixtures}, let $S_{x}$ be the sample covariance matrix of $\gX_{i}$.
Chen et al.~\cite{chen2009inference} recommends the penalty function:
\[p(\theta) = \text{tr}(\Sigma^{-1}S_{x}) + \log\text{det}(\Sigma)\]
where $\text{tr}(\cdot)$ is the trace of a square matrix.
The penalty function reduces to $p(\theta) = s_x^2/\sigma^2 + \log \sigma^2$ under univariate Gaussian mixtures~\cite{chen2008inference}. 

\item 
Under \emph{finite location-scale mixtures} with location $\mu$ and scale $\sigma$, the sample variance $s_x^2$ in the previous penalty function can be replaced by a scale-invariant statistic, such as the squared sample interquartile range.
This is particularly useful when the variance of $f_0(\cdot)$ is not finite.

\item 
Under finite mixture of \emph{two-parameter Gamma distributions}, the likelihood is also unbounded. 
Chen~\cite{chen2016consistency} recommends the penalty function to be $p(\theta)= r - \log r$, where $r$ is the shape parameter in the Gamma distribution.
\end{itemize}

\paragraph{The EM algorithm for penalized MLE}
The EM algorithm can be easily adapted to compute the pMLE with the recommended penalty functions~\cite{chen2009inference}. 
Using the latent variable interpretation, the penalized complete data log-likelihood is:
\[
p\ell_{n}^c(G) = \sum_{j=1}^{n} \sum_{k=1}^K \mathbbm{1}(z_{ij}=k)\log \{ w_k f(x_{ij};\theta_k)\} - a_n \sum_{k=1}^{K} p(\theta_k) 
\]
The only random quantity in $p\ell_{n}^c(G)$ is $\{z_{ijk}, j \in [n], k\in[K]\}$ when conditioning on $\gX_i$.
The E-step involves computing the conditional expectation of $p\ell_{n}^c(G)$, which remains valid as in~\eqref{eq:conditional_expectation}. 
The M-step maximizes $Q(G|G^{(t)})$, which now includes a penalty term:
\begin{equation*}
Q(G|G^{(t)}) = \sum_{j=1}^{n} \sum_{k=1}^K w_{ijk}^{(t)} \log \{w_k f(x_{ij};\theta_k)\} - a_n \sum_{k=1}^{K} p(\theta_k).
\end{equation*}

With the recommended penalty function, the updated mixing weights and subpopulation parameters are
\[
w_{k}^{(t+1)} = n^{-1} \sum_{j=1}^{n} w_{ijk}^{(t)}
\]
and
\begin{equation}
\label{eq:m_step_subpop}
\theta_{k}^{(t+1)} 
= \argmax_{\theta} \left\{\sum_{j=1}^{n} w_{ijk}^{(t)} \log f(x_{ij};\theta) - a_n p(\theta)\right\}.
\end{equation}

For Gaussian mixtures~\cite{chen2009inference}, the solution to~\eqref{eq:m_step_subpop} has the closed-form solution:
\[
\begin{split}
  \mu_{k}^{(t+1)}  
      &= \big \{ n w_k^{(t+1)} \big \}^{-1}  \sum_{j=1}^{n} w_{ijk}^{(t)} x_{ij}, \\
  \Sigma_{k}^{(t+1)}  
      &= \big \{ 2a_{n} + n w_{k}^{(t+1)} \big \}^{-1} \big \{ 2a_{n} S_{x} + S_k^{(t+1)} \big \},
\end{split}
\]
where 
\[
S_k^{(t+1)} = \sum_{j=1}^{n} w_{ijk}^{(t)} (x_{ij} - \mu_k^{(t+1)}) (x_{ij}-\mu_k^{(t+1)} )^{\top}.
\] 
For general location-scale mixtures, the M-step may not have a closed-form solution but only requires solving a two-variable optimization problem, which is usually simple.

The EM algorithm for pMLE, like its MLE counterpart, increases the value of the penalized likelihood after each iteration. 
For all $t$, $\Sigma_{k}^{(t)} \geq { 2a_{n}/ (n+2a_{n})} S_{x} > 0$, ensuring the covariance matrices in $G^{(t)}$ have a lower bound that does not depend on the parameter values. 
This property guarantees that the log-likelihood under Gaussian mixture at $G^{(t)}$ has a finite upper bound. 
The above iterative procedure is guaranteed to have $p\ell_{n}^c(G^{(t)})$ converge to at least a non-degenerate local maximum.

\subsection{Reduction estimator and the MM algorithm}
\label{app:mm_algorithm}
Zhang and Chen~\cite{zhang2022distributed} proposed a novel approach to mixture reduction (MR) using a composite transportation divergence.
Specifically, let $c: \gF \times \gF \to \sR_{+}$ be a cost function, where we denote $c(\theta, \theta') = c(f(\cdot; \theta), f(\cdot; \theta') )$ for simplicity.
Let $G = \sum_{k=1}^{K} w_k \delta_{\theta_k}$ and  $G' = \sum_{k'=1}^{K'} w'_{k'} \delta_{\theta'_{k'}}$ be two mixing distributions with $K$ and $K'$ support points.
The composite transportation divergence between these mixing distributions is defined to be
\begin{equation}
\label{eq:CTD_def}
    T_{c}(G, G') = \min \left \{
    \sum_{k=1}
    ^{K}\sum_{k'=1}^{K'} \pi_{k, k'} c(\theta_k, \theta'_{k'}):~\sum_{k=1}^{K} \pi_{k,k'} = w'_{k'},  \sum_{k'=1}^{K'} \pi_{k,k'}= w_k, \pi_{k,k'} \geq 0
    \right \}.
\end{equation}
where the minimum is taking over all transportation plans $\pi$ that satisfies two sets of marginal constraints.
This divergence is the lowest cost of transporting subpopulations
of $G$ to those of $G'$~\cite{chen2017optimal,delon2020wasserstein,bing2022estimation,nguyen2013convergence}.
The MR estimator is then defined to be
\begin{equation*}
\widehat{G}^{\mr} =\argmin\left\{T_{c}(\widebar{G}, G):~G\in\sG_{K}\right\}.
\end{equation*}
Note that while $\widehat G^{\mr}$ is ostensibly the solution to a complicated bilevel optimization program, the program actually simplifies and has a clear interpretation, as discussed in Zhang and Chen~\cite{zhang2022distributed}.
More precisely, in \eqref{eq:CTD_def}, the column-wise marginal constraints on the transportation plan $\pi$ are redundant and the full mass of each support point $\theta_{ik}$ in $\widebar G$ is transported to its ``closest'' support point in $\widehat{G}^{\mr}$.
The optimal transportation plan can thus be interpreted as an optimal clustering of the subpopulation parameters $\braces*{\theta_{ik} \colon (i,k) \in [m] \times [K]}$ into $K$ clusters $\{\gC_{j}\}_{j=1}^{K}$, with the support points of $\widehat{G}^{\mr}$ forming the cluster ``barycenters''.
Mathematically, the simplified optimization program can be written as:
\begin{equation*}
    \min \braces*{\sum_{j=1}^K \sum_{(i,k) \in \gC_j} w_{ik}c(\hat\theta_{ik},\theta_j) \colon \theta_1,\ldots,\theta_K \in \Theta, [m] \times [K] = \gC_1\sqcup\cdots\sqcup\gC_K }
\end{equation*}
where $\sqcup$ denotes disjoint union.
This formulation shows that MR implicitly aligns local components through clustering while simultaneously aggregating them via barycentric averaging.
The numerical algorithm for computing the estimator is given in Appendix~\ref{app:mm_algorithm}.
The clustering interpretation of MR naturally leads to an efficient majorization–minimization (MM) algorithm for computation.
Each iteration alternates between assigning local subpopulations to the nearest tentative cluster centers and updating these centers by barycentric averaging.
Formally, let $G^\dagger = \sum_{j=1}^K w^\dagger_j \delta_{\theta_j^{\dagger}}$ denote an initial mixing distribution.
The MM algorithm then proceeds by repeating the following two steps until convergence:
\begin{itemize}
  \item
        \underline{\emph{Majorization step.}} 
        Partition the support points of $\widebar{G}$,
        $\{\widehat \theta_{i k}: i \in [m], k\in [K]\}$, into $K$ groups based on their proximity to subpopulation parameter $\{\theta_j^{\dagger}, j=1,2,\ldots,K\}$, of the current iterate.
        In other words, we assign the $k$th subpopulation on the $i$th machine to the $j$th cluster $\gC_j$ if $c(\widehat\theta_{ik},\theta_j^{\dagger})\leq c(\widehat\theta_{ik},\theta_{j'}^{\dagger})$ for any $j'$.
        We denote this as $(i,k)\in \gC_j$.
  \item
        \underline{\emph{Minimization step.}} Update $\theta_j^{\dagger}$ by computing the barycenter of the subpopulation parameters in its cluster.
        Specifically, set $\theta_j^{\dagger} = \argmin_{\theta\in\Theta} \sum_{\{(i,k) \in \gC_j\}} c(\widehat\theta_{ik}, \theta)$.
        We also update the mixing weight via $w^\dagger_j=\sum_{\{(i,k) \in \gC_j\}} \lambda_i \widehat w_{ik}$.
\end{itemize}

\begin{algorithm}
\begin{algorithmic}
\State {\bfseries Initialization:} $f(\cdot;\theta^{\dagger}_{k})$, $k\in [K]$
\Repeat
\For {$k\in[K]$}
\State {\textbf{Majorization:} For $\gamma\in[mK]$, let
\begin{equation*}
  \pi_{\gamma k}=
  \begin{cases}
  (n_{\gamma}/N)\widehat w_{\gamma} & \text{if}~k=\argmin_{k'} c(f(\cdot;\widehat\theta_{\gamma}),f(\cdot;\theta^{\dagger}_{k'}))\\
  0& \text{otherwise}
  \end{cases}
\end{equation*}}
% \Comment{Majorization}
\State{\textbf{Minimization:} Let 
\begin{equation}
\label{eq:barycentre}
\theta^{\dagger}_{k}=\argmin_{\theta} \sum_{\gamma} \pi_{\gamma k} c(f(\cdot;\widehat\theta_{\gamma}), f(\cdot;\theta))
\end{equation}}
% \Comment{}
\EndFor
\Until 
the change in $\sum_{\gamma,k}\pi_{\gamma k}c(f(\cdot;\widehat\theta_{\gamma}), f(\cdot;\theta^{\dagger}_k))$ is below some threshold
\State{
Let $w^{\dagger}_{k} = \sum_{\gamma} \pi_{\gamma k}$}
\State {\bfseries Output:} $\sum_{k}w_k^{\dagger}f(\cdot;\theta^{\dagger}_{k})$
\end{algorithmic}
\caption{MM algorithm for GMR estimator under a general cost function $c(\cdot, \cdot)$.}
\label{alg:mm_reduction}
\end{algorithm}

The MM algorithm is summarized in Algorithm~\ref{alg:mm_reduction} for a general cost function $c(\cdot,\cdot)$ defined on $\gF$.
For simplicity, we order $\{\widehat\theta_{ik}\}$ arbitrarily and use a single index $\gamma$ in the algorithm.
When $\gF = \{\phi(x;\mu,\Sigma)\}$ is the space of Gaussian distributions and the cost function is the KL divergence between two Gaussian distributions:
\begin{align*}
&c(\phi(x;\widehat\mu_{\gamma},\widehat\Sigma_{\gamma}),\phi(x;\mu^{\dagger}_k,\Sigma^{\dagger}_k))\\
=&~\KL(\phi(x;\widehat\mu_{\gamma},\widehat\Sigma_{\gamma})\|\phi(x;\mu^{\dagger}_k,\Sigma^{\dagger}_k))\\
=&~-\log \phi(\widehat\mu_{\gamma};\mu^{\dagger}_k,\Sigma^{\dagger}_k) 
- \frac{1}{2}\Big \{\log\text{det}(2\pi\widehat\Sigma_{\gamma})-\mbox{tr}(\{\Sigma^{\dagger}_k\}^{-1}\widehat\Sigma_{\gamma}) + d
\Big\},
\end{align*}
the minimization step~\eqref{eq:barycentre} has an analytical solution:
\begin{equation*}
\begin{split}
&\mu^{\dagger}_{k} =\Big\{\sum_{\gamma}\pi_{\gamma k}\Big\}^{-1}\sum_{\gamma}\pi_{\gamma k}\widehat\mu_{\gamma},\\
&\Sigma^{\dagger}_{k}= \Big\{\sum_{\gamma}\pi_{\gamma k}\Big\}^{-1}\sum_{\gamma}\pi_{\gamma k}\{\widehat\Sigma_{\gamma} + (\widehat\mu_{\gamma}-\mu^{\dagger}_{k})(\widehat\mu_{\gamma}-\mu^{\dagger}_{k})^{\top}
\}.  
\end{split}
\end{equation*}

\subsection{Numerical algorithm for TRIM}
\label{app:trim_algorithm}
In this section, we summarize the iterative algorithm for the TRIM estimator proposed in Del Barrio et al.~\cite{del2019robust} for reference in Algorithm~\ref{alg:trimmed_barycentre}.
\begin{algorithm}[!htbp]
\begin{algorithmic}
\State {\bfseries Initialization:} $\{f(\cdot;\theta^{\dagger}_{k})\}_{k=1}^{K}$, threshold $\eta \in (0,1)$
\Repeat
\For {$\gamma\in[mK]$}
\State{$c_{\gamma} = \argmin_{j} c(f(\cdot;\tilde\theta_{\gamma}), f(\cdot;\theta^{\dagger}_{j})), l_{\gamma} = c(f(\cdot;\tilde\theta_{\gamma}), f(\cdot;\theta^{\dagger}_{c_{\gamma}}))$}
\Comment{Find the cluster assignment and compute its distance to the cluster center}
\EndFor
\State Find the permutation $\{(1),\ldots, (mK)\}$ such that $l_{(1)}\leq l_{(2)}\leq \cdots\leq l_{(mK)}$
\State{Set $\tau = \inf\{\eta \in [mK]: \sum_{\gamma=1}^{\eta}\widetilde{w}_{(\gamma)}\leq 1-\eta\}$, let
\begin{equation*}
\kappa_{\gamma} = 
\begin{cases}
\widetilde{w}_{(\gamma)} & \gamma <\tau\\
1-\eta - \sum_{\gamma < \tau} \widetilde{w}_{(\gamma)} & \gamma = \tau\\
0 & \text{otherwise}
\end{cases}
\end{equation*}}
\Comment{Trimming}
\For {$k\in[K]$}
\State{$\theta^{\dagger}_{k} = \argmin_{\theta} \sum_{\{\gamma: c_{\gamma} =k \}} \kappa_{\gamma} c(f(\cdot;\tilde\theta_{\gamma}), f(\cdot;\theta))$}
\Comment{Update cluster centers}
\State{$w^{\dagger}_{k} = (1-\eta)^{-1}\sum_{\{\gamma: c_{\gamma} =k \}} \kappa_{\gamma} \widetilde{w}_{\gamma}$}
\Comment{Update weights}
\EndFor
\Until 
there is no change in $c_{\gamma}$ for all $\gamma\in[mK]$
\State {\bfseries Output:} $\sum_{k}w_k^{\ddagger}f(\cdot;\theta^{\dagger}_{k})$
\end{algorithmic}
\caption{Iterative algorithm to compute the TRIM estimator in Del Barrio et al.~\cite{del2019robust}.}
\label{alg:trimmed_barycentre}
\end{algorithm}

%% The following is a directive for TeXShop to indicate the main file
%%!TEX root = ../main.tex

\section{Theoretical results}
\label{app:proof}
This appendix establishes the theoretical guarantees for \ours{} in three steps.
We first collect the stochastic properties of the local MLEs and translate them into component-level cost bounds.
We then analyze the ideal case in which the anchor machine is failure-free, since this is the event on which the component-wise alignment and filtering steps are well behaved.
Finally, we prove that the data-driven anchor selection recovers such a failure-free machine with high probability and use this event to derive the convergence rate of the proposed estimator.
Throughout the appendix, $K$ is fixed and all asymptotic statements are along
sequences for which $m,n\to\infty$, together with the additional relative
growth conditions stated in the corresponding result.

%%%%%%%%%%%%%%%%%%%%%%%%%%%%%%%%%%%%%%%%%%%%%%%%%%%%%%%%%%%%%%%%%%%%%%%%%%%%%%%%%%%%%%%%%%%
\subsection{Discussion of the adversarial assumptions}
\label{app:assumptions}
The assumptions used in the main theoretical results are stated in Section~\ref{subsec:theoretical_properties}.
This appendix does not restate them; instead, it records sufficient conditions
and counterexamples for the three corruption restrictions: anchor separability
and the two clauses of corruption capacity.
Throughout this section, $d_\psi=p+1$ is the dimension of a component parameter $\psi=(w,\theta^\top)^\top$, and $C$ denotes a generic positive constant whose value may change from line to line.
The examples below illustrate when the conditions hold and what types of
coordinated attacks they exclude.

\begin{example}[Independent component contamination]
Fix the component-specific failure sets $\gB_1,\ldots,\gB_K$.
Suppose the corrupted subpopulation parameters
$\{\xi_{i\ell,\theta}:\ell\in[K],\ i\in\gB_\ell\}$ have densities bounded
by $M_Q<\infty$, are independent of the authentic local estimators, and are
independent across machines.  Dependence among the at most $K$ corrupted
coordinates on one machine is allowed.  The corrupted weights may be chosen
arbitrarily and dependently, subject to the validity of each transmitted
mixture. Assume additionally that $A$ is uniformly strongly convex on
$\Theta$, so that
\[
D(\theta',\theta)\ge (\underline\eta/2)\|\theta'-\theta\|^2,
\qquad \theta,\theta'\in\Theta,
\]
for some $\underline\eta>0$. If $m\to\infty$, $\rho/n\to0$, and
$m=o(n^{p/2})$, then Assumption~\ref{assump:anchor_separability} and
Assumption~\ref{assump:component_capacity}(a) hold.  The same model also
satisfies Assumption~\ref{assump:component_capacity}(b).

We first verify truth-centred component capacity. For a fixed truth $\psi_k^*$ and radius $r=C\rho/n$, membership in the cost ball implies
\[
\|\xi_{i\ell,\theta}-\theta_k^*\|\le C_0r^{1/2}.
\]
The bounded density therefore gives
\[
\sP\{c_\lambda(\xi_{i\ell},\psi_k^*)\le r\}
\le C_1(\rho/n)^{p/2}
\]
for every $\ell\in[K]$ and $i\in\gB_\ell$. Define the machine-level count
\[
Y_{i,k}(r)=
\sum_{\ell:i\in\gB_\ell}
\mathbbm{1}\{\|\xi_{i\ell,\theta}-\theta_k^*\|\le C_0r^{1/2}\}.
\]
The variables $\{Y_{i,k}(r):i\in[m]\}$ are independent across machines, satisfy
$0\le Y_{i,k}(r)\le K$, and have total expectation $o(m)$. For any fixed
$\beta\in(0,1/2)$, Bernstein's inequality and a union bound over the fixed
number $K$ of truths give
\[
\sP\left\{
\max_{k\in[K]}
\left|\left\{(i,\ell):\ \ell\in[K],\ i\in\gB_\ell,\
c_\lambda(\xi_{i\ell},\psi_k^*)\le C\rho n^{-1}\right\}\right|
>\beta m
\right\}
\le C_2e^{-c_2m}
\]
for all sufficiently large $n$. This proves
Assumption~\ref{assump:component_capacity}(a).

We next verify anchor separability, keeping the order of quantifiers explicit.
Fix the radius multiplier $C>0$ and choose a constant $A=A(C)>0$ below.
For every $i'\in\gB$, choose one genuinely corrupted component occurrence with
latent origin $\ell(i')\in[K]$, so that $i'\in\gB_{\ell(i')}$. Bounded density gives
\[
\sP\left\{
\min_{s\in[K]}
\|\xi_{i'\ell(i'),\theta}-\theta_s^*\|
\le A n^{-1/2}
\right\}
\le C_3A^pn^{-p/2}.
\]
Because $m=o(n^{p/2})$, a union bound shows that, with probability tending to one, every selected corrupted anchor component is farther than $An^{-1/2}$ from every true subpopulation parameter.

Fix one such separated anchor component. Global strong convexity implies that any transmitted component within cost $Cn^{-1}$ of it has its subpopulation parameter within $C_4n^{-1/2}$ of the anchor. If that transmitted component is authentic with latent origin $\ell$, then
\[
\sqrt n\|\widehat\theta_{i\ell}-\theta_\ell^*\|\ge A-C_4.
\]
The $q$th-moment bound in
Corollary~\ref{cor:Berry-Esseen_for_component_estimates}(c) allows $A=A(C)$
to be chosen so that the expected number of such authentic components on
machines other than $i'$ is at most $m/16$. Conditional on the anchor,
bounded density makes the expected number of nearby corrupted components on
the other machines $O(mn^{-p/2})=o(m)$, hence at most $m/16$ for all
sufficiently large $n$.

Group both kinds of indicators by machine. Conditional on the anchor, these machine-level counts are independent across machines and bounded by $K$. The candidate machine itself contributes at most $K$ component occurrences. Bernstein's inequality therefore yields
\[
\sP\{\text{at least }m/2\text{ transmitted components lie within cost }Cn^{-1}
\mid\text{the anchor position is separated}\}
\le C_5e^{-c_5m}.
\]
The assigned cluster is a subset of all transmitted components in this ball.
A final union bound over at most $m$ Byzantine candidate machines gives
$mC_5e^{-c_5m}=o(1)$ and proves
Assumption~\ref{assump:anchor_separability}.

Finally, fix a fully authentic candidate component $\widehat\psi_{ik}$ and
condition on all authentic local estimates.  Each corrupted value falls in a
cost ball of radius $C/n$ around this center with probability $O(n^{-p/2})$.
The same machine-level Bernstein argument bounds the probability that more
than $\beta_\dagger m$ such values occur by $C\exp(-cm)$.  A union bound over
the at most $Km$ fully authentic candidate components proves
Assumption~\ref{assump:component_capacity}(b).  This last argument uses
independence between the stochastic contamination and the authentic local
estimates in the model stated above.  More generally, the same verification
of part~(b) remains valid when, conditional on all authentic estimates, the
corrupted subpopulation parameters have uniformly bounded conditional
densities and are conditionally independent across machines.  Thus the
condition permits diffuse data-dependent contamination, but not unrestricted
adaptive point-mass placement around a selected authentic candidate.
\end{example}

\begin{example}[Attacks excluded by the capacity assumptions]
The conditions can be violated when attacks are coordinated and concentrated.
If the adversary sets more than $\beta m$ corrupted component estimates equal to 
$\psi_k^*$ for some $k$, then the corresponding component cluster is dominated
by corrupted estimates, violating Assumption~\ref{assump:component_capacity}(a).
Similarly, if every Byzantine machine transmits component estimates that mimic
clean-scale local estimators,
for example \[\widetilde{\psi}_{i'k}=\psi_k^*+O(n^{-1/2}),\qquad k\in[K],
\]
then each anchor component can collect a majority-sized cluster of 
failure-free estimates within cost radius $O(n^{-1})$.
In this case the radius score of a Byzantine-failure candidate 
can be of the same order as that of a failure-free candidate, 
and Assumption~\ref{assump:anchor_separability} fails.
Even if truth-centred capacity holds, adaptive corrupted values can instead be
placed around an unusually noisy fully authentic candidate.  Such an attack
can violate Assumption~\ref{assump:component_capacity}(b), which explains
why that extra condition appears only in the tail-trimming regime.
\end{example}

%%%%%%%%%%%%%%%%%%%%%%%%%%%%%%%%%%%%%%%%%%%%%%%%%%%%%%%%%%%%%%%%%%%%%%%%%%%%%%%%%%%%%%%%%%%
\subsection{Properties of the local estimate}
\label{app:local_mle}
This subsection records the probabilistic behavior of a single failure-free ordinary MLE needed below.
The local score expansion and Berry--Esseen argument follow the same local
chart analysis as Lemma~3.1 and Appendix~B.1 of Zhang, Tan and
Chen~\cite{zhang2026byzantine}.  Their compactness, identifiability,
negative-definite population-Hessian, and local eighth-moment conditions
correspond, respectively, to Assumptions~\ref{assump:compact_parameter_space},
\ref{assump:identifiability}, \ref{assump:local_strong_concavity}, and
\ref{assump:smoothness}(a) here when $q=8$.  Huang and
Huo~\cite{huang2019distributed} obtain polynomial
moment bounds for a global $M$-estimator using global concavity of the sample
criterion.  That condition is generally unavailable for an ordinary
finite-mixture likelihood, and qualitative consistency alone does not control
the complement of the local chart at the polynomial rate needed for moment
bounds.  Assumption~\ref{assump:smoothness}(b) supplies that quantitative
localization directly.  The proof below therefore establishes the moment and
bias bounds self-containedly, for every fixed exponent $s\le q$, rather than
importing them from the globally concave setting.

The next lemma collects the resulting local-MLE properties in the form needed
for the subsequent aggregation analysis.

\begin{lemma}[Properties of local MLE]
\label{lemma:local_MLE_properties}
Let $\{X_1, \ldots, X_n\}$ be $n $ IID observations 
from the mixture density
$f_{G^*}$ corresponding to the finite mixing distribution $G^*$.
Under Assumptions~\ref{assump:compact_parameter_space}--\ref{assump:smoothness}, 
let
$\widehat{G} = \sum_{k=1}^K \widehat{w}_k \delta_{\widehat{\theta}_k}$ 
be an ordinary global maximizer of the likelihood over the compact constrained model class, and let $\widehat{\mG}$ be its oracle-aligned free-coordinate vector. Then, for all sufficiently large $n$, the following bounds hold with finite constants that may differ from line to line and may depend on the fixed model and its regularity constants, but not on $n$ or the machine index:
\begin{enumerate}[label=(\alph*), leftmargin=*]
\item (Berry-Esseen central limit theorem (CLT))
\[
\sup_{A \in \mathscr{A}_{d_G}}
\left| 
\sP \left\{ \sqrt{n} \mI(\mG^*)^{1/2} (\widehat{\mG} - \mG^*) \in A \right\} - \sP(Z_{d_G} \in A) \right| \le Cn^{-1/2},
\]
where $Z_{d_G}\sim N(0,\mI_{d_G})$ and $\mathscr{A}_{d_G}$ is the collection of convex Borel subsets of $\sR^{d_G}$;
\item (Bias bound)
$\| \mathbb{E} \{ \widehat{\mG} - \mG^* \} \| \le Cn^{-1}$;
\item (Moment bounds)
$\mathbb{E} \{\| \widehat{\mG} - \mG^* \|^s \} \le C_sn^{-s/2}$ for
$1 \leq s \leq q$, where $q\geq8$ is the fixed order in
Assumption~\ref{assump:smoothness}.
\end{enumerate}
\end{lemma}

\begin{proof}
We first control the event on which the ordinary global maximizer lies outside
the local likelihood chart.  Let
\[
d_{\mathrm{match}}(G,G^*)
=\left[
\min_{\sigma\in S_K}
\sum_{k=1}^K
\left\{
|w_{\sigma(k)}-w_k^*|^2
+\|\theta_{\sigma(k)}-\theta_k^*\|^2
\right\}
\right]^{1/2}.
\]
Choose a fixed $\delta_1>0$ small enough that
$d_{\mathrm{match}}(G,G^*)<\delta_1$, followed by oracle ordering, places the free vector
$\mG$ in the interior of the $\delta_0$-neighborhood from
Assumption~\ref{assump:smoothness}.  Continuity of $L$, compactness, and
identifiability give the strictly positive population gap
\[
\Delta_1
=L(G^*)-
\sup_{G:\,d_{\mathrm{match}}(G,G^*)\ge\delta_1}L(G)>0.
\]
If a global maximizer $\widehat G$ satisfies
$d_{\mathrm{match}}(\widehat G,G^*)\ge\delta_1$, its empirical likelihood is no smaller than
that of $G^*$.  Hence
\[
\sup_{G\in\sG_K}
\left|\{L_n(G)-L_n(G^*)\}-\{L(G)-L(G^*)\}\right|
\ge\Delta_1.
\]
Markov's inequality and Assumption~\ref{assump:smoothness}(b) therefore give
\begin{equation}
\label{eq:global_mle_polynomial_localization}
\sP\{d_{\mathrm{match}}(\widehat G,G^*)\ge\delta_1\}=O(n^{-q/2}).
\end{equation}

We next work in the oracle-ordered local chart.  Write
$\lambda_0=\lambda_{\min}\{\mI(\mG^*)\}>0$.  The $q$th-moment bounds on the
Hessian and on $W$ in Assumption~\ref{assump:smoothness}(a), together with the
fixed dimension and the standard moment
inequality for sums of independent centered vectors, imply that the event
\[
\left\|
\frac1n\sum_{j=1}^n
\{\nabla^2\ell(\mG^*;X_j)-\nabla^2L(\mG^*)\}
\right\|\le\lambda_0/4,
\qquad
\frac1n\sum_{j=1}^nW(X_j)\le 1+2\sE W(X)
\]
has complement of probability $O(n^{-q/2})$.  Reducing $\delta_1$ if
necessary, the Hessian-Lipschitz condition then makes the empirical
log-likelihood uniformly strongly concave throughout this local chart.
Intersect this event with the localization event in
\eqref{eq:global_mle_polynomial_localization} and call the result $\gE_n$.
Thus $\sP(\gE_n^c)=O(n^{-q/2})$.

On $\gE_n$, the oracle-aligned global maximizer is the unique interior local
maximizer.  Its score equation and the mean-value expansion around $\mG^*$
give
\[
\|\widehat{\mG}-\mG^*\|
\le C\left\|\frac1n\sum_{j=1}^n s_{\mG^*}(X_j)\right\|.
\]
For every fixed $1\le s\le q$, the same moment inequality for independent
centered scores gives
\[
\sE\left\|
\frac1n\sum_{j=1}^n s_{\mG^*}(X_j)
\right\|^s=O(n^{-s/2}).
\]
The constrained parameter class has finite diameter.  Consequently,
\[
\begin{aligned}
\sE\|\widehat{\mG}-\mG^*\|^s
&\le
\sE\{\|\widehat{\mG}-\mG^*\|^s\mathbbm1(\gE_n)\}
+C\sP(\gE_n^c)\\
&=O(n^{-s/2})+O(n^{-q/2})
=O(n^{-s/2}),
\end{aligned}
\]
which proves part~(c), including the claimed extension beyond order eight.

For completeness, define the remainder globally by
\[
R_n=\widehat{\mG}-\mG^*
-\mI(\mG^*)^{-1}\frac1n\sum_{j=1}^n s_{\mG^*}(X_j).
\]
On $\gE_n$, the usual second-order score expansion gives
\[
\|R_n\|
\le C\left\|
\frac1n\sum_{j=1}^n
\{\nabla^2\ell(\mG^*;X_j)-\nabla^2L(\mG^*)\}
\right\|
\left\|\frac1n\sum_{j=1}^n s_{\mG^*}(X_j)\right\|
+C\|\widehat{\mG}-\mG^*\|^2.
\]
The moment bounds above and
Cauchy--Schwarz therefore give
$\sE\{\|R_n\|\mathbbm1(\gE_n)\}=O(n^{-1})$.
On $\gE_n^c$, finite diameter of the parameter class, H\"older's inequality,
$\sP(\gE_n^c)=O(n^{-q/2})$, and the $q$th-moment bound for the score average
give
\[
\sE\{\|R_n\|\mathbbm1(\gE_n^c)\}
\le C\sP(\gE_n^c)
+C\left(\sE\left\|\frac1n\sum_{j=1}^n s_{\mG^*}(X_j)\right\|^q\right)^{1/q}
\sP(\gE_n^c)^{1-1/q}
=O(n^{-q/2}).
\]
Thus $\sE\|R_n\|=O(n^{-1})$.  Since the score has mean zero, this proves the
bias bound in part~(b).  Finally, the convex-set Berry--Esseen theorem for the estimator in
the local likelihood chart, applied as in Zhang, Tan and
Chen~\cite{zhang2026byzantine}, gives part~(a).  Replacing that local root by
the oracle-aligned global
maximizer changes its probability bound by at most
$\sP(\gE_n^c)=O(n^{-q/2})=O(n^{-1/2})$.  This completes the proof.
\end{proof}

\noindent 
In Lemma~\ref{lemma:local_MLE_properties}, the components of $\widehat G$ are ordered by \eqref{eq:parameter_vector} before $\widehat{\mG}$ is formed.
This oracle alignment is used only in the theoretical analysis, 
so that the local estimate of component $k$ 
can be compared with the true component $\psi_k^*$.
The lemma neither covers a penalized MLE nor an arbitrary stationary point returned by EM; those estimators require separate penalty and optimization conditions.

The following corollary is an immediate component-wise consequence of
Lemma~\ref{lemma:local_MLE_properties}. It will be used later to derive
concentration bounds for the component discrepancy
$c_{\lambda}(\widehat{\psi}_k,\psi_k^*)$.

\begin{corollary}[Properties of local component estimates]
\label{cor:Berry-Esseen_for_component_estimates}
Let $\widehat{\psi}_k=(\widehat{w}_k,\widehat{\theta}_k^{\top})^{\top}$ 
denote the oracle-aligned local estimate of the $k$th component.
Let $\mathscr A'$ denote the collection of all convex Borel subsets of
$\mathbb R^{p+1}$.
Let $h_k(\mG)=\psi_k$ be the affine map from the free mixture coordinates to the $k$th component parameter, and define
\[
P_k=\frac{\partial h_k(\mG)}{\partial\mG^\top},
\qquad
\Sigma_k=P_k\mI(\mG^*)^{-1}P_k^\top,
\qquad
\mI(\psi_k^*)=\Sigma_k^{-1}.
\]
For $k<K$, the weight row of $P_k$ selects $w_k$; for $k=K$, it is $(-1,\ldots,-1,0,\ldots,0)$ because $w_K=1-\sum_{\ell=1}^{K-1}w_\ell$.
In either case $P_k$ has full row rank. Since $\mI(\mG^*)$ is positive definite by Assumption~\ref{assump:local_strong_concavity}, $\Sigma_k$ is positive definite.
Then, under the assumptions of Lemma~\ref{lemma:local_MLE_properties},
\begin{enumerate}[label=(\alph*), leftmargin=*]
\item
\[
\sup_{A \in \mathscr{A}'} \left| \mathbb{P}\left( \sqrt{n} \, \mI(\psi_k^*)^{1/2} (\widehat{\psi}_k - \psi_k^*) \in A \right)
  - \mathbb{P}(Z \in A) \right| = O(n^{-1/2}),
\]
where $Z$ is a standard Gaussian random vector having the same dimension as
$\psi_k^*$;
\item $\|\sE(\widehat\psi_k-\psi_k^*)\|=O(n^{-1})$;
\item $\sE\|\widehat\psi_k-\psi_k^*\|^s=O(n^{-s/2})$ for every
$1\le s\le q$.
\end{enumerate}
\end{corollary}

\begin{proof}
Set
\[
V_n=\sqrt n\,\mI(\mG^*)^{1/2}(\widehat{\mG}-\mG^*),
\qquad
L_k=\Sigma_k^{-1/2}P_k\mI(\mG^*)^{-1/2}.
\]
Because $\Sigma_k=P_k\mI(\mG^*)^{-1}P_k^\top$, we have $L_kL_k^\top=\mI_{p+1}$, and the affine identity $\widehat\psi_k-\psi_k^*=P_k(\widehat{\mG}-\mG^*)$ gives
\[
\sqrt n\,\Sigma_k^{-1/2}(\widehat\psi_k-\psi_k^*)=L_kV_n.
\]
For every convex $A\subseteq\sR^{p+1}$, the inverse image $\{u:L_ku\in A\}$ is convex in $\sR^{d_G}$. Applying Lemma~\ref{lemma:local_MLE_properties}(a) to this inverse image and observing that $L_kZ_{d_G}\sim N(0,\mI_{p+1})$ proves part (a).
Parts (b) and (c) follow from the same affine identity,
Lemma~\ref{lemma:local_MLE_properties}(b)--(c), and the fact that the fixed
matrix $P_k$ has finite operator norm.
\end{proof}

%%%%%%%%%%%%%%%%%%%%%%%%%%%%%%%%%%%%%%%%%%%%%%%%%%%%%%%%%%%%%%%%%%%%%%%%%%%%%%%%%%%%%%%%%%%
\subsection{Theoretical properties of \ours{}}
We now study the theoretical properties of the proposed \ours{} estimator.
The analysis proceeds in three parts.
First, we study the ideal case in which the anchor machine is failure-free and show that 
the anchor induces correct component-wise alignment and 
a half-sample radius of the right statistical order.
Second, we show that Algorithm~\ref{alg:initial_estimate} selects 
such a failure-free anchor with probability tending to one.
Third, we combine the ideal-anchor analysis and the selection guarantee 
to establish the convergence rate of the final \ours{} estimator.

\subsubsection{Ideal failure-free anchor: alignment and radius properties}
\label{app:ideal_anchor_properties}
This part analyzes what happens once the initial anchor machine is failure-free.
The main results in this section are Lemma~\ref{lemma:accurate_alignment}, 
which shows that a fixed failure-free anchor resolves label switching, 
Corollary~\ref{cor:uniform_accurate_alignment}, 
which extends this alignment guarantee uniformly over all failure-free candidate anchors, 
and Lemma~\ref{lemma:radius_concentration}, 
which shows that the radius used by the filter is of order $n^{-1}$.
The remaining results in this section are preparatory.
Specifically, Lemma~\ref{lemma:reverse_Bregman_divergence} and Corollary~\ref{cor:cost_function} 
translate the component cost into Euclidean distance, while Lemma~\ref{lemma:concentration_of_cost_function} 
characterizes the stochastic scale of failure-free component costs.
These facts are later transferred to the data-driven anchor selected by Algorithm~\ref{alg:initial_estimate}.

We first relate the reverse Bregman divergence to Euclidean distance in a local neighborhood.
This conversion is needed because the filtering rule is stated in terms of the cost $c_{\lambda}$, 
whereas the stochastic properties of local MLEs are naturally stated in Euclidean norm.
\begin{lemma}
\label{lemma:reverse_Bregman_divergence}
Let $\mathcal U\subseteq\sR^p$ be open and let $A:\mathcal U\to\sR$ be twice continuously differentiable and strictly convex. Consider the reverse Bregman divergence
\[
D(\theta', \theta) = A(\theta) - A(\theta') - (\theta - \theta')^\top \nabla A(\theta')
\]
for $\theta,\theta'\in\mathcal U$.
Assume that, for some $\theta^*\in\mathcal U$ and $\epsilon>0$, the ball $B_\epsilon(\theta^*;d_E)$ is contained in $\mathcal U$ and, for constants $\eta_+ \geq \eta_- > 0$, the Hessian satisfies
\[
\eta_- \mI \preceq \nabla^2 A(\theta) \preceq \eta_+ \mI,
\qquad \theta\in B_\epsilon(\theta^*;d_E).
\]
Then, for all $\theta,\theta'\in B_\epsilon(\theta^*;d_E)$,
\begin{equation}
\label{eq:bound_of_Bregman}
     \frac{\eta_-}{2} \|\theta' - \theta\|^2\leq D(\theta', \theta) 
     \leq \frac{\eta_+}{2} \|\theta' - \theta\|^2.
\end{equation}
\end{lemma}

\begin{proof}
Consider the Taylor expansion of $A(\theta)$ around $\theta'$:
\begin{equation}
\label{eq:taylor_expansion_of_Bregman_divergence}
    A(\theta) = A(\theta') + \langle \nabla A(\theta'), \theta - \theta' \rangle + \int_{0}^{1} (1-t)(\theta - \theta')^{\top} \nabla^2 A(t \theta + (1 - t)\theta') (\theta - \theta')  dt.
\end{equation}
Plugging this into the definition of $D(\theta', \theta)$, we may write
\begin{equation}
\label{eq:simplified_Bregman}
D(\theta', \theta) = (\theta - \theta')^{\top} M(\theta,\theta')(\theta - \theta'),
\end{equation}
where
\[
M(\theta,\theta') = \int_{0}^{1} (1-t)\nabla^2 A(t \theta + (1 - t)\theta')  dt.
\]

When both $\theta$ and $\theta'$ satisfy $\|\theta - \theta^*\|, \|\theta' - \theta^*\| \leq \epsilon$, the convexity of the ball $B_{\epsilon}(\theta^*;d_E)$ implies $t \theta + (1 - t)\theta' \in B_{\epsilon}(\theta^*;d_E)$ for any $t \in [0, 1]$.
Therefore, the strong convexity of $A(\cdot)$ on $B_{\epsilon}(\theta^*;d_E)$ ensures that $(\eta_-/2)\mI\preceq M(\theta,\theta')\preceq(\eta_+/2)\mI$.
Plugging this into \eqref{eq:simplified_Bregman} gives \eqref{eq:bound_of_Bregman}.
\end{proof}

The next corollary extends the preceding local equivalence from the 
subpopulation parameter $\theta$ 
to the full component parameter $\psi=(w,\theta^\top)^\top$.
This is the main tool for comparing cost balls with Euclidean neighborhoods in the alignment and filtering arguments.

\begin{corollary}
\label{cor:cost_function}
Consider the component-level cost function 
$c_{\lambda}(\psi',\psi)=D(\theta',\theta)+\lambda (w'-w)^2$, where $\lambda>0$ and $D$ satisfies the conditions of Lemma~\ref{lemma:reverse_Bregman_divergence}. For all $w,w'\in\sR$ and all $\theta,\theta'\in B_\epsilon(\theta^*;d_E)$,
\[
\min(\eta_-/2,\lambda) \|\psi' - \psi\|^2\triangleq \eta_-'\|\psi' - \psi\|^2\leq c_{\lambda}(\psi', \psi) \leq \max(\eta_+/2,\lambda) \|\psi' - \psi\|^2 \triangleq \eta'_+\|\psi' - \psi\|^2.
\]
\end{corollary}

\begin{proof}
It follows by applying Lemma~\ref{lemma:reverse_Bregman_divergence} 
to the subpopulation parameters 
and then adding the nonnegative weight term.
Specifically, when both $\theta$ and $\theta'$ lie in $B_{\epsilon}(\theta^*;d_E)$,
\[
c_{\lambda}(\psi',\psi)
\leq \eta_+/2\|\theta'-\theta\|^2+\lambda(w'-w)^2
\leq \max(\eta_+/2,\lambda)\|\psi'-\psi\|^2,
\]
and similarly
\[
c_{\lambda}(\psi',\psi)
\geq \eta_-/2\|\theta'-\theta\|^2+\lambda(w'-w)^2
\geq \min(\eta_-/2,\lambda)\|\psi'-\psi\|^2.
\]
\end{proof}

The following lemma formalizes the role of a failure-free anchor in resolving label switching.
It shows that, once one clean machine is available, 
nearest-anchor assignment based on the cost $c_{\lambda}$ 
correctly aligns every failure-free component estimate with high probability.

\begin{lemma}[Accurate alignment]
\label{lemma:accurate_alignment} 
Suppose $i^*$ is a failure-free anchor machine and
$mn^{-q/2}\to0$.
Here, $\widetilde{\bm\psi}_{i^*}$ denotes the transmitted local estimate on the anchor machine, 
and $\widetilde{\psi}_{ij}$ denotes the transmitted occurrence whose latent
population origin is $j$.  The server observes these occurrences only up to
permutation.  The set $\sO_k=\{(i,k):i\in\gI_k\}$ in the conclusion contains
the authentic occurrences corresponding to the $k$th true component.
Let $\sC_k(\widetilde{\bm\psi}_{i^*})$ be the cluster produced by the
cost-based assignment rule in~\eqref{eq:cost_based_alignment}, with the
permutation-invariant tie convention stated there.
Then, under Assumptions~\ref{assump:compact_parameter_space}--\ref{assump:smoothness} and~\ref{assump:reverse_bregman_divergence},
\begin{equation}
\label{eq:accurate_alignment_event}
\sP\left(
    \bigcap_{k=1}^{K}
    \left\{\sO_k \subseteq \sC_k(\widetilde{\bm\psi}_{i^*})\right\}
\right)\rightarrow 1.
\end{equation}
\end{lemma}

\begin{remark}
Lemma~\ref{lemma:accurate_alignment} shows that, as long as the anchor machine is failure-free, the cost-based assignment rule correctly assigns every failure-free component estimate to the corresponding cluster.
Consequently, after assignment, $\sO_k\subseteq \sC_k$, and each cluster $\sC_k$ contains all $(1-\alpha_k)m$ failure-free estimates of $\psi^*_k$ for $k\in[K]$.
This is precisely where the distinction between observed and latent labels
matters: the server applies the rule to unordered transmitted occurrences, but
the event $\sO_k\subseteq\sC_k$ states that all authentic occurrences with
latent population identity $k$ are recovered into the $k$th anchor-induced
cluster.
The growth condition $mn^{-q/2}\to0$ is used only to obtain uniform
localization over all local machines at a fixed separation scale.
It is implied by the stronger conditions imposed later for the final convergence-rate result.
\end{remark}

\begin{proof}
We first establish a high-probability event on which all failure-free local
component estimates are uniformly close to their corresponding true
components. On this event, the remaining argument reduces to a deterministic
comparison of the assignment costs.

\noindent\textit{\underline{Step 1: Uniform localization of failure-free local estimates.}}
Assumption~\ref{assump:smoothness}, Corollary~\ref{cor:Berry-Esseen_for_component_estimates}(c), and Markov's inequality give, for
every fixed $\varepsilon_0>0$,
\[
\sup_{i\le m}\sP\{\|\widehat{\mG}_i-\mG^*\|>\varepsilon_0\}
\le C_q\varepsilon_0^{-q}n^{-q/2}.
\]
Hence, by a union bound,
\begin{equation}
\label{eq:concentration_of_max_distance}
\sP\left(\max_{1\le i\le m}
\|\widehat{\mG}_i-\mG^*\|>\varepsilon_0\right)
\le C_q\varepsilon_0^{-q}mn^{-q/2}=o(1).
\end{equation}
Since the same statement holds for every fixed $\varepsilon_0>0$, we conclude
that
\begin{equation}
\label{eq:order_of_maxinum_distance}
\max_{1\leq i\leq m}\|\widehat{\mG}_i- \mG^*\|=o_P(1).
\end{equation}

\noindent
\textit{\underline{Step 2: Choose separated component neighborhoods.}}
By Assumption~\ref{assump:compact_parameter_space}, the true component parameters $\{\psi_k^*\}_{k=1}^K$ are distinct. Because the Bregman divergence $c_{\lambda}(\cdot, \cdot)$ is strictly positive for distinct arguments, the minimal wrong-component alignment cost at the truth is bounded away from zero:
\[
\Delta_0 \coloneqq \min_{k \neq l} c_{\lambda}(\psi_k^*, \psi_l^*) > 0.
\]
We next enlarge each true component $\psi_k^*$ to a small neighborhood 
and require that this positive gap is preserved inside these neighborhoods.
From the continuity of $c_{\lambda}$ at $(\psi_k^*,\psi_l^*)$ and the finiteness of $K$, there exists a sufficiently small, fixed radius $\epsilon'>0$, satisfying $\epsilon' < \epsilon/2$ and $4\eta'_+ (\epsilon')^2 < \Delta_0/2$ such that, for the local neighborhoods
\[
\gN_k(\epsilon') \coloneqq \left\{\psi = (w, \theta^\top)^\top : |w - w_k^*| \leq \epsilon', \, \|\theta - \theta_k^*\| \leq \epsilon'\right\}, \qquad k \in [K],
\]
the following uniform lower bound holds across distinct clusters:
\begin{equation}
\label{eq:separation_bound}
\min_{k \neq l} \inf_{\psi \in \gN_k(\epsilon'), \, \psi' \in \gN_l(\epsilon')} c_{\lambda}(\psi, \psi') \geq \frac{\Delta_0}{2}.
\end{equation}
This separation is the quantity needed for alignment: 
any other anchor component belongs to a different component neighborhood, 
so its cost to a failure-free $k$th component estimate is uniformly bounded below by $\Delta_0/2$.
By \eqref{eq:order_of_maxinum_distance}, the event
\begin{equation}
\label{eq:universal_bound}
\gE_{\epsilon'}=
\left\{\|\widehat{\mG}_i-\mG^*\|\leq \epsilon', \forall~i=1,2,\cdots,m\right\}
\end{equation}
has probability tending to one.

\noindent\textit{\underline{Step 3: 
Compare the cost to the correct anchor component with the costs to the other anchor components.}}
Work on the event $\gE_{\epsilon'}$.
Then $\|\widehat{\psi}_{ik}-\psi^*_{k}\|\leq\epsilon'$ for all $i\in[m]$ 
and $k\in[K]$.
Because machine $i^*$ is failure-free, 
$\widetilde{\psi}_{i^*k}=\widehat{\psi}_{i^*k}$ and 
hence $\widetilde{\psi}_{i^*k}\in\gN_k(\epsilon')$ for every $k\in[K]$.
For every $(i,k)\in\sO_k$, the component is authentic, so
$\widetilde{\psi}_{ik}=\widehat{\psi}_{ik}\in\gN_k(\epsilon')$.
For such a failure-free estimate of the $k$th component, the correct anchor component is $\widetilde{\psi}_{i^*k}$.
The triangle inequality gives
\begin{equation}
\label{eq:true_alignment}
    \|\widetilde{\psi}_{ik}-\widetilde{\psi}_{i^*k}\|\leq\|\widetilde{\psi}_{ik}-\psi^*_k\|+\|\psi^*_k-\widetilde{\psi}_{i^*k}\|\leq2\epsilon'.
\end{equation}
Combining \eqref{eq:true_alignment} with Corollary~\ref{cor:cost_function}, we get
\[
c_{\lambda}(\widetilde{\psi}_{ik},\widetilde{\psi}_{i^*k})\leq 4\eta'_+{\epsilon'}^2.
\]
For any other anchor component $\widetilde{\psi}_{i^*l}$ with $l\neq k$, the component itself is still failure-free because the anchor machine is failure-free, but it corresponds to the $l$th true component rather than the $k$th true component.
Thus, $\widetilde{\psi}_{ik}\in\gN_k(\epsilon')$ and $\widetilde{\psi}_{i^*l}\in\gN_l(\epsilon')$.
Therefore, by the separation chosen in Step 2,
\[
\min_{l\neq k}c_{\lambda}(\widetilde{\psi}_{ik},\widetilde{\psi}_{i^*l})\geq \Delta_0/2.
\]
Since $4\eta'_+{\epsilon'}^2<\Delta_0/2$, it follows that, for every $(i,k)\in\sO_k$,
\[
c_{\lambda}(\widetilde{\psi}_{ik},\widetilde{\psi}_{i^*k})< \min_{l\neq k}c_{\lambda}(\widetilde{\psi}_{ik},\widetilde{\psi}_{i^*l}).
\]
Thus, on $\gE_{\epsilon'}$, each failure-free $k$th component estimate is assigned to $\sC_k(\widetilde{\bm\psi}_{i^*})$.
Since $\sP(\gE_{\epsilon'})\rightarrow1$, \eqref{eq:accurate_alignment_event} follows.
\end{proof}

\begin{corollary}[Uniform accurate alignment over failure-free anchors]
\label{cor:uniform_accurate_alignment}
Suppose $mn^{-q/2}\to0$.
Under Assumptions~\ref{assump:compact_parameter_space}--\ref{assump:smoothness} and~\ref{assump:reverse_bregman_divergence},
\[
\sP\left(
    \bigcap_{i\in\gB^c}
    \bigcap_{k=1}^{K}
    \left\{\sO_k\subseteq\sC_k(\widetilde{\bm\psi}_i)\right\}
\right)\to1.
\]
\end{corollary}

\begin{remark}
Corollary~\ref{cor:uniform_accurate_alignment} is the uniform version of Lemma~\ref{lemma:accurate_alignment}.
The lemma is stated for a fixed failure-free anchor, whereas the corollary applies simultaneously to all failure-free candidate anchors.
This uniform form is used later when the anchor is selected from the transmitted local estimates.
\end{remark}
\begin{proof}
Recall the uniform localization event $
\gE_{\epsilon'}
=
\left\{
\|\widehat{\mG}_i-\mG^*\|
\le
\epsilon',
\forall i=1,\ldots,m
\right\}$. Step~1 in the proof of Lemma~\ref{lemma:accurate_alignment} establishes that $
\sP(\gE_{\epsilon'})\to1$.
Conditional on $\gE_{\epsilon'}$, every failure-free local component
estimate lies in its corresponding neighborhood
$\gN_k(\epsilon')$.
The separation argument in Step~2 is deterministic and independent of the
choice of the anchor machine.
Now let $i\in\gB^c$ be arbitrary.
Since the anchor machine is failure-free,
$\widetilde\psi_{ik}\in\gN_k(\epsilon')$
for every $k\in[K]$.
Repeating the deterministic cost comparison in Step~3 yields
$\sO_k
\subseteq
\sC_k(\widetilde{\bm\psi}_i),
 \forall k\in[K]$.

Because the above argument holds for every failure-free anchor
simultaneously on the same event $\gE_{\epsilon'}$, we obtain
\[
\gE_{\epsilon'}
\subseteq
\bigcap_{i\in\gB^c}
\bigcap_{k=1}^K
\{
\sO_k
\subseteq
\sC_k(\widetilde{\bm\psi}_i)
\}.
\]
Hence
\[
\sP\left(
\bigcap_{i\in\gB^c}
\bigcap_{k=1}^K
\{
\sO_k
\subseteq
\sC_k(\widetilde{\bm\psi}_i)
\}
\right)
\ge
\sP(\gE_{\epsilon'})
\rightarrow1,
\]
which completes the proof.
\end{proof}

The next lemma quantifies the scale of the cost among failure-free component estimates.
Although Byzantine costs are arbitrary, 
the failure-free costs concentrate at order $n^{-1}$ and have polynomial tail control.
This result provides the main concentration bound used to 
show that the filtering radius is neither too small nor too large.

\begin{lemma}[Concentration of cost function]
\label{lemma:concentration_of_cost_function}
Under Assumption~\ref{assump:component_wise_majority},
Assumptions~\ref{assump:compact_parameter_space}--\ref{assump:smoothness},
Assumption~\ref{assump:reverse_bregman_divergence}, and independent local data
splits, the following concentration results hold:
\begin{enumerate}[label=(\alph*)]
\item 
Here $\widehat{\psi}_{ik}$ denotes the oracle-aligned local estimate of the component on machine $i$ corresponding to $\psi_k^*$.
Let $j_1,\cdots,j_{(1-\alpha_k)m}$ be a permutation of $\gI_k$, i.e.
the index set of machines whose oracle-aligned $k$th component estimates are failure-free, such that $c_{\lambda}(\widehat{\psi}_{j_1k},\psi^*_k)\leq \cdots \leq c_{\lambda}(\widehat{\psi}_{j_{(1-\alpha_k)m}k},\psi^*_k)$.
For any $0 <\varepsilon<1/2$, as $n,m \rightarrow \infty$, we have
\begin{equation}
\label{eq:Bulk}
c_{\lambda}(\widehat{\psi}_{j_{\lfloor\varepsilon m \rfloor} k},\psi^*_k)=\Omega_P(n^{-1}), \quad c_{\lambda}(\widehat{\psi}_{j_{\lceil m/2\rceil} k},\psi^*_k)=O_P(n^{-1}).
\end{equation}

\item 
For any machine $i\in[m]$, as $n,\rho_n\to\infty$ with $\rho_n=O(n)$,
\begin{equation}
\label{eq:large_deviation}
\sP\left(c_{\lambda}(\widehat{\psi}_{ik},\psi^*_k)\geq \rho_n n^{-1}\right)=O(\rho_n^{-q/2}).
\end{equation}
\end{enumerate} 
\end{lemma}

\begin{remark}
Lemma~\ref{lemma:concentration_of_cost_function} has three consequences.
First, the first relation in \eqref{eq:Bulk} implies that the dispersion of 
the failure-free local estimates is estimable from the data.  This first
consequence is motivational and is not invoked in the subsequent proofs.
Second, together with the component-wise majority rule, the second relation in \eqref{eq:Bulk} 
implies that a majority of the local estimates lie within $O(n^{-1})$ cost of
the true parameters.  It therefore gives every failure-free candidate anchor
an $O_P(n^{-1})$ radius score, the clean-side benchmark in
Theorem~\ref{theorem:accurate_selection}.
Third, \eqref{eq:large_deviation} implies that a cost ball with radius $O(\rho_n n^{-1})$ 
around a good initial estimate contains almost all failure-free local estimates 
for a moderately growing inflation factor $\rho_n$; this is the tail input to
Lemma~\ref{lemma:filtering_event}(b)--(c).
For part (a), no relative growth rate between $m$ and $n$ is required: $n\to\infty$ determines the local statistical scale $n^{-1}$, while $m\to\infty$ is used only to compare the empirical order statistics, equivalently empirical quantiles, with their population quantiles.
\end{remark}

\begin{proof}
Let $T_{ik}=c_{\lambda}(\widehat{\psi}_{ik},\psi_k^*)$ be the cost of the $k$th 
component estimate on machine $i$ to the true $k$th component.
Since $\widehat{\psi}_{ik}$ is oracle-aligned, 
the scalar cost $T_{ik}$ is well-defined without ambiguity for every $i\in\gI_k$.
Then, the variables $\{T_{ik}: i\in\gI_k\}$ are IID
because they are computed from independent local data splits.
Let $H_{n,k}(t)=\sP(T_{ik}\leq t)$ denote the common distribution function 
of the failure-free component-$k$ costs, and let
\[
    q_{n,k}(a)=\inf\{t:H_{n,k}(t)\geq a\}
\]
denote the population quantile function associated with $H_{n,k}$.
We use the subindices $n$ and $k$ to emphasize the dependence of these quantities 
on the local sample size and the component index.

The empirical quantile function based on the $M_k=|\gI_k|$ authentic costs is denoted by $\widehat q_{M_k,k}$; equivalently, after ordering the costs, $\widehat q_{M_k,k}(a)=T_{j_{\lceil aM_k\rceil}k}$ for $a\in(0,1]$.
For example, $T_{j_{\lfloor\varepsilon m\rfloor}k}$ and 
$T_{j_{\lceil m/2\rceil}k}$ are the empirical quantiles used in part (a).
Write $\mI(\psi_k^*)$ as $\mI_k$ for simplicity.
Let $Y=Z^{\top}\mI_k^{-1}Z$, where $Z$ is a standard Gaussian vector in $\sR^{p+1}$, and let
\[
q^*(a)=\inf\{x:\sP(Y\leq x)\geq a\}
\]
be the quantile function of $Y$.

\noindent\underline{\textit{Step 1: Comparing cost quantiles with generalized chi-square quantiles.}}
Let $U_{ik}=\sqrt{n}\,\mI_k^{1/2}(\widehat{\psi}_{ik}-\psi_k^*)$.
Then
\[
    \left\{n\|\widehat{\psi}_{ik}-\psi_k^*\|^2\leq t\right\}
    =
    \left\{U_{ik}^{\top}\mI_k^{-1}U_{ik}\leq t\right\}.
\]
The set $\{u:u^{\top}\mI_k^{-1}u\leq t\}$ is an ellipsoid and hence is convex.
Therefore, this set is admissible in the convex-set Berry--Esseen bound of Corollary~\ref{cor:Berry-Esseen_for_component_estimates}, which gives, for every $t\geq0$,
\[
    \left|
    \sP\left(n\|\widehat{\psi}_{ik}-\psi_k^*\|^2\leq t\right)
    -
    \sP\left(Y\leq t\right)
    \right|
    \leq Cn^{-1/2},
\]
		Let $\gL_{ik,n}=\{\|\widehat{\psi}_{ik}-\psi_k^*\|\leq\epsilon\}$, where $\epsilon$ is the local-neighborhood radius in Assumption~\ref{assump:reverse_bregman_divergence}.
			By Markov's inequality and
        Corollary~\ref{cor:Berry-Esseen_for_component_estimates}(c),
		\[
		    \sP(\gL_{ik,n}^c)
		    \leq
		    \epsilon^{-q}\sE\|\widehat{\psi}_{ik}-\psi_k^*\|^q
    =
		    O(n^{-q/2}).
	\]
		On $\gL_{ik,n}$, Corollary~\ref{cor:cost_function} gives
	\[
	    \eta'_- \|\widehat{\psi}_{ik}-\psi_k^*\|^2
	    \leq
	    T_{ik}
	    \leq
	    \eta'_+ \|\widehat{\psi}_{ik}-\psi_k^*\|^2.
	\]
		We now decompose the probability according to the partition
        $\gL_{ik,n}\cup\gL_{ik,n}^c$.  On $\gL_{ik,n}$, the preceding
        cost--Euclidean comparison applies, while the contribution from
        $\gL_{ik,n}^c$ is bounded by
        $\sP(\gL_{ik,n}^c)=O(n^{-q/2})$.
Therefore, for any $x\geq0$,
\[
    \sP(T_{ik}\leq x)
    \leq
    \sP\left(\eta'_-\|\widehat{\psi}_{ik}-\psi_k^*\|^2\leq x\right)+O(n^{-q/2})
    \leq
    \sP\left(Y\leq \frac{nx}{\eta'_-}\right)+Cn^{-1/2},
\]
and similarly
\[
    \sP(T_{ik}\leq x)
    \geq
    \sP\left(\eta'_+\|\widehat{\psi}_{ik}-\psi_k^*\|^2\leq x\right)-O(n^{-q/2})
    \geq
    \sP\left(Y\leq \frac{nx}{\eta'_+}\right)-Cn^{-1/2}.
\]
Recall that $q_{n,k}(a)$ is the population quantile function of $T_{ik}$, and $q^*(a)$ is the population quantile function of $Y$.
These CDF bounds imply the quantile bounds by inverting the two CDF inequalities, we have that for each fixed quantile level $a$ bounded away from $0$ and $1$,
\begin{equation}
\label{eq:quantile_bound}
\frac{\eta'_-}{n}q^*(a-Cn^{-1/2})\leq q_{n,k}(a) \leq \frac{\eta'_+}{n}q^*(a+Cn^{-1/2}).
\end{equation}
In particular, at any fixed level $a\in(0,1)$, $q_{n,k}(a)$ is of order $n^{-1}$ because $Y$ is a nondegenerate generalized chi-square random variable.

\noindent \textit{\underline{Step 2: Transferring population quantile bounds to empirical order statistics.}}
Fix $0<\varepsilon<1/2$ and set $M_k=|\gI_k|=(1-\alpha_k)m$.
The order statistics in \eqref{eq:Bulk} are taken over the $M_k$ failure-free component-$k$ costs.
Thus, the index $\lfloor\varepsilon m\rfloor$ corresponds, among the failure-free costs, to the quantile level approximately $a_{1,k}=\varepsilon/(1-\alpha_k)$, whereas the index $\lceil m/2\rceil$ corresponds to the quantile level approximately $a_{2,k}=1/\{2(1-\alpha_k)\}$.

By Assumption~\ref{assump:component_wise_majority}, $1-\alpha_k\ge 1/2+\kappa$ for all sufficiently large $m$. Hence, uniformly over $k\in[K]$,
\[
\varepsilon\le a_{1,k}\le \frac{\varepsilon}{1/2+\kappa}<1,
\qquad
\frac12\le a_{2,k}\le \frac{1}{1+2\kappa}<1.
\]
Choose a fixed slack $\tau_\varepsilon>0$, independent of $k$, such that
\[
    \tau_\varepsilon<\frac12
    \min\left\{
        \varepsilon,
        1-\frac{\varepsilon}{1/2+\kappa},
        1-\frac{1}{1+2\kappa}
    \right\}.
\]

All three quantities in the minimum are strictly positive, so this choice is possible.
Applying Lemma~\ref{lemma:concentration_of_order_statistics} to the independent variables $\{T_{j_1k},\ldots,T_{j_{M_k}k}\}$, there is an event with probability at least $1-2\exp(-2M_k\tau_\varepsilon^2)$ on which every empirical quantile at a level bounded away from $0$ and $1$ is trapped between the corresponding population quantiles with levels shifted by $\tau_\varepsilon$:
\[
    \widehat q_{M_k,k}(a)=T_{j_{\lceil aM_k\rceil}k},
\]

and $\widehat q_{M_k,k}(a)\in[q_{n,k}(a-\tau_\varepsilon),q_{n,k}(a+\tau_\varepsilon)]$. The exact levels of the two displayed order statistics differ from $a_{1,k}$ and $a_{2,k}$ by at most $M_k^{-1}$ because of the floor and ceiling operations; since $M_k\to\infty$, this discrepancy is absorbed by the fixed slack $\tau_\varepsilon$.
Applying this with $a=a_{1,k}$ and $a=a_{2,k}$ gives
\[
    T_{j_{\lfloor\varepsilon m \rfloor}k}
    \geq
    q_{n,k}\left(\frac{\varepsilon}{1-\alpha_k}-\tau_\varepsilon\right)
\quad
\text{and}
\quad
    T_{j_{\lceil m/2\rceil}k}
    \leq
    q_{n,k}\left(\frac{1}{2(1-\alpha_k)}+\tau_\varepsilon\right),
\]

where the index $\lceil m/2\rceil$ is valid for all sufficiently large $m$ because $M_k\ge(1/2+\kappa)m$.

The shifted levels are bounded away from $0$ and $1$ uniformly over $k$ by the choice of $\tau_\varepsilon$ and the majority margin $\kappa$.
Therefore, Step 1 and \eqref{eq:quantile_bound} imply that the lower population quantile on the right of the first display is bounded below by a positive constant times $n^{-1}$, while the upper population quantile on the right of the second display is bounded above by a finite constant times $n^{-1}$.
Consequently,
\[
    T_{j_{\lfloor\varepsilon m \rfloor}k}=\Omega_P(n^{-1}),
    \qquad
    T_{j_{\lceil m/2\rceil}k}=O_P(n^{-1}),
\]
which is exactly \eqref{eq:Bulk}.

The only role of $m$ in the order-statistic argument is through the probability $1-2\exp(-2M_k\tau_\varepsilon^2)$, which tends to one because $M_k\ge(1/2+\kappa)m\to\infty$.
Thus part (a) does not impose any relative rate condition between $m$ and $n$.

\noindent \textit{\underline{Step 3: Proving the final bound in~\eqref{eq:large_deviation}.}}
As in Step 1, decompose the sample space into the localization event
$\gL_{ik,n}=\{\|\widehat{\psi}_{ik}-\psi_k^*\|\leq\epsilon\}$ and its
complement.  On $\gL_{ik,n}$, the upper cost--Euclidean comparison in
Corollary~\ref{cor:cost_function} gives
\[
    T_{ik}
    \leq
    \eta'_+\|\widehat{\psi}_{ik}-\psi_k^*\|^2.
\]
Therefore,
\[
\sP(T_{ik}\geq \rho_n n^{-1})
\leq
\sP(\gL_{ik,n}^c)
+
\sP\left(\{T_{ik}\geq \rho_n n^{-1}\}\cap\gL_{ik,n}\right)
\leq
	O(n^{-q/2})
+
\sP\left(\eta'_+\|\widehat{\psi}_{ik}-\psi_k^*\|^2\geq \rho_n n^{-1}\right).
\]
	By Corollary~\ref{cor:Berry-Esseen_for_component_estimates}(c) and
    Markov's inequality,
\[
\sP\left(\eta'_+\|\widehat{\psi}_{ik}-\psi_k^*\|^2\geq \rho_n n^{-1}\right)
\leq
	\frac{\sE\|\widehat{\psi}_{ik}-\psi_k^*\|^q}
    {(\rho_n n^{-1}/\eta'_+)^{q/2}}
	=O(\rho_n^{-q/2}).
\]

Since $\rho_n=O(n)$, we have
$n^{-q/2}=O(\rho_n^{-q/2})$. The two terms in the preceding probability
bound are therefore both $O(\rho_n^{-q/2})$, which proves
\eqref{eq:large_deviation}. This completes the proof of
Lemma~\ref{lemma:concentration_of_cost_function}.
\end{proof}

Lemma~\ref{lemma:concentration_of_cost_function} gives the stochastic scale 
of the failure-free component costs when the center is the true component $\psi_k^*$.
The radius used by the algorithm is slightly different: it is centered at the anchor component $\widetilde{\psi}_{i^*k}$, and it is computed 
after the local components have been assigned to anchor-induced clusters.
Thus, before this radius can be used in the filtering argument, 
we need to verify that replacing the unknown truth by a failure-free anchor does not change its order.
The next lemma proves this transfer.
When the anchor machine is failure-free, Lemma~\ref{lemma:accurate_alignment} ensures that the failure-free estimates of component $k$ are assigned to the $k$th anchor cluster, and the local consistency of the anchor keeps $\widetilde{\psi}_{i^*k}$ within the same $n^{-1/2}$ neighborhood of $\psi_k^*$.
Consequently, the empirical half-sample radius centered at $\widetilde{\psi}_{i^*k}$ has the same $n^{-1}$ order as the truth-centered clean costs.
This order is essential for the filter: a radius that is too small would discard many failure-free estimates, whereas a radius that is too large would allow too many corrupted estimates to remain.

\begin{lemma}[Radius concentration]
\label{lemma:radius_concentration}
Suppose $i^*$ is a failure-free anchor machine.  Under Assumptions
\ref{assump:component_wise_majority},
\ref{assump:compact_parameter_space}--\ref{assump:smoothness},
\ref{assump:reverse_bregman_divergence}, and part~(a) of
Assumption~\ref{assump:component_capacity}, assume independent local data splits,
$mn^{-q/2}\to0$, $\rho\to\infty$, and $\rho/n\to0$.  Using
$\widetilde{\bm\psi}_{i^*}$ as the anchor and applying the cost-based
assignment~\eqref{eq:cost_based_alignment}, we have, for $k\in[K]$,
\[
r(\widetilde{\psi}_{i^*k})
=
\left\{c_{\lambda}(\widetilde{\psi}_{ij},\widetilde{\psi}_{i^*k}):(i,j)\in \sC_k(\widetilde{\bm\psi}_{i^*})\right\}_{\lceil m/2\rceil}
=\Theta_P(n^{-1}).
\]
\end{lemma}
\begin{proof}
Recall that $B_r(z;d)=\{z':d(z',z)\le r\}$ denotes the ball centered at $z$ with radius $r$ under the discrepancy $d$.
By definition, the half-sample radius can be written as
\begin{equation}
\label{eq:inf_definition_of_radius}
r(\widetilde{\psi}_{i^*k})=\inf\left\{r:\left|\{(i,j)\in\sC_k(\widetilde{\bm\psi}_{i^*}):\widetilde{\psi}_{ij}\in B_r(\widetilde{\psi}_{i^*k};c_{\lambda})\}\right|\geq \left\lceil\frac{m}{2}\right\rceil \right\}.
\end{equation}

\noindent\underline{\textit{Step 1: Upper bound of the radius.}}
By Lemma~\ref{lemma:accurate_alignment}, all $(1-\alpha_k)m$ failure-free $k$th component estimates are assigned to $\sC_k(\widetilde{\bm\psi}_{i^*})$.

Since $1-\alpha_k\ge1/2+\kappa$, the collection of failure-free $k$th component estimates contains a majority of machines.
Thus $r(\widetilde{\psi}_{i^*k})$ is upper bounded by the $\lceil m/2\rceil$th order statistic of the costs from the failure-free $k$th component estimates to the anchor.

Recall that $j_1,\cdots,j_{m_k}$ is a permutation of $\gI_k$ (the fixed set
of machines with authentic $k$th component estimates), where
$m_k=(1-\alpha_k)m$, such that
$c_{\lambda}(\widehat{\psi}_{j_1k},\psi^*_k)\leq \cdots \leq
c_{\lambda}(\widehat{\psi}_{j_{m_k}k},\psi^*_k)$.
By \eqref{eq:order_of_maxinum_distance}, on an event with probability tending to one, $\|\widehat{\psi}_{ik}-\psi^*_k\|\leq \epsilon'$ for all $(i,k)\in[m]\times[K]$ and some fixed $0<\epsilon'<\epsilon/2$, where $\epsilon$ is the neighborhood radius in Assumption~\ref{assump:reverse_bregman_divergence}.

Since $\{(j_1,k),\cdots,(j_{(1-\alpha_k)m},k)\}\subseteq \sC_k(\widetilde{\bm\psi}_{i^*})$, \eqref{eq:inf_definition_of_radius} implies that, for sufficiently large $n$,
\[r(\widetilde{\psi}_{i^*k})\leq \max_{1\leq i\leq \lceil m/2\rceil }c_{\lambda}(\widehat{\psi}_{j_ik},\widetilde{\psi}_{i^*k})\leq \eta'_+ \max_{1\leq i\leq \lceil m/2\rceil}\|\widehat{\psi}_{j_ik}-\widetilde{\psi}_{i^*k}\|^2,\]
where the second inequality comes from Corollary~\ref{cor:cost_function}.
We can further get
\[
r(\widetilde{\psi}_{i^*k}) \leq
2\eta'_+ \left(
\max_{1\leq i\leq \lceil m/2\rceil}
\|\widehat{\psi}_{j_ik}-\psi^*_k\|^2
+
\|\widetilde{\psi}_{i^*k}-\psi^*_k\|^2
\right)
\leq
\frac{2\eta'_+}{\eta'_-}
c_{\lambda}(\widehat{\psi}_{j_{\lceil m/2\rceil}k},\psi^*_k)
+
2\eta'_+
\|\widetilde{\psi}_{i^*k}-\psi^*_k\|^2.
\]
The first term on the right-hand side is bounded by Lemma~\ref{lemma:concentration_of_cost_function}(a), which is of order $O_P(n^{-1})$.
For the second term, since machine $i^*$ is failure-free, Lemma~\ref{lemma:local_MLE_properties}(c) gives
\[
\sE\{\|\widetilde{\psi}_{i^*k}-\psi^*_k\|^2\}=O(n^{-1}).
\]
Markov's inequality then yields
\begin{equation}
\label{eq:order_of_fixed_failure-free_estimate}
\|\widetilde{\psi}_{i^*k}-
\psi^*_k\|=O_P(n^{-1/2}).
\end{equation}
Combining these bounds yields
\[
r(\widetilde{\psi}_{i^*k})=O_P(n^{-1}).
\]

\noindent\underline{\textit{Step 2: Lower bound via quantile functions for off-center cost.}}
For the lower bound, we use the quantile of the cost distribution centered at the anchor component.
Conditional on $\widetilde{\psi}_{i^*k}$, let $H_{n,k}^{i^*}(t) = \sP(c_{\lambda}(\widehat{\psi}_{ik},\widetilde{\psi}_{i^*k})\leq t\mid \widetilde{\psi}_{i^*k})$,  $i\in\gI_k, i\neq i^*$, denote the anchor-centered population distribution function, where $\widehat{\psi}_{ik}$ is the oracle-aligned local MLE of the $k$th component on machine $i$.
The right-hand side does not depend on the particular choice of $i\in\gI_k\backslash\{i^*\}$ because the authentic local samples are IID and independent across machines.
Define its conditional population quantile function by
\[
    q_{n,k}^{i^*}(a)=\inf\{t:H_{n,k}^{i^*}(t)\geq a\}.
\]
This is the anchor-centered counterpart of the truth-centered distribution and quantile functions $H_{n,k}$ and $q_{n,k}$ used in Lemma~\ref{lemma:concentration_of_cost_function}.
Let
\[
\gE_{\mathrm{align}}
=
\bigcap_{\ell=1}^{K}
\left\{\sO_\ell\subseteq \sC_\ell(\widetilde{\bm\psi}_{i^*})\right\}.
\]
By Lemma~\ref{lemma:accurate_alignment},
$\sP(\gE_{\mathrm{align}})\to 1$ under $mn^{-q/2}\to0$.
On $\gE_{\mathrm{align}}$, the anchor-induced cluster $\sC_k(\widetilde{\bm\psi}_{i^*})$ contains all failure-free estimates in $\sO_k$; since the cost-based assignment forms disjoint clusters, any failure-free estimate assigned to $\sC_k(\widetilde{\bm\psi}_{i^*})$ has latent origin $k$.
The upper-bound conclusion in Step 1, $r(\widetilde{\psi}_{i^*k})=O_P(n^{-1})$, implies that, for any sufficiently large constant $L>0$, the event $\{r(\widetilde{\psi}_{i^*k})\le L/n\}$ has probability arbitrarily close to one.
Also, by \eqref{eq:order_of_fixed_failure-free_estimate}, for any $\delta>0$ there exists a fixed constant $A_\delta>0$ such that
\[
    \sP\left(
    \|\widetilde{\psi}_{i^*k}-\psi_k^*\|\le A_\delta n^{-1/2}
    \right)\ge 1-\delta
\]
for all sufficiently large $n$.
On this anchor-localization event, $\widetilde{\psi}_{i^*k}$ lies in the local neighborhood of $\psi_k^*$ based on Lemma~\ref{lemma:local_MLE_properties}.
By compactness of the parameter space and strict positivity of the Bregman divergence away from its diagonal, there exists a constant $\gamma_k>0$ such that, whenever $\widetilde{\psi}_{i^*k}$ lies sufficiently close to $\psi_k^*$, any component parameter whose subpopulation parameter is outside the local neighborhood of $\theta_k^*$ has cost at least $\gamma_k$ from $\widetilde{\psi}_{i^*k}$.
Since $L/n<\gamma_k$ for all sufficiently large $n$, every corrupted component estimate in
$B_{r(\widetilde{\psi}_{i^*k})}(\widetilde{\psi}_{i^*k};c_{\lambda})$ must have its subpopulation parameter inside this local neighborhood.
Therefore, on the same event, Corollary~\ref{cor:cost_function} gives
\[
    \|\xi_{i\ell}-\widetilde{\psi}_{i^*k}\|^2
    \le \eta_-^{\prime -1}L n^{-1}
\]
for every corrupted occurrence with latent origin $\ell$ on machine $i$ whose
transmitted value $\xi_{i\ell}$ lies inside the radius-defining ball.
Consequently,
\[
    \|\xi_{i\ell}-\psi_k^*\|^2
    \le
    2\eta_-^{\prime -1}L n^{-1}
    +2A_\delta^2 n^{-1}.
\]
Applying Corollary~\ref{cor:cost_function} once more, there exists a fixed constant $L_\delta'=\eta_+'\left(2\eta_-^{\prime -1}L+2A_\delta^2\right)$ such that each corrupted component in the anchor-centered radius-defining ball also belongs to the true-centered ball $B_{L_\delta'/n}(\psi_k^*;c_{\lambda})$.
For all sufficiently large $n$, this ball is contained in
$B_{L_\delta'\rho/n}(\psi_k^*;c_\lambda)$.  Hence
Assumption~\ref{assump:component_capacity}(a), applied with the fixed constant
$L_\delta'$, implies that it contains at most $\beta m$ corrupted component
estimates with probability tending to one.
Since $\delta>0$ is arbitrary, the anchor-centered radius-defining ball contains at most $\beta m$ corrupted component estimates with probability tending to one.
Since $B_{r(\widetilde{\psi}_{i^*k})}(\widetilde{\psi}_{i^*k};c_{\lambda})$ contains at least $\lceil m/2\rceil$ assigned estimates by definition of $r(\widetilde{\psi}_{i^*k})$, the same cost ball must contain at least $s_m:=\left\lceil (1/2-\beta)m\right\rceil$ failure-free estimates of the $k$th component.

The anchor estimate itself may be one of these failure-free estimates, so we remove it from the count.
Conditional on the anchor component $\widetilde{\psi}_{i^*k}$, the remaining
authentic estimates $\{\widehat{\psi}_{lk}:l\in\gI_k,\ l\neq i^*\}$ are IID.
Let $l_1,\cdots,l_{m_k-1}$ be a permutation of
$\gI_k\backslash\{i^*\}$ such that
\[
    c_{\lambda}(\widehat{\psi}_{l_1k},\widetilde{\psi}_{i^*k})
    \leq
    \cdots
    \leq
    c_{\lambda}(\widehat{\psi}_{l_{m_k-1}k},\widetilde{\psi}_{i^*k}).
\]
Then
\[
    r(\widetilde{\psi}_{i^*k})
    \ge
    c_{\lambda}(\widehat{\psi}_{l_{s_m-1}k},\widetilde{\psi}_{i^*k})
\]
whenever $s_m\ge 2$.

Set $a_0=(1/2-\beta)/2>0$. Since $1-\alpha_k\le1$, uniformly over $k\in[K]$,
\[
\frac{s_m-1}{(1-\alpha_k)m-1}
= \frac{1/2-\beta}{1-\alpha_k}+O(m^{-1})
\ge \frac34(1/2-\beta)>a_0
\]
for all sufficiently large $m$. Thus $a_0$ is a deterministic quantile level lying strictly below the relevant order-statistic fraction, even when $\alpha_k$ varies with $m$.
Lemma~\ref{lemma:concentration_of_order_statistics}, applied conditionally on $\widetilde{\psi}_{i^*k}$, gives with probability tending to one that
\[
    c_{\lambda}(\widehat{\psi}_{l_{s_m-1}k},\widetilde{\psi}_{i^*k})
    \ge q_{n,k}^{i^*}(a_0).
\]

Step 2 reduces the lower bound for the empirical radius to a lower bound for the conditional population quantile $q_{n,k}^{i^*}(a_0)$.
The only remaining issue is that this quantile is centered at the random anchor component $\widetilde{\psi}_{i^*k}$ rather than at the true component $\psi_k^*$.
The next step proves a deterministic-in-the-center quantile bound for any center $\psi$ in a small neighborhood of $\psi_k^*$; Step 4 then applies it to $\psi=\widetilde{\psi}_{i^*k}$ using the consistency of the failure-free anchor.

\noindent\underline{\textit{Step 3: Lower bound for quantile functions for off-center cost.}}
Consider a fixed parameter $\psi$ satisfying $\|\psi-\psi^*_k\|<\epsilon$, where $\epsilon$ is the neighborhood radius in Assumption~\ref{assump:reverse_bregman_divergence}.
We lower bound the population quantile of $c_{\lambda}(\widehat{\psi}_{ik},\psi)$ for any fixed $i\in\gI_k$ and any fixed $a\in(0,1)$.
Write $\mI(\psi_k^*)$ as $\mI_k$ and let $Y=Z^{\top}\mI_k^{-1}Z$, where $Z$ is a standard Gaussian vector in $\sR^{p+1}$.
As in Lemma~\ref{lemma:concentration_of_cost_function}, let
\[
    q^*(a)=\inf\{x:\sP(Y\leq x)\geq a\}
\]
denote the quantile function of this generalized chi-square random variable.
Let $t\geq 0$ be any threshold.
Following the argument in Lemma~\ref{lemma:concentration_of_cost_function}, set
$U_{ik}=\sqrt{n}\mI_k^{1/2}(\widehat{\psi}_{ik}-\psi_k^*)$.
For any $x>0$, the Euclidean ball $\{u:\|u\|\le \sqrt{p+1}+x\}$ is convex, so Corollary~\ref{cor:Berry-Esseen_for_component_estimates} can be applied to this ball and then subtracted from one.
Together with the Gaussian concentration bound
\[
    \sP\{\|Z\|\ge \sqrt{p+1}+x\}\le \exp(-x^2/2),
\]
this gives
\[
    \sP\{\|U_{ik}\|\ge \sqrt{p+1}+x\}
    \le
    \exp(-x^2/2)+Cn^{-1/2}.
\]
Taking $x=\sqrt{\log n}$ and using
\[
    \sqrt{n}\|\widehat{\psi}_{ik}-\psi_k^*\|
    \le \lambda_{\min}(\mI_k)^{-1/2}\|U_{ik}\|,
\]
we obtain
\[
\sP\left(
\sqrt{n}\|\widehat{\psi}_{ik}-\psi^*_k\|
\geq
\lambda_{\min}(\mI_k)^{-1/2}\{\sqrt{p+1}+\sqrt{\log n}\}
\right)
\leq Cn^{-1/2}.
\]
For sufficiently large $n$, $\epsilon\sqrt n\ge \lambda_{\min}(\mI_k)^{-1/2}\{\sqrt{p+1}+\sqrt{\log n}\}$, and hence $\|\widehat{\psi}_{ik}-\psi^*_k\|<\epsilon$ on an event with probability at least $1-Cn^{-1/2}$.
On the event $\{\|\widehat{\psi}_{ik}-\psi^*_k\|<\epsilon\}$,
\[\eta'_-\|\widehat{\psi}_{ik}-\psi\|^2\leq c_{\lambda}(\widehat{\psi}_{ik},\psi) \leq \eta'_+\|\widehat{\psi}_{ik}-\psi\|^2.\]
So  
\[
\sP(nc_{\lambda}(\widehat{\psi}_{ik},\psi)\leq t) \leq \sP(n\eta'_-\|\widehat{\psi}_{ik}-\psi\|^2\leq t)+\frac{C}{\sqrt{n}} \leq\sP\left((Z+v)^{\top}\mI_k^{-1}(Z+v)\leq \frac{t}{\eta'_-}\right)+\frac{C}{\sqrt{n}},
\]
where $v=n^{1/2}\mI_k^{1/2}(\psi^*_k-\psi)$.
Applying Lemma~\ref{lemma:maximum_probability_of_convex_sublevel_sets}, the right-hand side is bounded above by
\[\sP\left(Z^{\top}\mI_k^{-1}Z\leq \frac{t}{\eta'_-}\right)+\frac{C}{\sqrt{n}},\]
and therefore
	the $a$th population quantile of $c_{\lambda}(\widehat{\psi}_{ik},\psi)$ is at least
\[
\frac{\eta'_-}{n}q^*(a-Cn^{-1/2}),
\]
for any fixed $a\in(0,1)$ with $a-Cn^{-1/2}>0$.

\noindent\underline{\textit{Step 4: Putting everything together.}}
From \eqref{eq:universal_bound}, the failure-free anchor satisfies $\|\widetilde{\psi}_{i^*k}-\psi_k^*\|\leq \epsilon'$ for some $0<\epsilon'<\epsilon$ with probability tending to one.
Thus Step 3 applies to $\psi=\widetilde{\psi}_{i^*k}$ and gives
\[
q_{n,k}^{i^*}(a_0)\geq \frac{\eta'_-}{n}q^*(a_0-Cn^{-1/2}).
\]

Since $a_0>0$ is fixed, the right-hand side is $\Omega(n^{-1})$ for all sufficiently large $n$.
Combining the quantile lower bound with the order-statistic comparison in Step 2 gives $r(\widetilde{\psi}_{i^*k})=\Omega_P(n^{-1})$.
Combining the lower bound with the upper bound from Step 1 proves $r(\widetilde{\psi}_{i^*k})=\Theta_P(n^{-1})$.
\end{proof}

%%%%%%%%%%%%%%%%%%%%%%%%%%%%%%%%%%%%%%%%%%%%%%%%%%%%%%
\subsubsection{Data-driven selection of a failure-free anchor}
\label{app:data_driven_anchor_selection}
The preceding part shows what would happen if a failure-free anchor were already known.
The purpose of this part is to justify the data-driven anchor selection step in Algorithm~\ref{alg:initial_estimate}.
For each candidate machine $i\in[m]$, write
$\widetilde{\bm\psi}_i=(\widetilde{\psi}_{i1},\ldots,\widetilde{\psi}_{iK})$ for the transmitted local mixture used as a candidate anchor.
After assigning all transmitted component estimates to their nearest components of this candidate anchor, write $\sC_k(\widetilde{\bm\psi}_i)$ for the corresponding index set.
If $|\sC_k(\widetilde{\bm\psi}_i)|\ge \lceil m/2\rceil$, let $r(\widetilde{\psi}_{ik})$ be the $\lceil m/2\rceil$th order statistic of $\{c_{\lambda}(\widetilde{\psi}_{rj},\widetilde{\psi}_{ik}):(r,j)\in\sC_k(\widetilde{\bm\psi}_i)\}$; otherwise set $r(\widetilde{\psi}_{ik})=\infty$.
The total radius score of candidate machine $i$ is
\[
    R(\widetilde{\bm\psi}_i)=\sum_{k=1}^{K}r(\widetilde{\psi}_{ik}).
\]
The proof compares the total radius score of failure-free and Byzantine-failure candidate anchors.
By Lemma~\ref{lemma:accurate_alignment} and the upper-bound argument in Lemma~\ref{lemma:radius_concentration}, a failure-free anchor yields correctly aligned clusters and a total radius score of order $n^{-1}$.
Together with Assumption~\ref{assump:one_failure-free_machine-method}, the
component-wise majority condition, and
Assumption~\ref{assump:anchor_separability}, this comparison proves
Theorem~\ref{theorem:accurate_selection}: Byzantine-failure anchors have
larger radius scores with high probability, so minimizing
$R(\widetilde{\bm\psi}_i)$ selects a completely failure-free machine.
The condition $mn^{-q/2}\to0$ enters through uniform localization at the
fixed component-separation scale in Lemma~\ref{lemma:accurate_alignment}.
No stronger growth restriction on $m$ is needed for selecting a failure-free anchor.
The stronger conditions such as $m=O(n)$ or $m\le n$ are used later in the final convergence-rate analysis, in particular for the oracle aggregation rate and the chosen inflation factor $\rho$.

\begin{proof}[Proof of Theorem~\ref{theorem:accurate_selection}]
We compare the best objective value among failure-free machines with the best objective value among Byzantine-failure machines.
By Assumption~\ref{assump:one_failure-free_machine-method}, there exists at least one failure-free machine $i^0$.
If $\gB=\emptyset$, then every candidate machine is failure-free and the theorem is immediate.
Thus, in the remainder of the proof, assume $\gB\neq\emptyset$.
We will show that
\[
    \min_{i'\in\gB}R(\widetilde{\bm\psi}_{i'})=\omega_P(n^{-1}).
\]
Thus every minimizer of $R$ is failure-free with probability tending to one.

\noindent\underline{\textit{Step 1: Upper bound of the radius summation for the failure-free machine.}}
On the uniform alignment event, all $m_k\ge(1/2+\kappa)m$ authentic
component-$k$ estimates belong to
$\sC_k(\widetilde{\bm\psi}_{i^0})$.  The majority radius is therefore bounded
by the $\lceil m/2\rceil$th authentic truth-centred cost, plus a constant
multiple of $\|\widehat\psi_{i^0k}-\psi_k^*\|^2$.  The former is
$O_P(n^{-1})$ by Lemma~\ref{lemma:concentration_of_cost_function}(a), and the
latter is $O_P(n^{-1})$ by
Corollary~\ref{cor:Berry-Esseen_for_component_estimates}(c).  Hence
$r(\widetilde{\psi}_{i^0k})=O_P(n^{-1})$ for every $k\in[K]$.
Since $K$ is fixed, we therefore have
\[
R(\widetilde{\bm\psi}_{i^0})=\sum_{k=1}^Kr(\widetilde{\psi}_{i^0k})=O_P(n^{-1}).
\] 

\noindent\underline{\textit{Step 2: Lower bound of the radius summation for the Byzantine-failure machine.}}
Fix an arbitrary constant $C>0$.
By Assumption~\ref{assump:anchor_separability}, the following simultaneous event has probability tending to one:
\[
\gE_C
=
\left\{
\forall i'\in\gB,\ \exists k=k(i')\in[K]\ \text{such that}\
\left|
\left\{
(i,j)\in\sC_k(\widetilde{\bm\psi}_{i'}):
c_{\lambda}(\widetilde{\psi}_{ij},\widetilde{\psi}_{i'k})\le Cn^{-1}
\right\}
\right|
<\left\lceil \frac{m}{2}\right\rceil
\right\}.
\]
The quantifier ``$\forall i'\in\gB$'' is important here: the event is uniform over all Byzantine-failure machines, so the argument below remains valid even when $|\gB|$ diverges with $m$.
On $\gE_C$, for every Byzantine-failure machine $i'\in\gB$ there exists a component $k=k(i')$ such that
\[
\left|
\left\{
(i,j)\in\sC_k(\widetilde{\bm\psi}_{i'}):
c_{\lambda}(\widetilde{\psi}_{ij},\widetilde{\psi}_{i'k})\le Cn^{-1}
\right\}
\right|
<\left\lceil \frac{m}{2}\right\rceil .
\]
Because fewer than $\lceil m/2\rceil$ assigned component estimates lie within cost radius $Cn^{-1}$ of $\widetilde{\psi}_{i'k}$, the definition of the half-sample radius gives
\[
    r(\widetilde{\psi}_{i'k})>Cn^{-1}.
\]
Since the radius score of candidate machine $i'$ is the sum of its component radii,
\[
    R(\widetilde{\bm\psi}_{i'})
    =
    \sum_{\ell=1}^{K}r(\widetilde{\psi}_{i'\ell}),
\]
and all component radii are nonnegative, the single lower bound
$r(\widetilde{\psi}_{i'k})>Cn^{-1}$ implies
$R(\widetilde{\bm\psi}_{i'})>Cn^{-1}$.
The preceding argument holds for every $i'\in\gB$ on the same event. Therefore,
\[
    \min_{i'\in\gB}R(\widetilde{\bm\psi}_{i'})
    > Cn^{-1}.
\]
	Because the event above has probability tending to one, we have shown that, for this fixed $C>0$,
	\begin{equation}
	\label{eq:byzantine_radius_score_fixed_C}
	    \sP\left\{\min_{i'\in\gB}R(\widetilde{\bm\psi}_{i'})>Cn^{-1}\right\}\to 1.
	\end{equation}
Since the same argument holds for every fixed $C>0$, the definition of $\omega_P$ gives
\[
    \min_{i'\in\gB}R(\widetilde{\bm\psi}_{i'})=\omega_P(n^{-1}).
\]
	Since the failure-free machine achieves $R(\widetilde{\bm\psi}_{i^0}) = O_P(n^{-1})$, for any $\varepsilon>0$ we can choose a fixed constant $C_{\varepsilon}>0$ such that
	\begin{equation}
	\label{eq:failure_free_radius_score_upper}
	    \sP\{R(\widetilde{\bm\psi}_{i^0})\le C_{\varepsilon}n^{-1}\}\ge 1-\varepsilon
	\end{equation}
	for all sufficiently large $n$.
	Applying \eqref{eq:byzantine_radius_score_fixed_C} with $C=C_{\varepsilon}$ and combining it with \eqref{eq:failure_free_radius_score_upper} shows that, with probability at least $1-\varepsilon-o(1)$, the failure-free candidate $i^0$ has a strictly smaller radius score than every Byzantine-failure candidate.
Because $\varepsilon$ is arbitrary, every minimizer of $R$ is completely failure-free with probability tending to one as $n,m \to \infty$.
\end{proof}

%%%%%%%%%%%%%%%%%%%%%%%%%%%%%%%%%%%%%%%%%%%%%%%%%%%%%%
\subsubsection{Final convergence rate of \ours{}}
\label{app:final_convergence_rate}
This subsection proves Theorem~\ref{theorem:rate_of_convergence}, together
with its corollaries.
The proof first establishes the two properties needed for the final decomposition: the selected anchor has the same alignment, radius, and localization properties as a failure-free anchor, and the subsequent filter keeps almost all failure-free component estimates while localizing any corrupted estimates that remain.

We begin with the consequences of the data-driven anchor selection step.
The goal is to show that, although $i^*$ is chosen from the transmitted estimates, it can be treated on a high-probability event as a failure-free anchor: its induced clusters contain the correctly aligned failure-free component estimates, its majority radii are of order $n^{-1}$, and its component estimates are close enough to the truth for the local cost comparisons used in the filtering argument.

The constants in the next two lemmas are fixed in the following order.  The
majority and capacity margins $\kappa,\beta,\beta_\dagger$ and the local
geometry constants are fixed by the assumptions.  In the tail-trimming
localization argument, a probability tolerance $\varepsilon$ first determines
$C_\varepsilon$, then $R_\varepsilon$, and finally
$A_\varepsilon$.  For the lower-radius argument, $\tau_0$ is fixed from
$\beta$, the off-centre quantile bound determines $a_k$ and
$a=\min_k a_k$, and only then are the true-centred radius constants $C_0,C_1$
formed.  The lower-radius proof uses its own localization multiplier
$b_{\rm rad}=1$.  After $a$ has been obtained, the filtering proof chooses
separate smaller constants $b_{\rm filt},c_{\rm near}$.
Assumption~\ref{assump:component_capacity} holds for
every fixed radius multiplier $C$, so invoking it after these choices creates
no circular dependence.
\begin{lemma}[Selected-anchor event]
\label{lemma:selected_anchor_event}
		Recall from Section~\ref{sec:construction_of_CFMR} that $\sC_k(\widetilde{\bm\psi}_{i^*})$ denotes the cluster induced by the cost-based assignment rule in \eqref{eq:cost_based_alignment}, $r(\widetilde{\psi}_{i^*k})$ denotes the corresponding majority radius in \eqref{eq:majority_radius}, and $\sO_k=\{(i,k):i\in\gI_k\}$ is the latent analytical set of authentic occurrences corresponding to $\psi_k^*$.
		Under Assumptions~\ref{assump:component_wise_majority},
        \ref{assump:one_failure-free_machine-method},
        \ref{assump:compact_parameter_space}--\ref{assump:smoothness},
        \ref{assump:reverse_bregman_divergence},
        \ref{assump:anchor_separability}, and part~(a) of
        Assumption~\ref{assump:component_capacity}, assume independent local data splits
        and $mn^{-q/2}\to0$.
		Let $\rho$ be deterministic with $\rho\to\infty$ and $\rho/n\to0$.
		Let $i^*$ be any minimizer of the radius score $R(\widetilde{\bm\psi}_i)$ over $i\in[m]$.
		Then the following conclusions hold:
	\begin{enumerate}[label=(\alph*)]
	\item \textit{Correct assignment of failure-free components}: $\sP\left(
	\bigcap_{k=1}^K
	\{\sO_k\subseteq\sC_k(\widetilde{\bm\psi}_{i^*})\}
	\right)\to1$.
		\item \textit{Upper control of selected-anchor radii}:
        $\max_{k\in[K]}r(\widetilde{\psi}_{i^*k})=O_P(n^{-1})$.

		    \item \textit{Exact-retention localization}: if
        $m\rho^{-q/2}\to0$, then for every fixed $b>0$,
        $\sP\left(
			\max_{k\in[K]}\|\widetilde{\psi}_{i^*k}-\psi_k^*\|^2
			\le b\rho n^{-1}
			\right)\to1$.
		    \item \textit{Tail-trimming localization}: if, instead,
        Assumption~\ref{assump:component_capacity}(b) holds, then
        \[
        \max_{k\in[K]}\|\widetilde\psi_{i^*k}-\psi_k^*\|
        =O_P(n^{-1/2}).
        \]
			\item \textit{Lower control of selected-anchor radii}: under either
        $m\rho^{-q/2}\to0$ or
        Assumption~\ref{assump:component_capacity}(b), there exists a
        constant $a>0$ such that $\sP\left(
			\min_{k\in[K]}r(\widetilde{\psi}_{i^*k})\ge a n^{-1}
			\right)\to1$.
	\end{enumerate}
	\end{lemma}
\begin{remark}
Lemma~\ref{lemma:selected_anchor_event} collects the properties of the data-selected anchor that are needed in the final rate proof.
Part (a) says that the data-selected anchor inherits the uniform alignment guarantee for failure-free anchors.
Part (b) gives the selected-radius upper bound.  Part (c) is the exact-retention
union-bound localization, while part (d) uses authentic-anchor capacity and an
off-centre count to obtain the tail-trimming root-$n$ localization.  Either route is
sufficient for the lower-radius argument in part (e).  The condition
$\rho/n\to0$ keeps the filtering neighborhood inside the local regularity
region.  The eighth-moment choice is recovered by setting $q=8$; no exponent
$1/4$ is intrinsic to the algorithm.
Part~(d) also shows that radius-score minimization regularizes the choice among
authentic candidates: an unusually noisy authentic component cannot produce
a majority-sized ball with an authentic-scale radius unless the
authentic-anchor capacity condition is violated.
Together, these statements allow the subsequent filtering analysis to condition on a single high-probability event without repeatedly separating anchor selection, alignment, radius control, and localization.
\end{remark}

\begin{proof}[Proof of Lemma~\ref{lemma:selected_anchor_event}]
\noindent\underline{\textit{Step 1: Accurate alignment.}}
The condition $mn^{-q/2}\to0$ is exactly the condition under which
Corollary~\ref{cor:uniform_accurate_alignment} gives the uniform alignment
event over all failure-free candidate anchors:
\[
\sP\left(
\bigcap_{i\in\gB^c}\bigcap_{k=1}^K
\{\sO_k\subseteq\sC_k(\widetilde{\bm\psi}_i)\}
\right)\to1.
\]
Intersecting this event with $\{i^*\in\gB^c\}$ in Theorem~\ref{theorem:accurate_selection} gives
\[
\sP\left(
\bigcap_{k=1}^K
\{\sO_k\subseteq\sC_k(\widetilde{\bm\psi}_{i^*})\}
\right)\to1,
\]
which proves part (a).

\noindent\underline{\textit{Step 2: Upper radius control for the selected anchor.}}
By Assumption~\ref{assump:one_failure-free_machine-method}, there exists a completely failure-free machine $i^0$.
The upper-bound argument in Lemma~\ref{lemma:radius_concentration}, applied to this fixed failure-free machine, gives
\[
R(\widetilde{\bm\psi}_{i^0})=O_P(n^{-1}).
\]
Since $i^*$ minimizes the radius score,
\[
R(\widetilde{\bm\psi}_{i^*})\le R(\widetilde{\bm\psi}_{i^0})=O_P(n^{-1}).
\]
All component radii are nonnegative, hence
\[
\max_{k\in[K]}r(\widetilde\psi_{i^*k}) \le R(\widetilde{\bm\psi}_{i^*}) =O_P(n^{-1}).
\]
This is exactly part (b).

\noindent\underline{\textit{Step 3: Exact-retention localization of the selected anchor.}}
For a failure-free machine $i\in\gB^c$, $\widetilde{\psi}_{ik}=\widehat{\psi}_{ik}$.
Corollary~\ref{cor:Berry-Esseen_for_component_estimates}(c) and Markov's inequality give, for every fixed $b>0$,
\[
\sP\left\{
\|\widehat\psi_{ik}-\psi_k^*\|^2>b\rho n^{-1}
\right\}
\le
\frac{\sE\|\widehat\psi_{ik}-\psi_k^*\|^q}
{b^{q/2}\rho^{q/2}n^{-q/2}}
=O(\rho^{-q/2}).
\]
A union bound over at most $Km$ failure-free components yields
\[
\sP\left\{
\max_{i\in\gB^c}\max_{k\in[K]}
\|\widehat\psi_{ik}-\psi_k^*\|^2>b\rho n^{-1}
\right\}
=O(m\rho^{-q/2})=o(1).
\]
Intersecting this event with the anchor-selection event $\{i^*\in\gB^c\}$ proves part (c).

\noindent\underline{\textit{Step 4: Root-$n$ localization for tail trimming.}}
Assume now that Assumption~\ref{assump:component_capacity}(b) holds; the
condition $m\rho^{-q/2}\to0$ is not used in this step.  Fix $\varepsilon>0$.
By Step 2 and the definition of stochastic boundedness, there is a fixed
$C_\varepsilon<\infty$ such that
\[
\sP\{R(\widetilde{\bm\psi}_{i^*})\le C_\varepsilon n^{-1}\}
\ge1-\varepsilon
\]
for all sufficiently large $n,m$.  On this event, every selected component
radius is at most $C_\varepsilon/n$.

Let $\varepsilon_0>0$ be a fixed neighborhood radius small enough for the
local cost comparison.  The uniform coarse-localization event from the proof
of Lemma~\ref{lemma:accurate_alignment} gives
\[
\max_{i\in\gB^c}\max_{k\le K}
\|\widehat\psi_{ik}-\psi_k^*\|\le\varepsilon_0
\]
with probability tending to one.  Work on this event, the uniform alignment
event, and the authentic-anchor capacity event with $C=C_\varepsilon$.

Fix a fully authentic candidate $i\in\gB^c$ and component $k$.  Suppose
$r(\widetilde\psi_{ik})\le C_\varepsilon/n$ and
\[
\|\widehat\psi_{ik}-\psi_k^*\|>A n^{-1/2}
\]
for a constant $A$ to be chosen.  The radius-defining ball contains at least
$\lceil m/2\rceil$ assigned estimates.  By
Assumption~\ref{assump:component_capacity}(b), at most
$\beta_\dagger m$ of all corrupted estimates lie in this ball.  Uniform
alignment implies that every authentic estimate assigned to cluster $k$ has
population label $k$.  Hence at least
\[
s_m^\dagger=\left\lceil(1/2-\beta_\dagger)m\right\rceil
\]
authentic component-$k$ estimates lie in the ball.

Set $R_\varepsilon=(C_\varepsilon/\eta_-')^{1/2}$.  For another authentic
machine $\ell\in\gI_k\setminus\{i\}$, local cost equivalence and the reverse
triangle inequality show that
\[
c_\lambda(\widehat\psi_{\ell k},\widehat\psi_{ik})
\le C_\varepsilon n^{-1}
\quad\Longrightarrow\quad
\sqrt n\|\widehat\psi_{\ell k}-\psi_k^*\|\ge A-R_\varepsilon.
\]
Here the local comparison is legitimate uniformly: on the coarse-localization
event the candidate center lies in a fixed smaller neighborhood of
$\psi_k^*$, and compactness plus positivity of the divergence away from the
diagonal force every point in its $C_\varepsilon/n$ cost ball into the same
local chart for all sufficiently large $n$.
Conditional on $\widehat\psi_{ik}$, the remaining authentic local estimates
are independent.  The $q$th-moment bound therefore gives, uniformly over the
conditioning value,
\[
p_A:=\sup_{\ell\in\gI_k\setminus\{i\}}
\sP\{c_\lambda(\widehat\psi_{\ell k},\widehat\psi_{ik})
\le C_\varepsilon n^{-1}\mid\widehat\psi_{ik}\}
\le C_q(A-R_\varepsilon)^{-q}.
\]
Choose $A=A_\varepsilon>R_\varepsilon$ so large that
$p_A<(1/2-\beta_\dagger)/4$.  Hoeffding's inequality, conditional on the
candidate anchor, then gives constants $c_\varepsilon,C<\infty$ for which
\[
\sP\left\{
\sum_{\ell\in\gI_k\setminus\{i\}}
\mathbbm1\{c_\lambda(\widehat\psi_{\ell k},\widehat\psi_{ik})
\le C_\varepsilon n^{-1}\}
\ge s_m^\dagger-1
\right\}
\le C\exp(-c_\varepsilon m).
\]
A union bound over at most $Km$ fully authentic candidate-component pairs
shows that, with probability tending to one, every fully authentic candidate
whose component radii are all at most $C_\varepsilon/n$ satisfies
\[
\max_{k\le K}\|\widehat\psi_{ik}-\psi_k^*\|
\le A_\varepsilon n^{-1/2}.
\]
The selected anchor is fully authentic with probability tending to one and
has this radius bound with probability at least $1-\varepsilon-o(1)$.  Since
$\varepsilon$ is arbitrary, part (d) follows.

\noindent\underline{\textit{Step 5: Lower radius control for the selected anchor.}}
Let
\[
s_m=\left\lceil(1/2-\beta)m\right\rceil,
\qquad
\tau_0=(1/2-\beta)/2\in(0,1).
\]
As in the preceding radius argument, $\tau_0$ lies strictly below
$(s_m-1)/\{(1-\alpha_k)m-1\}$ for every $k$ and all sufficiently large $m$.
The deterministic-in-the-center quantile bound from Step~3 of
Lemma~\ref{lemma:radius_concentration}, applied at this fixed level, gives
$a_k>0$ such that, uniformly over centers in the local neighborhood of
$\psi_k^*$, the $\tau_0$-quantile of
$c_\lambda(\widehat\psi_{\ell k},\psi)$ is at least $a_kn^{-1}$ for all
sufficiently large $n$.  Set
\[
a=\min_{k\in[K]}a_k>0
\]
and fix $b_{\rm rad}=1$ in the localization event below.  Thus $a$ and
$b_{\rm rad}$ are fixed before the capacity radius constants $C_0,C_1$ are
formed.

We first collect the events needed in the lower-bound argument:
\[
\gE_{\mathrm{sel}}=\{i^*\in\gB^c\},
\qquad \gE_{\mathrm{align}}^{\mathrm{unif}} = \bigcap_{i\in\gB^c}\bigcap_{k=1}^K \{\sO_k\subseteq\sC_k(\widetilde{\bm\psi}_i)\},
\quad
\text{and}
\quad
\gE_{\mathrm{loc}}^{\rm rad}
=\left\{\max_{k\in[K]}\|\widetilde{\psi}_{i^*k}-\psi_k^*\|^2
\le b_{\rm rad}\rho n^{-1}\right\}.
\]
Theorem~\ref{theorem:accurate_selection} and
Corollary~\ref{cor:uniform_accurate_alignment} give probability tending to one
for the first two events.  Under the exact-retention regime, part (c) gives
$\sP(\gE_{\mathrm{loc}}^{\rm rad})\to1$ directly.  Under the tail-trimming regime,
part (d), together with $\rho\to\infty$, gives
\[
\max_k\|\widetilde\psi_{i^*k}-\psi_k^*\|^2/(\rho/n)=o_P(1),
\]
and hence the same conclusion for the fixed $b_{\rm rad}>0$.
On $\gE_{\mathrm{sel}}\cap\gE_{\mathrm{align}}^{\mathrm{unif}}$, the selected anchor $i^*$ is one of these failure-free anchors, and hence it satisfies $\bigcap_{k=1}^K\{\sO_k\subseteq\sC_k(\widetilde{\bm\psi}_{i^*})\}$.
Hence, the intersection
$\gE_0^{\rm rad}=\gE_{\mathrm{sel}}\cap
\gE_{\mathrm{align}}^{\mathrm{unif}}\cap\gE_{\mathrm{loc}}^{\rm rad}$ has
probability tending to one.  In the rest of the proof, we work on
$\gE_0^{\rm rad}$.

Fix $k\in[K]$.
If $r(\widetilde{\psi}_{i^*k})<a n^{-1}$, then the anchor-centered ball $B_{a n^{-1}}(\widetilde{\psi}_{i^*k};c_\lambda)$ contains at least $\lceil m/2\rceil$ assigned transmitted estimates.
We first show that any corrupted component in this anchor-centered ball must also lie in a slightly larger true-centered ball.
On the event $\gE_{\mathrm{loc}}^{\rm rad}$ defined above,
\[
\|\widetilde{\psi}_{i^*k}-\psi_k^*\|^2
\le b_{\rm rad}\rho n^{-1}.
\]
Because $\rho n^{-1}\to0$, the selected anchor component lies in the local neighborhood of $\psi_k^*$ for all sufficiently large $n$ on this event.
If a corrupted component $\xi_{i\ell}$ belongs to $B_{a n^{-1}}(\widetilde{\psi}_{i^*k};c_\lambda)$, then for sufficiently large $n$ both $\xi_{i\ell}$ and $\widetilde{\psi}_{i^*k}$ lie in the local neighborhood of $\psi_k^*$; otherwise compactness and positivity of the Bregman divergence away from the diagonal would give a fixed positive lower bound on $c_\lambda(\xi_{i\ell},\widetilde{\psi}_{i^*k})$, contradicting $a n^{-1}\to0$.
Thus Corollary~\ref{cor:cost_function} gives
\[
\|\xi_{i\ell}-\widetilde{\psi}_{i^*k}\|^2
\le \eta_-^{\prime -1}a n^{-1}.
\]
By the triangle inequality,
\[
\|\xi_{i\ell}-\psi_k^*\|^2
\le
2\eta_-^{\prime -1}a n^{-1}
{}+
2b_{\rm rad}\rho n^{-1}
\le C_0\rho n^{-1}
\]
for a fixed constant $C_0>0$, since $\rho\to\infty$.
Applying the local upper bound in Corollary~\ref{cor:cost_function} once more,
every such corrupted component is contained in a true-centered ball
$B_{C_1\rho n^{-1}}(\psi_k^*;c_\lambda)$ for a fixed constant $C_1>0$
depending only on the local equivalence constants, $a$, and $b_{\rm rad}$.
Define
\[
\gE_{\mathrm{anti},k}
=
\left\{
\left|
\left\{(i,\ell):\ \ell\in[K],\ i\in\gB_\ell,\
c_{\lambda}(\xi_{i\ell},\psi_k^*)\le C_1\rho n^{-1}
\right\}
\right|
\le \beta m
\right\}.
\]
Assumption~\ref{assump:component_capacity}(a) gives
$\sP(\gE_{\mathrm{anti},k})\to1$.
Hence, on $\gE_{\mathrm{anti},k}$, the anchor-centered ball must contain at least $s_m$ failure-free estimates of the $k$th component.

It remains to show that this many failure-free estimates cannot lie in the
ball with the already fixed radius $an^{-1}$.
Fix a failure-free candidate anchor $i\in\gB^c$.
Conditional on $\widetilde{\psi}_{ik}$, the remaining authentic estimates
$\{\widehat{\psi}_{\ell k}:\ell\in\gI_k,\ \ell\neq i\}$ are IID.

Lemma~\ref{lemma:concentration_of_order_statistics}, applied conditionally on $\widetilde{\psi}_{ik}$ to the IID costs $\{c_{\lambda}(\widehat{\psi}_{\ell k},\widetilde{\psi}_{ik}):\ell\in\gI_k,\ \ell\neq i\}$, then implies that, for some constant $c_k>0$,
\[
\sup_{\psi:\|\psi-\psi_k^*\|<\epsilon}
\sP\left(
    \sum_{\ell\in\gI_k,\ \ell\neq i}
    \mathbbm{1}\{c_{\lambda}(\widehat{\psi}_{\ell k},\psi)\le a_k n^{-1}\}
    \ge s_m-1
\right)
\le 2\exp(-c_k m).
\]
Set $c_0=\min_{k\in[K]}c_k>0$.
Since $a\le a_k$ for all $k\in[K]$, a union bound gives
\[
\sP\left(\gE_{\mathrm{ord}}^c\right)
:=
\sP\left(
    \bigcup_{k=1}^K
    \bigcup_{i\in\gB^c}
    \left\{
    \|\widetilde{\psi}_{ik}-\psi_k^*\|<\epsilon,\ 
    \sum_{\ell\in\gI_k,\ \ell\neq i}
    \mathbbm{1}\{c_{\lambda}(\widehat{\psi}_{\ell k},\widetilde{\psi}_{ik})\le a n^{-1}\}
    \ge s_m-1
    \right\}
\right)
\le 2Km\exp(-c_0m)=o(1).
\]
Thus $\sP(\gE_{\mathrm{ord}})\to1$.
We now connect this event to the desired radius lower bound.
On $\gE_0^{\rm rad}\cap\gE_{\mathrm{anti},k}$, the selected anchor is failure-free and satisfies
\[
\|\widetilde{\psi}_{i^*k}-\psi_k^*\|^2
\le b_{\rm rad}\rho n^{-1}=o(1),
\]
so $\|\widetilde{\psi}_{i^*k}-\psi_k^*\|<\epsilon$ for all sufficiently large $n$.
If $r(\widetilde{\psi}_{i^*k})<a n^{-1}$, then $B_{a n^{-1}}(\widetilde{\psi}_{i^*k};c_\lambda)$ contains at least $\lceil m/2\rceil$ assigned transmitted estimates.
By the preceding true-centered-ball argument and the definition of $\gE_{\mathrm{anti},k}$, at most $\beta m$ of these estimates can be corrupted.
Therefore, the same anchor-centered ball must contain at least $s_m$ failure-free estimates of the $k$th component.
At most one of these failure-free estimates is the anchor component itself, so the remaining failure-free estimates satisfy
\[
\sum_{\ell\in\gI_k,\ \ell\neq i^*}
\mathbbm{1}\{c_{\lambda}(\widehat{\psi}_{\ell k},\widetilde{\psi}_{i^*k})\le a n^{-1}\}
\ge s_m-1.
\]
Thus the pair $(k,i^*)$ belongs to the union event excluded by $\gE_{\mathrm{ord}}$, which contradicts $\gE_{\mathrm{ord}}$.
Since $\sP(\gE_0^{\rm rad}\cap\gE_{\mathrm{anti},k}\cap
\gE_{\mathrm{ord}})\to1$, this contradiction shows that the event
$\{r(\widetilde{\psi}_{i^*k})<a n^{-1}\}$ has probability tending to zero.
Hence, for this fixed $k$,
\[
\sP\{r(\widetilde{\psi}_{i^*k})<a n^{-1}\}\to0.
\]
Taking a union bound over $k\in[K]$ proves part (e).
\end{proof}

We next describe the behavior of the filtering step after the selected anchor has been fixed.
\begin{lemma}[Filtering event]
\label{lemma:filtering_event}
Under the common assumptions of Lemma~\ref{lemma:selected_anchor_event}, assume
either the exact-retention localization condition $m\rho^{-q/2}\to0$ or the
tail-trimming condition in Assumption~\ref{assump:component_capacity}(b), and consider the
selected anchor $i^*$.
For each $k\in[K]$, abbreviate the anchor-induced cluster and its majority radius by $\sC_k=\sC_k(\widetilde{\bm\psi}_{i^*})$ and $r_k=r(\widetilde{\psi}_{i^*k})$.
The filter retains the indices $\widehat{\sO}_k
=
\{(i,j)\in\sC_k:
c_{\lambda}(\widetilde{\psi}_{ij},\widetilde{\psi}_{i^*k})\leq \rho r_k\}$.
Moreover, with probability tending to one, simultaneously for every
$k\in[K]$,
\begin{equation}
\label{eq:retained_nonoracle_are_corrupted}
\widehat\sO_k\setminus\sO_k
=\left\{(i,\ell)\in\widehat\sO_k:i\in\gB_\ell\right\}.
\end{equation}
For any constant $c>0$, define the ideal true-centered near-clean set
\[
\sO_k^{\mathrm{near}}(c,\rho)
=
\left\{
(i,k)\in\sO_k:
c_{\lambda}(\widetilde{\psi}_{ik},\psi_k^*)\le c\rho n^{-1}
\right\}.
\]
Then the following conclusions hold:
\begin{enumerate}[label=(\alph*)]
\item \textit{Retention of near-clean estimates}:
there exists a constant $c>0$ such that the ideal true-centered near-clean set is contained in the data-selected filter:
\[
\sP\left(
\bigcap_{k=1}^K
\left\{
\sO_k^{\mathrm{near}}(c,\rho)\subseteq\widehat{\sO}_k
\right\}
\right)\to1.
\]
\item \textit{Authentic tail losses}: if
$L_k=|\sO_k\setminus\widehat{\sO}_k|$, then
\[
L_k=O_P(m\rho^{-q/2})=o_P(m)
\]
and, uniformly over fixed $k$,
\[
\frac1m\sum_{(i,k)\in\sO_k\setminus\widehat{\sO}_k}
\|\widetilde\psi_{ik}-\psi_k^*\|
=O_P\{n^{-1/2}\rho^{-(q-1)/2}\}.
\]
Consequently, $|\widehat{\sO}_k|=\Theta_P(m)$.
\item \textit{Exact authentic retention in regime~(a)}: if
$m\rho^{-q/2}\to0$, then
\[
\sP\left\{\bigcap_{k=1}^K(\sO_k\subseteq\widehat{\sO}_k)\right\}\to1.
\]
\item \textit{Localization of all retained estimates}:
let $\epsilon$ denote the radius of the local neighborhood $B_{\epsilon}(\theta_k^*;\|\cdot\|)$ in Assumption~\ref{assump:reverse_bregman_divergence}. The retained estimates remain in the local region where the cost--Euclidean comparisons are valid:
\[
\sP\left(
\bigcap_{k=1}^K
\left\{
\forall (i,j)\in\widehat{\sO}_k,\ 
\|\widetilde{\theta}_{ij}-\theta_k^*\|\le \epsilon
\right\}
\right)\to1.
\]
\item \textit{Filtering-scale localization of retained corrupted estimates}:
\[
\max_{(i,j)\in\widehat{\sO}_k\setminus\sO_k}
\|\widetilde{\psi}_{ij}-\psi_k^*\|
=O_P\{\sqrt{\rho/n}\},
\]
uniformly over fixed $k$.
Here and below, the maximum over an empty set is taken to be zero.
\end{enumerate}
\end{lemma}

\begin{proof}[Proof of Lemma~\ref{lemma:filtering_event}]
We prove the five assertions on a common sequence of high-probability events.

Lemma~\ref{lemma:selected_anchor_event}(a) gives the simultaneous correct
alignment event.  On that event, every authentic occurrence belongs to
exactly one of the disjoint sets $\sO_1,\ldots,\sO_K$, and
$\sO_\ell\subseteq\sC_\ell$ for every $\ell$.  Since the anchor-induced
clusters are disjoint, an authentic occurrence in $\sC_k$ must therefore belong
to $\sO_k$.  This proves
\eqref{eq:retained_nonoracle_are_corrupted}.

\noindent\underline{\textit{Step 1: Retention of near-clean estimates.}}
Let $a>0$ be the lower-radius constant from
Lemma~\ref{lemma:selected_anchor_event}(e).  Under either localization regime,
\[
\max_{k\le K}\|\widetilde\psi_{i^*k}-\psi_k^*\|^2
=o_P(\rho/n).
\]
Choose fixed $b_{\rm filt},c_{\rm near}>0$ so small that
$2\eta'_+b_{\rm filt}\le a/2$ and
$2\eta'_+c_{\rm near}/\eta'_-\le a/2$.  Define
\[
\gE_{\mathrm{master}}^{\rm filt}
=\{i^*\in\gB^c\}
\cap\bigcap_{k=1}^K\{\sO_k\subseteq\sC_k\}
\cap\{\min_k r_k\ge an^{-1}\}
\cap\left\{\max_k
\|\widetilde\psi_{i^*k}-\psi_k^*\|^2\le b_{\rm filt}\rho/n\right\}.
\]
Lemma~\ref{lemma:selected_anchor_event} gives
$\sP(\gE_{\mathrm{master}}^{\rm filt})\to1$.

Fix $(i,k)\in\sO_k^{\mathrm{near}}(c_{\rm near},\rho)$.  On the master event it belongs
to $\sC_k$.  Its true-centred cost is at most
$c_{\rm near}\rho/n=o(1)$, so compactness
and strict positivity of the divergence away from the diagonal put its
subpopulation parameter in the local neighborhood of $\theta_k^*$.  The
selected anchor is in the same neighborhood.  Corollary~\ref{cor:cost_function}
therefore gives
\begin{equation}
\label{eq:near_clean_anchor_cost_bound}
\begin{aligned}
c_\lambda(\widetilde\psi_{ik},\widetilde\psi_{i^*k})
&\le2\eta'_+\left(
\|\widetilde\psi_{ik}-\psi_k^*\|^2
+\|\widetilde\psi_{i^*k}-\psi_k^*\|^2\right)\\
&\le\left(2\eta'_+c_{\rm near}/\eta'_-+
2\eta'_+b_{\rm filt}\right)\rho/n
\le a\rho/n
\le\rho r_k.
\end{aligned}
\end{equation}
Thus $(i,k)\in\widehat\sO_k$, proving part (a).

\noindent\underline{\textit{Step 2: Count and size of authentic tail losses.}}
On $\gE_{\mathrm{master}}^{\rm filt}$, part (a) implies
\[
\sO_k\setminus\widehat\sO_k
\subseteq
\left\{(i,k):i\in\gI_k,
c_\lambda(\widehat\psi_{ik},\psi_k^*)>c_{\rm near}\rho/n\right\}.
\]
Let $N_k^{\mathrm{tail}}$ be the cardinality of the set on the right.  By
Lemma~\ref{lemma:concentration_of_cost_function}(b),
\[
\sE N_k^{\mathrm{tail}}\le Cm\rho^{-q/2},
\qquad
N_k^{\mathrm{tail}}=O_P(m\rho^{-q/2}).
\]

Set $e_{ik}=\widehat\psi_{ik}-\psi_k^*$.  The local cost comparison and the
$q$th-moment bound imply the truncated first-moment estimate
\begin{equation}
\label{eq:truncated_authentic_tail}
\sE\left[
\|e_{ik}\|\mathbbm1\{c_\lambda(\widehat\psi_{ik},\psi_k^*)>
c_{\rm near}\rho/n\}
\right]
\le Cn^{-1/2}\rho^{-(q-1)/2}.
\end{equation}
Indeed, on the local neighborhood the tail event forces
$\|e_{ik}\|\ge C\sqrt{\rho/n}$, and
$\sE\{X\mathbbm1(X>t)\}\le\sE(X^q)t^{-(q-1)}$; the complement of the local
neighborhood is controlled by the same inequality at a fixed threshold and is
smaller because $\rho\le n$ eventually.  Summing
\eqref{eq:truncated_authentic_tail} over $i\in\gI_k$ and applying Markov's
inequality proves the weighted tail bound in part (b).

Because $\rho\to\infty$, $N_k^{\mathrm{tail}}=o_P(m)$.  The majority margin
gives $|\sO_k|\ge(1/2+\kappa)m$, hence
\[
|\widehat\sO_k|\ge|\sO_k|-N_k^{\mathrm{tail}}=\Theta_P(m).
\]
If $m\rho^{-q/2}\to0$, then
$\sE N_k^{\mathrm{tail}}=o(1)$; since the count is integer-valued, it is zero
with probability tending to one.  A union bound over fixed $K$ proves part (c).

\noindent\underline{\textit{Step 3: Localize every retained estimate.}}
Lemma~\ref{lemma:selected_anchor_event}(b) gives
$\max_k r_k=O_P(n^{-1})$.  Consequently,
\[
\max_{k\le K}\max_{(i,j)\in\widehat\sO_k}
c_\lambda(\widetilde\psi_{ij},\widetilde\psi_{i^*k})
\le\rho\max_k r_k=O_P(\rho/n)=o_P(1).
\]
The selected anchor is $o_P(1)$ from the corresponding truth.  If a retained
subpopulation parameter stayed outside the fixed local neighborhood, compactness
and positivity of the Bregman divergence on the resulting two disjoint compact
sets would give a fixed positive lower bound on its anchor-centred cost,
contradicting the preceding display.  This proves part (d).

\noindent\underline{\textit{Step 4: Localize retained corrupted estimates.}}
On the local event from part (d), Corollary~\ref{cor:cost_function} gives
\[
\max_{(i,j)\in\widehat\sO_k\setminus\sO_k}
\|\widetilde\psi_{ij}-\widetilde\psi_{i^*k}\|^2
\le(\eta_-')^{-1}\rho r_k=O_P(\rho/n).
\]
Under either localization regime, the selected-anchor error squared is
$o_P(\rho/n)$.  The triangle inequality proves part (e), completing the
proof.
\end{proof}

\begin{lemma}[Oracle component aggregation]
\label{lemma:oracle_component_aggregation}
For $k\in[K]$, let
$m_k=|\gI_k|=|\sO_k|=(1-\alpha_k)m$. Define
\[
\widehat w_k^o=\frac1{m_k}\sum_{i\in\gI_k}\widehat w_{ik},
\qquad
\widehat\theta_k^o
=\operatorname*{arg\,min}_{\theta\in\Theta}
\sum_{i\in\gI_k}\widehat w_{ik}D(\widehat\theta_{ik},\theta).
\]
Under Assumptions~\ref{assump:compact_parameter_space}--\ref{assump:reverse_bregman_divergence} and independent local data splits,
\[
|\widehat w_k^o-w_k^*|+\|\widehat\theta_k^o-\theta_k^*\|
=O_P\{(m_kn)^{-1/2}+n^{-1}\}.
\]
Under Assumption~\ref{assump:component_wise_majority} and $m\le n$, this is $O_P(N^{-1/2})$, uniformly over the fixed number $K$ of components.
\end{lemma}

\begin{proof}
The authentic local component estimators are independent across $i\in\gI_k$. Corollary~\ref{cor:Berry-Esseen_for_component_estimates}(b)--(c) gives
\[
\left|\sE(\widehat w_{ik}-w_k^*)\right|=O(n^{-1}),
\qquad
\operatorname{Var}(\widehat w_{ik})=O(n^{-1}).
\]
Consequently, the bias--variance decomposition
\[
\widehat w_k^o-w_k^*
=\frac1{m_k}\sum_{i\in\gI_k}
\{\widehat w_{ik}-\sE\widehat w_{ik}\}
+\sE(\widehat w_{ik}-w_k^*)
\]
implies
\[
|\widehat w_k^o-w_k^*|
=O_P\{(m_kn)^{-1/2}+n^{-1}\}.
\]

For the subpopulation parameter, set $q_k(w,\theta)=w\nabla A(\theta)$ and $q_k^*=w_k^*\nabla A(\theta_k^*)$. A second-order Taylor expansion in the local regularity neighborhood, together with Corollary~\ref{cor:Berry-Esseen_for_component_estimates}(b)--(c), gives
\[
\left\|\sE\{q_k(\widehat\psi_{ik})-q_k^*\}\right\|=O(n^{-1}),
\qquad
\sE\|q_k(\widehat\psi_{ik})-q_k^*\|^2=O(n^{-1}).
\]
The contribution from outside that neighborhood is of smaller order by the
$q$th-moment bound and compactness of $\Theta$. Hence, for
\[
\overline q_k=\frac1{m_k}\sum_{i\in\gI_k}
\widehat w_{ik}\nabla A(\widehat\theta_{ik}),
\]
the same bias--variance argument yields
\[
\|\overline q_k-q_k^*\|
=O_P\{(m_kn)^{-1/2}+n^{-1}\}.
\]
Since $w_k^*>0$, the weight bound implies $\widehat w_k^o$ is bounded away from zero with probability tending to one. The oracle barycenter first-order condition is therefore
\[
\nabla A(\widehat\theta_k^o)=\frac{\overline q_k}{\widehat w_k^o}.
\]
As in the interior-minimizer argument used below for the selected estimator, strict convexity and the inverse function theorem make this condition valid and identify the unique minimizer in a neighborhood of $\theta_k^*$ with probability tending to one. Because $\nabla^2A(\theta_k^*)$ is nonsingular, the inverse of $\nabla A$ is locally Lipschitz, and thus
\[
\|\widehat\theta_k^o-\theta_k^*\|
=O_P\{(m_kn)^{-1/2}+n^{-1}\}.
\]
Finally, Assumption~\ref{assump:component_wise_majority} gives $m_k\ge(1/2+\kappa)m$, while $m\le n$ gives $n^{-1}\le(mn)^{-1/2}$. The stated $O_P(N^{-1/2})$ consequence follows.
\end{proof}

\noindent We now prove both regimes of the convergence theorem.  Their only
difference is the treatment of authentic estimates outside the near-clean
set.

\begin{proof}[Proof of Theorem~\ref{theorem:rate_of_convergence}]
Write
\[
R_{ok}=(m_kn)^{-1/2}+n^{-1},
\qquad
R_{o\max}=\max_{\ell\le K}R_{o\ell},
\qquad
b_{nq}=n^{-1/2}\rho^{-(q-1)/2}.
\]
Recall that $T_{nq}=0$ in regime~(a) and $T_{nq}=b_{nq}$ in
regime~(b).
Let
\[
\bar\alpha_k
=\frac1m|\widehat\sO_k\setminus\sO_k|,
\qquad
\bar\alpha_{\max}=\max_{k\le K}\bar\alpha_k.
\]
Lemma~\ref{lemma:filtering_event} gives
$|\widehat\sO_k|=\Theta_P(m)$ and
\[
\max_{(i,j)\in\widehat\sO_k\setminus\sO_k}
\|\widetilde\psi_{ij}-\psi_k^*\|
=O_P\{\sqrt{\rho/n}\}.
\]

\noindent\underline{\textit{Step 1: Preliminary and normalized weights.}}
Centering every summand at $w_k^*$ gives the exact decomposition
\begin{align*}
\widehat w_k'-w_k^*
=\frac1{|\widehat\sO_k|}\bigg[&
\sum_{(i,k)\in\sO_k}(\widetilde w_{ik}-w_k^*)
-\sum_{(i,k)\in\sO_k\setminus\widehat\sO_k}
(\widetilde w_{ik}-w_k^*)\\
&+\sum_{(i,j)\in\widehat\sO_k\setminus\sO_k}
(\widetilde w_{ij}-w_k^*)\bigg].
\end{align*}
The oracle aggregation lemma and $m_k/|\widehat\sO_k|=O_P(1)$ bound the first
term by $O_P(R_{ok})$.  Lemma~\ref{lemma:filtering_event}(b) bounds the second
term by $O_P(b_{nq})$.  In regime~(a) it is exactly zero with probability
tending to one by part (c), so in either regime it is $O_P(T_{nq})$.  The retained-corruption count and
localization bound the third term by
$O_P\{\bar\alpha_k\sqrt{\rho/n}\}$.  Hence
\begin{equation}
\label{eq:preliminary_weight_component_rate}
\widehat w_k'-w_k^*
=O_P\left\{R_{ok}+T_{nq}
+\bar\alpha_k\sqrt{\rho/n}\right\}.
\end{equation}

Because $K$ is fixed,
\[
\sum_{\ell=1}^K(\widehat w_\ell'-w_\ell^*)
=O_P\left\{R_{o\max}+T_{nq}
+\bar\alpha_{\max}\sqrt{\rho/n}\right\}.
\]
The random rate on the right is $o_P(1)$: the oracle and authentic-tail terms
vanish, $0\le\bar\alpha_{\max}\le K$, and $\rho/n\to0$.  Thus
$\sum_\ell\widehat w_\ell'$ is bounded away from zero with probability tending
to one.  Expanding its reciprocal in
$\widehat w_k=\widehat w_k'/\sum_\ell\widehat w_\ell'$ yields
\begin{equation}
\label{eq:CFMR_weight_rate}
|\widehat w_k-w_k^*|
=O_P\left\{R_{o\max}+T_{nq}
+\bar\alpha_{\max}\sqrt{\rho/n}\right\}.
\end{equation}

\noindent\underline{\textit{Step 2: Reverse-Bregman barycentres.}}
Set
\[
h_k(\theta)=\sum_{(i,j)\in\widehat\sO_k}
\widetilde w_{ij}D(\widetilde\theta_{ij},\theta),
\quad
W_k=\sum_{(i,j)\in\widehat\sO_k}\widetilde w_{ij},
\quad
\overline g_k=W_k^{-1}\sum_{(i,j)\in\widehat\sO_k}
\widetilde w_{ij}\nabla A(\widetilde\theta_{ij}).
\]
Every retained weight is $o_P(1)$ from $w_k^*$ uniformly: this follows
directly from the filter cost, the selected-anchor localization, and the
weight term in $c_\lambda$.  Together with
$|\widehat\sO_k|=\Theta_P(m)$ and $w_k^*>0$, this gives
$W_k=\Theta_P(m)$.

Part (d) of Lemma~\ref{lemma:filtering_event} puts all retained subpopulation
parameters in the local regularity neighborhood.  The same decomposition as
in Step 1, now applied to
\[
q_k(\psi)=w\{\nabla A(\theta)-\nabla A(\theta_k^*)\},
\]
obeys the global bound
\begin{equation}
\label{eq:global_qk_linear_bound}
\|q_k(\psi)\|\le C\|\psi-\psi_k^*\|
\end{equation}
on the compact admissible component-parameter set.  Indeed, in the local
neighborhood this follows from the Hessian bound for $A$ and $0<w\le1$.
Outside that neighborhood, $\nabla A$ is bounded by continuity and
compactness, while $\|\theta-\theta_k^*\|$ is bounded below by a positive
constant.  Thus boundedness away from the truth converts to the same linear
bound globally.
The three terms in the decomposition have the following orders after division
by $W_k$:
\begin{enumerate}[label=(\roman*)]
\item the oracle sum is $O_P(R_{ok})$ by
Lemma~\ref{lemma:oracle_component_aggregation};
\item the excluded-authentic sum is $O_P(T_{nq})$ by
Lemma~\ref{lemma:filtering_event}(b) and
\eqref{eq:global_qk_linear_bound}, with exact retention making this term zero
in regime~(a); and
\item the retained-corruption sum is
$O_P\{\bar\alpha_k\sqrt{\rho/n}\}$ by part (e) of that lemma.
\end{enumerate}
Consequently,
\[
S_k:=\overline g_k-\nabla A(\theta_k^*)
=O_P\left\{R_{ok}+T_{nq}
+\bar\alpha_k\sqrt{\rho/n}\right\}=o_P(1).
\]

Because $\nabla^2A(\theta_k^*)$ is nonsingular, the inverse function theorem
provides a unique interior point satisfying
\[
\nabla A(\widehat\theta_k)=\overline g_k.
\]
Strict convexity identifies this point with the barycentre minimizing $h_k$.
The local inverse is Lipschitz, and therefore
\begin{equation}
\label{eq:CFMR_theta_rate}
\|\widehat\theta_k-\theta_k^*\|
=O_P\left\{R_{ok}+T_{nq}
+\bar\alpha_k\sqrt{\rho/n}\right\}.
\end{equation}
Equations \eqref{eq:CFMR_weight_rate} and~\eqref{eq:CFMR_theta_rate} prove both parts of
the theorem.
\end{proof}

\begin{proof}[Proof of Corollary~\ref{cor:oracle_rate_small_ball}]
Fix $C_0>0$ and define the component-occurrence count
\[
M_{n,k}(C_0)=
\left|
\left\{(i,\ell):\ \ell\in[K],\ i\in\gB_\ell,\
c_\lambda(\xi_{i\ell},\psi_k^*)\le C_0\rho/n
\right\}
\right|.
\]
Conditional linearity of expectation and the small-ball condition give
\[
\max_{k\le K}
\sE\{M_{n,k}(C_0)\mid\mathcal F_{\mathrm{clean}}\}
\le Cm(\rho/n)^\gamma.
\]
Because $m\le n$, $\rho=m^a$, and $a<1$, we have $\rho/n\to0$.
Conditional Markov's inequality followed by a union bound over fixed $K$
shows that $\max_kM_{n,k}(C_0)/m=o_P(1)$.  Hence
Assumption~\ref{assump:component_capacity}(a) holds for any fixed
$\beta\in(0,1/2)$.

In the exact-retention regime, $a>2/q$ gives
$m\rho^{-q/2}=m^{1-aq/2}\to0$.  In the tail-trimming regime, the assumed
authentic-anchor capacity supplies the alternative localization argument and
\[
\frac{n^{-1/2}\rho^{-(q-1)/2}}{N^{-1/2}}
=m^{1/2}\rho^{-(q-1)/2}
=m^{\{1-a(q-1)\}/2}\to0.
\]
Thus the corresponding regime of Theorem~\ref{theorem:rate_of_convergence}
applies in either case.

It remains to control the realized retained-corruption fraction.  Fix
$\varepsilon>0$.  Lemmas~\ref{lemma:selected_anchor_event} and
\ref{lemma:filtering_event}, together with the local cost--Euclidean
comparison, give a fixed $C_\varepsilon<\infty$ and an event with probability
at least $1-\varepsilon-o(1)$ on which every retained corrupted component
assigned to cluster $k$ satisfies
\[
c_\lambda(\xi_{i\ell},\psi_k^*)
\le C_\varepsilon\rho/n.
\]
By \eqref{eq:retained_nonoracle_are_corrupted}, every index counted by
$m\bar\alpha_k$ is one of these corrupted component occurrences.
Consequently, on this event,
\[
m\bar\alpha_k
\le
\left|
\left\{(i,\ell):\ \ell\in[K],\ i\in\gB_\ell,\
c_\lambda(\xi_{i\ell},\psi_k^*)\le C_\varepsilon\rho/n
\right\}
\right|.
\]
The conditional small-ball bound and Markov's inequality now imply, uniformly
over the fixed number of components,
\[
\bar\alpha_{\max}=O_P\{(\rho/n)^\gamma\}.
\]
Substitution in the appropriate theorem regime, and the preceding negligibility of
the authentic-tail term in the tail-trimming regime, yields
\[
|\widehat w_k-w_k^*|+\|\widehat\theta_k-\theta_k^*\|
=O_P\left\{N^{-1/2}+(\rho/n)^{\gamma+1/2}\right\}.
\]
Finally,
\[
(\rho/n)^{\gamma+1/2}=o(N^{-1/2})
\quad\Longleftrightarrow\quad
m^{1/2}\rho^{\gamma+1/2}=o(n^\gamma),
\]
which proves the oracle-rate assertion and its equivalent power form.
\end{proof}

\begin{proof}[Proof of Corollary~\ref{cor:wasserstein_rate}]
Let $D_\Theta=\sup_{\theta,\theta'\in\Theta}\|\theta-\theta'\|<\infty$.
Couple mass $\min(\widehat w_k,w_k^*)$ between
$\widehat\theta_k$ and $\theta_k^*$ for each $k$, and couple the remaining
mass arbitrarily.  This gives
\[
W_1(\widehat G,G^*)
\le \sum_{k=1}^K
\min(\widehat w_k,w_k^*)\|\widehat\theta_k-\theta_k^*\|
+\frac{D_\Theta}{2}\sum_{k=1}^K|\widehat w_k-w_k^*|.
\]
The result follows from the corresponding theorem because $K$ is fixed.
\end{proof}

%%%%%%%%%%%%%%%%%%%%%%%%%%%%%%%%%%%%%%%%%%%%%%%%%%%%%
\subsection{Technical lemmas}
We conclude with two standard empirical-process and Gaussian comparison tools used in the radius concentration argument.
They are stated separately to keep the main proof focused on the estimator-specific steps.

The first lemma controls empirical order statistics uniformly through the population quantile function.
It is used to convert the empirical half-sample radius into a deterministic quantile bound.
\begin{lemma}[Concentration of order statistics]
\label{lemma:concentration_of_order_statistics}
Let $X_1,\ldots,X_m$ be IID with distribution function $F$ and population quantile function $q(a)=\inf\{x:F(x)\ge a\}$. Let $\widehat F_m(x)=m^{-1}\sum_{i=1}^m\mathbbm{1}\{X_i\le x\}$ and $\widehat q_m(a)=\inf\{x:\widehat F_m(x)\ge a\}$. If $X_{(1)}\le\cdots\le X_{(m)}$ are the order statistics, then $\widehat q_m(a)=X_{(\lceil ma\rceil)}$ for $a\in(0,1]$. Moreover, for any $\varepsilon>0$, with probability at least $1-2\exp(-2m\varepsilon^2)$,
\[
q(a-\varepsilon)\le \widehat q_m(a)\le q(a+\varepsilon)
\]
uniformly over $a\in[\varepsilon,1-\varepsilon]$.
\end{lemma}

The second lemma compares shifted and unshifted Gaussian probabilities over convex symmetric sublevel sets.
It is used to lower bound off-center cost quantiles by the centered generalized chi-square quantile.
\begin{lemma}[Maximum probability of convex sublevel sets]
\label{lemma:maximum_probability_of_convex_sublevel_sets}
Let $\varphi: \sR^d \to \sR$ be any even convex function and let $t > 0$ be any threshold. Let $Z$ be a standard Gaussian random vector in $\sR^d$. Then for any vector $v \in \sR^d$, we have
\[
\sP(\varphi(Z + v) < t) \le \sP(\varphi(Z) < t).
\]
\end{lemma}

\begin{remark}
Lemma~\ref{lemma:concentration_of_order_statistics} and Lemma~\ref{lemma:maximum_probability_of_convex_sublevel_sets} were established in Appendix~B.8 of \cite{zhang2026byzantine} and are included here for completeness, as they serve as essential technical tools throughout our theoretical analysis.
\end{remark}

%% The following is a directive for TeXShop to indicate the main file
%%!TEX root = ../main.tex

\section{Simulation study: more details}

\subsection{Details about the performance metrics}
\label{app:metrics}
To evaluate the empirical performance of the estimation methods, we consider two primary metrics across the simulation runs. Let $G^{(r)}$ be the true mixing distribution in the $r$th replication, and $\widehat{G}^{(r)} = \sum_{k=1}^K \widehat{w}_k^{(r)} \delta_{\widehat{\theta}_k^{(r)}}$ represent its corresponding estimate. The performance is assessed via the following criteria:
\begin{enumerate}
    \item \textbf{Transportation distance} ($W_1$).
    
    The transportation distance between two mixing distributions is defined by
    \[
    W_1(\widehat{G},G)
    =
    \min_{\pi}
    \left\{\sum_{i,j}\pi_{ij}D_1(\widehat{\theta}_i,\theta_j)
    \,:\,
    \sum_j \pi_{ij}=\widehat{w}_i^{(r)},\;
    \sum_i \pi_{ij}=w_j
    \right\}.
    \]
    The metric $D_1(\widehat{\theta}_i,\theta_j)$ takes different forms depending on the model. 
    For the Gaussian case, we use
    $D_1(\widehat{\theta}_i,\theta_j)
    =
    \|\widehat{\mu}_i-\mu_j\|
    +
    \left\|
    \widehat{\Sigma}_i^{1/2}-\Sigma_j^{1/2}
    \right\|_F,$
    where $\Sigma^{1/2}\succeq 0$ denotes the matrix square root of $\Sigma\succeq 0$, and
    $\|\Sigma\|_F=\sqrt{\mathrm{tr}(\Sigma^\top\Sigma)}$ is the Frobenius norm. 
    \item \textbf{Adjusted Rand Index} (ARI).

    Mixture models are commonly used for model-based clustering. Given a mixing distribution $G$, an observation $x$ is assigned to a cluster by
    \[
    k(x)=\arg\max_{k}\{w_k\phi(x;\mu_k,\Sigma_k)\}.
    \]
    We assess the clustering performance of the full dataset based on $G^{(r)}$ and $\widehat{G}^{(r)}$ using the Adjusted Rand Index (ARI), which measures the agreement between two clustering assignments. Specifically, suppose one method assigns the observations into $K$ clusters
    $A_1,A_2,\ldots,A_K$, while another assigns them into $K'$ clusters
    $B_1,B_2,\ldots,B_{K'}$. Let $N_i=\mathrm{Card}(A_i)$,
    $M_j=\mathrm{Card}(B_j)$, and $N_{ij}=\mathrm{Card}(A_i\cap B_j)$ for
    $i\in [K]$ and $j\in [K']$. Then the ARI is given by
    \[
    \mathrm{ARI}
    =
    \frac{
    \sum_{i,j}\binom{N_{ij}}{2}
    - \binom{N}{2}^{-1}\sum_{i,j}\binom{N_i}{2}\binom{M_j}{2}
    }{
    \frac{1}{2}\sum_i \binom{N_i}{2}
    + \frac{1}{2}\sum_j \binom{M_j}{2}
    - \binom{N}{2}^{-1}\sum_{i,j}\binom{N_i}{2}\binom{M_j}{2}
    }.
    \]
    Here, $\binom{n}{k}$ denotes the number of ways to choose $k$ elements from a set of $n$ elements. An ARI value close to 1 indicates a high degree of agreement between the two clusterings, achieving exactly 1 when the cluster assignments match perfectly.
    
\end{enumerate}

\subsection{Additional results}

\subsubsection{Comparison under higher overlap}
\label{app.:maximum_overlap_results}
In this section, we report the results under the maximum overlap setting($\texttt{MaxOmega}=0.3$). 
Although the performance gap between our method and the Filtering method becomes small compared with the minimum overlap case, our method still remains the best-performing estimator and continues to achieve performance comparable to the oracle.
Figure~\ref{fig:boxplot_m_100_vary_N_vary_alpha_overlap_0.3} corresponds to the simulation settings in Figure~\ref{fig:boxplot_m_100_vary_N_vary_alpha}, except that $\texttt{MaxOmega}=0.3$. 
Figure~\ref{fig:boxplot_vary_m_fix_n_overlap_0.3} and Figure~\ref{fig:boxplot_vary_m_fix_N_overlap_0.3} correspond to the settings in Figure~\ref{fig:boxplot_vary_m_fix_n} and Figure~\ref{fig:boxplot_vary_m_fix_N}, respectively, with $\texttt{MaxOmega}=0.3$.

\begin{figure}[!htbp]
  \centering
  \includegraphics[width=\textwidth]{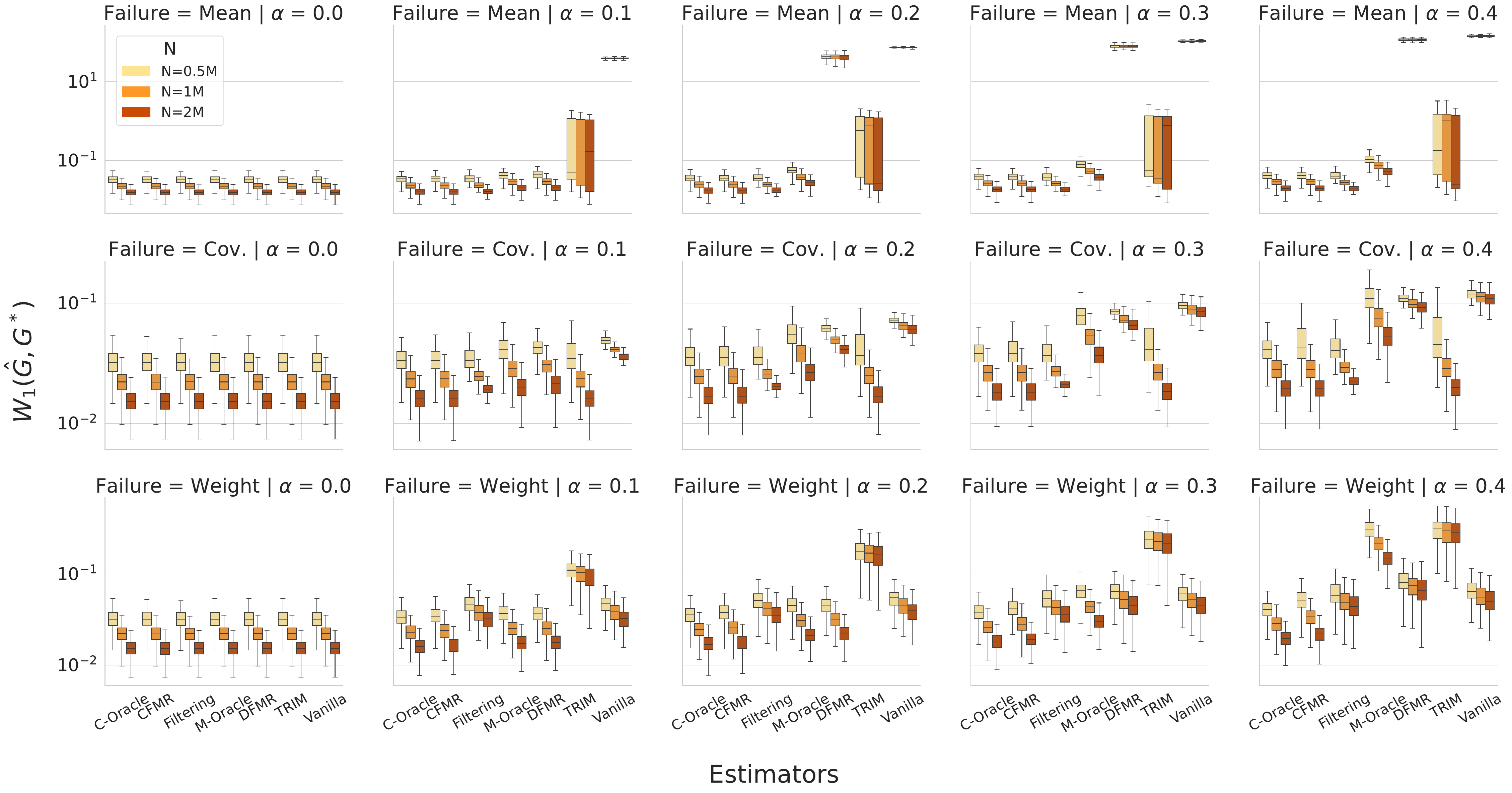}
   \caption{$W_1$ of each method as the total sample size $N$ and failure rate $\alpha$ vary, with $m=100$ and \texttt{MaxOmega}$=0.3$.}
  \label{fig:boxplot_m_100_vary_N_vary_alpha_overlap_0.3}
\end{figure}

\begin{figure}[!htbp]
  \centering
  \includegraphics[width=\textwidth]{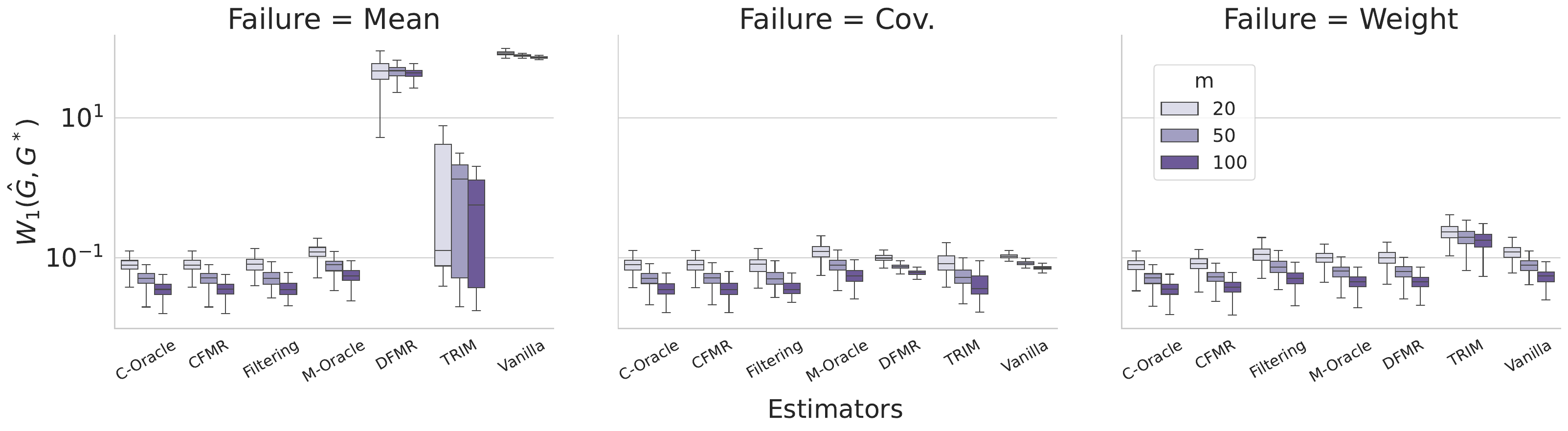}
  \caption{$W_1$ of each method as $m \in \{20,50,100\}$ varies, with $n=5000$ fixed (so $N=nm$ increases with $m$), \texttt{MaxOmega}$=0.3$, $\alpha_k=0.2$.}
  \label{fig:boxplot_vary_m_fix_n_overlap_0.3}
\end{figure}

\begin{figure}[!htbp]
  \centering
  \includegraphics[width=\textwidth]{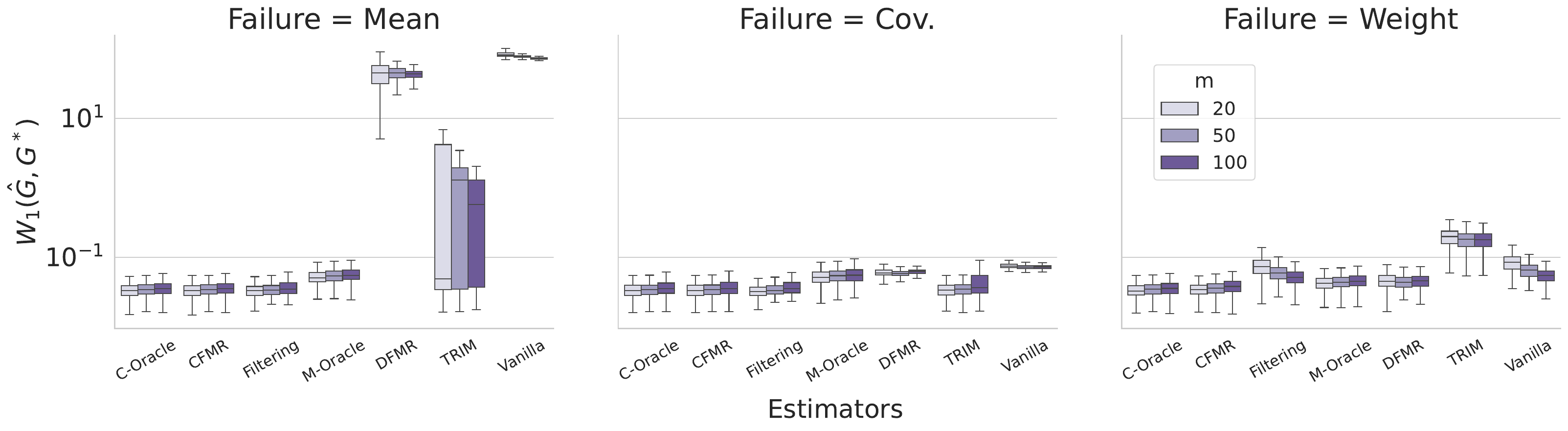}
  \caption{$W_1$ of each method as $m \in \{20,50,100\}$ varies, with $N=500{,}000$ fixed (so $n=N/m$ decreases with $m$), \texttt{MaxOmega}$=0.3$, $\alpha_k=0.2$.}
  \label{fig:boxplot_vary_m_fix_N_overlap_0.3}
\end{figure}

\subsubsection{Comparison in terms of ARI}
\label{app.:ARI_results}
In this section, we present the performance of various methods in terms of ARI. Overall, the relative performance of the estimators under ARI is consistent with that observed for $W_1$. Figure~\ref{fig:boxplot_ARI_m_100_vary_N_vary_alpha} corresponds to the settings in Figure~\ref{fig:boxplot_m_100_vary_N_vary_alpha}.
Figure~\ref{fig:boxplot_ARI_vary_m_fix_n} and Figure~\ref{fig:boxplot_ARI_vary_m_fix_N} correspond to the settings in Figure~\ref{fig:boxplot_vary_m_fix_n} and Figure~\ref{fig:boxplot_vary_m_fix_N}, respectively.
To better highlight the differences among estimators, the y-axis is restricted to $[0.85, 1.0]$.
Figure~\ref{fig:boxplot_ARI_m_100_vary_overlap_fix_alpha} corresponds to the settings in Figure~\ref{fig:boxplot_m_100_vary_overlap_fix_alpha}.
Note that under some settings, the ARI values of DFMR, TRIM, and Vanilla fall below the displayed range and are therefore omitted from the plots.

\begin{figure}[!htbp]
  \centering
  \includegraphics[width=\textwidth]{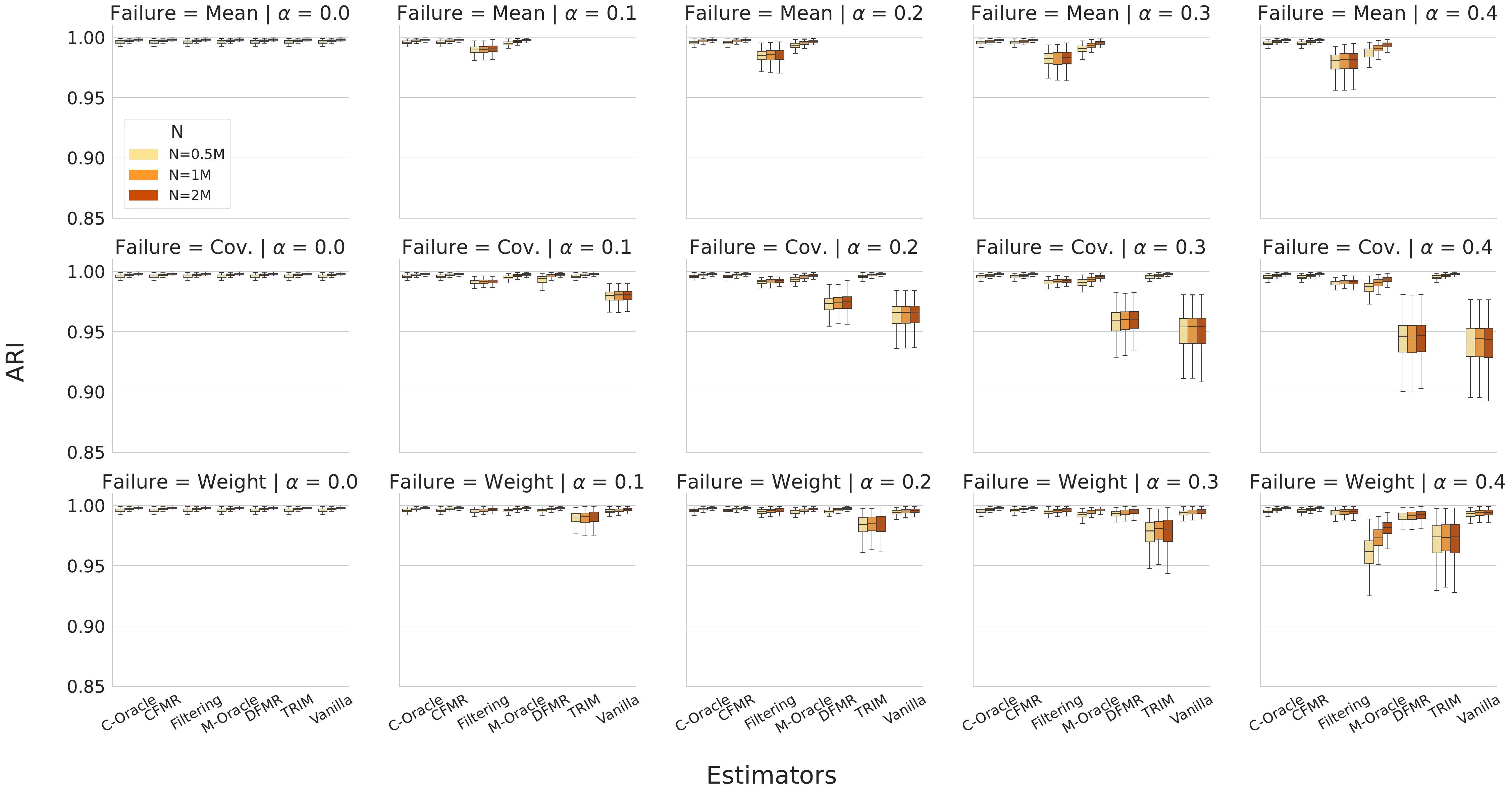}
  \caption{The ARI of clustering assignments based on different methods as the total sample size $N$ and failure rate $\alpha$ vary, with $m=100$ and \texttt{MaxOmega}$=0.1$.}
  \label{fig:boxplot_ARI_m_100_vary_N_vary_alpha}
\end{figure}

\begin{figure}[!htbp]
  \centering
  \includegraphics[width=\textwidth]{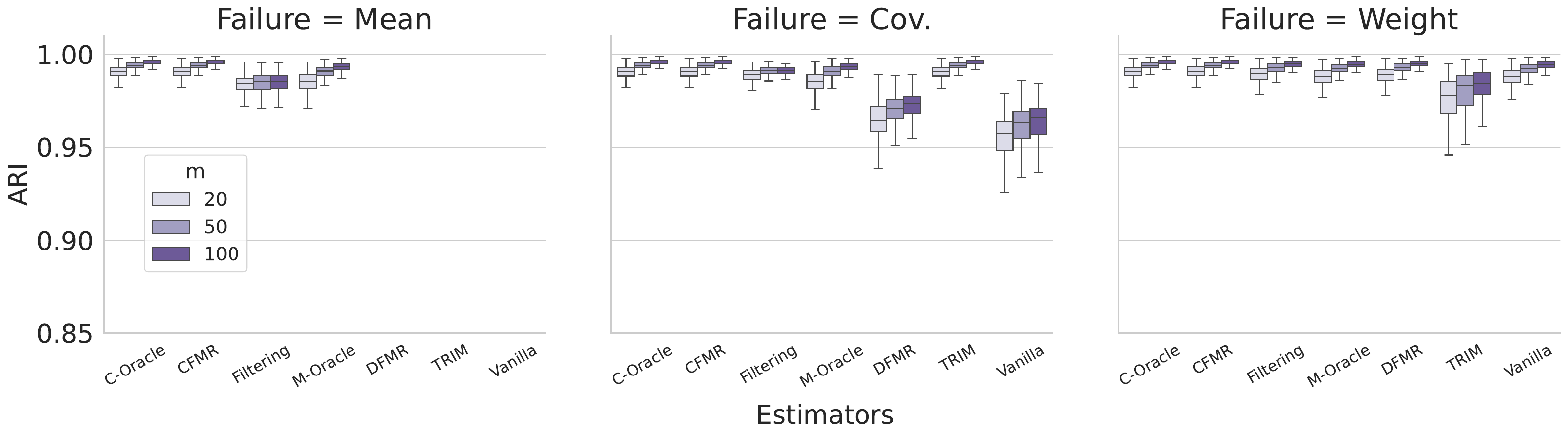}
  \caption{The ARI of clustering assignments based on different methods as $m \in \{20,50,100\}$ varies, with $n=5000$ fixed (so $N=nm$ increases with $m$), \texttt{MaxOmega}$=0.1$, $\alpha_k=0.2$.}
  \label{fig:boxplot_ARI_vary_m_fix_n}
\end{figure}

\begin{figure}[!htbp]
  \centering
  \includegraphics[width=\textwidth]{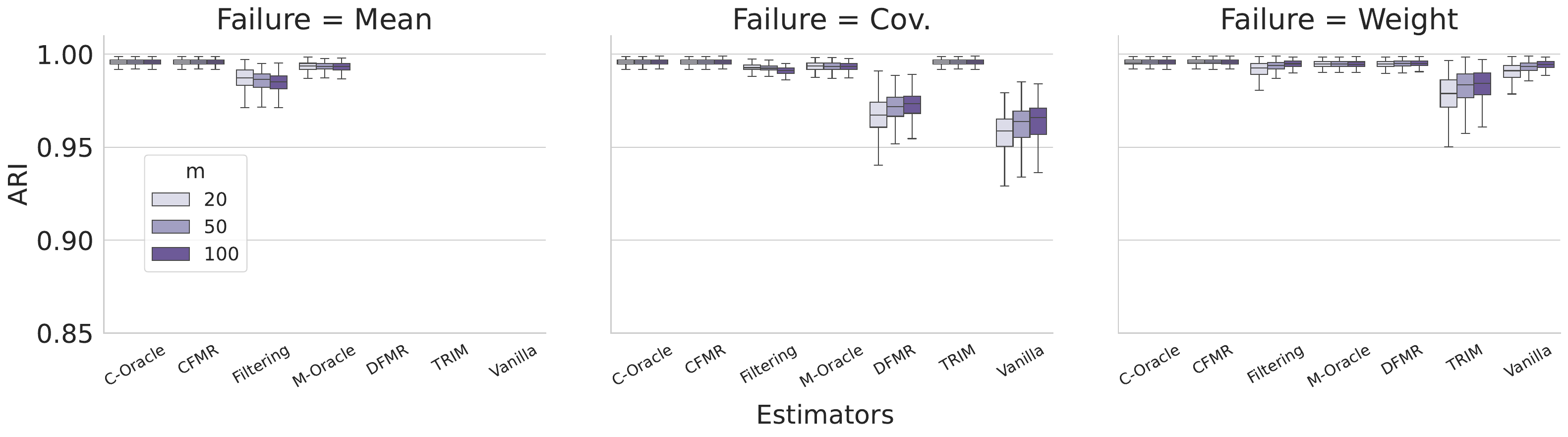}
  \caption{The ARI of clustering assignments based on different methods as $m \in \{20,50,100\}$ varies, with $N=500{,}000$ fixed (so $n=N/m$ decreases with $m$), \texttt{MaxOmega}$=0.1$, $\alpha_k=0.2$.}
  \label{fig:boxplot_ARI_vary_m_fix_N}
\end{figure}

\begin{figure}[!htbp]
  \centering
  \includegraphics[width=\textwidth]{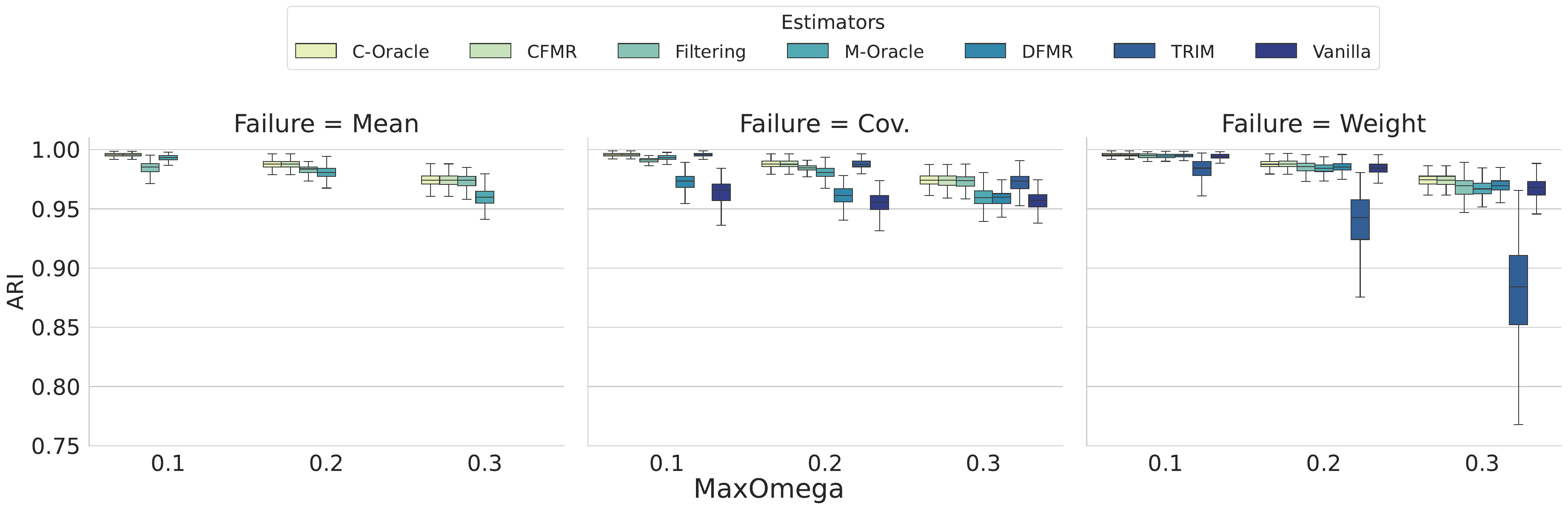}
   \caption{The ARI of clustering assignments based on different methods as \texttt{MaxOmega} $\in \{0.1, 0.2, 0.3\}$ varies, with $m=100$, $n=5000$, $\alpha_k=0.2$.}
  \label{fig:boxplot_ARI_m_100_vary_overlap_fix_alpha}
\end{figure}

\section{Real data application: more details}
\label{app:real_data}
We train a 1D-CNN to map the normalized FFT magnitudes of each vibration window to a 10-dimensional embedding.
Specifically, each input window has length \(L=512\) and yields $257$ non-negative frequency magnitudes, which are scaled by $1/512$ and standardized using the mean and standard deviation of the DE training set. 
The resulting transformation is applied identically to the DE test set.
Table~\ref{tab:cnn_architecture} summarizes the network architecture.

\begin{table}[htbp]
\centering
\caption{Architecture and layer specifications of the CNN used to extract 10-dimensional embeddings from the CWRU vibration data. $C_{\text{in}}$,  $C_{\text{out}}$, $H$, and $s$ denote the input channel size, output channel size, the convolutional kernel size, and the pooling stride, respectively.}
\label{tab:cnn_architecture}
\begin{tabular}{c|c|c}
\toprule
\textbf{Layer} & \textbf{Specification} & \textbf{Activation} \\
\midrule
Conv1d    & $C_{\text{in}} = 1$, $C_{\text{out}} = 20$, $H = 5$, $\text{padding}=2$ & ReLU    \\
MaxPool1d & $k = 2$, $s=2$ & --     \\
Conv1d    & $C_{\text{in}} = 20$, $C_{\text{out}} = 50$, $H = 5$, $\text{padding}=2$ & ReLU    \\
MaxPool1d & $k = 2$ , $s=2$   & --     \\
Flatten  &  $50\times64=3200$             &--     \\
 Linear &$H_{\mathrm{in}}=3200,\ H_{\mathrm{out}}=256$ & ReLU \\
Linear &
$H_{\mathrm{in}}=256,\ H_{\mathrm{out}}=10$
 & -- \\
Classifier & $H_{\mathrm{in}}=10,\ H_{\mathrm{out}}=9$
 & -- \\
\bottomrule
\end{tabular}
\end{table}

The final 10-dimensional output of the backbone is used as the embedding, while the classifier is used only during supervised training. The network is trained on the DE training windows using their nine fault-class labels and the cross-entropy loss. 
We use the Adam optimizer with a learning rate of \(0.001\), a batch size of \(64\), and the default Adam hyperparameters. 
Training is performed for exactly 20 epochs; the test set is evaluated after each epoch for monitoring, but no early stopping or best-epoch selection is used. 
The trained backbone is then applied to the DE training set, DE test set, and FE data to generate 10-dimensional embeddings.

\end{document}